\documentclass[A4paper,12pt]{article}
\usepackage[a4paper,total={170mm,257mm},left=20mm,top=20mm,]{geometry}
\usepackage{amsmath,amsthm,amssymb,mathtools}
\usepackage{natbib}
\usepackage[colorlinks,citecolor=blue,urlcolor=blue]{hyperref}
\usepackage{subcaption} % replaces subfigure
\usepackage{xcolor}
\usepackage{multirow}
\newcommand{\norm}[1]{\left\lVert#1\right\rVert}
\newcommand{\indep}{\perp \!\!\! \perp}

\newtheorem{theorem}{Theorem}
\newtheorem{definition}{Definition}
\newtheorem{lem}{Lemma}

\newtheorem{prop}{Proposition}

\begin{document}

{
	\title{\bf Pseudo-Observations for Bivariate Survival Data}
	\author{Yael Travis-Lumer\\
		 Department of Statistics and Data Science\\
		The Hebrew University of Jerusalem, Israel\\
\vspace{0.5cm} \\  
    Micha Mandel \\
    Department of Statistics and Data Science\\
    The Hebrew University of Jerusalem, Israel\\
\vspace{0.5cm}   \\ 
    Rebecca A. Betensky \\
    Department of Biostatistics\\
     New York University, New York, New York, U.S.A.
    }
	\date{}	
	\maketitle
}

\abstract{
The pseudo-observations approach has been gaining popularity as a method to estimate covariate effects on censored survival data. It is used regularly to estimate covariate effects on quantities such as survival probabilities, restricted mean life, cumulative incidence, and others. In this work, we propose to generalize the pseudo-observations approach to situations where a bivariate failure-time variable is observed, subject to right censoring. The idea is to first estimate the joint survival function of both failure times and then use it to define the relevant pseudo-observations. Once the pseudo-observations are calculated, they are used as the response in a generalized linear model. We consider two common nonparametric estimators of the joint survival function: the estimator of \cite{lin_simple_1993} and the Dabrowska estimator \citep{dabrowska_kaplan-meier_1988}. For both estimators, we show that our bivariate pseudo-observations approach produces regression estimates that are consistent and asymptotically normal. Our proposed method enables estimation of covariate effects on quantities such as the joint survival probability at a fixed bivariate time point, or simultaneously at several time points, and consequentially can estimate covariate-adjusted conditional survival probabilities. We demonstrate the method using simulations and an analysis of two real-world datasets.}

{\bf keywords:} 
Censoring;  Generalized estimating equations; Generalized linear models; Multi-variate survival analysis. 

\newpage
\section{Introduction}\label{sec1:intro}
The pseudo-observations approach, originally developed by \cite{andersen_generalised_2003}, has received considerable attention as a procedure for estimating covariate effects on time-to-event data. This approach is especially useful when no standard method exists, such as when interest lies in estimation of covariate effects on the survival function at a single time point, or on the restricted mean life \citep{andersen_pseudo-observations_2010}.

The pseudo-observations approach was studied extensively and was found to work well in various settings, both in theory \citep{graw_pseudo-values_2009, jacobsen_note_2016, overgaard_asymptotic_2017} and practice (see, for example, \citealt{andersen_pseudo-observations_2010}). The approach has been applied to situations where a univariate failure time is observed, including also generalizations to competing risks and multi-state models. However, as far as we know, an extension to bivariate survival data has not been studied before. \cite{furberg_bivariate_2023} studied bivariate pseudo-observations for the case where a recurrent
event process and a single failure time is observed, 
and interest lies in the joint estimation of two parameters (marginal mean and survival probability). This is different from the setting we consider in this paper, where two failure times are observed (subject to censoring), and interest lies on a single parameter.

Extending the pseudo-observations approach to bivariate failure times enables estimation of covariate effects on correlated pairs of survival times.
Our main examples are of estimating the joint survival of time to wound excision and time to wound infection of burn victims \citep{van_der_laan_locally_2002}, and time to blindness of both eyes in diabetic patients. For both examples, interest may be on the time free of complication at a certain bivariate time point, or the probability of survival of one variable conditionally on failure of the other. We demonstrate both analyses in Section~\ref{sec4:data-analysis} and in Web Appendix F.

Existing methods that model the dependence between two outcomes, and the marginal distribution of each, can be used to estimate the covariate-adjusted bivariate survival function. These include copulas (see, for example, \citealt{ marra_copula_2020}), frailty models \citep{wienke_frailty_2010}, and conditional independence models (see, for example, \citealt{xu_marginal_2023}). We refer the readers to \cite{pate_developing_2023} for a recent comparison of such methods.
These methods are not specifically designed for estimating the covariate-adjusted bivariate survival function at a fixed bivariate time point, and such covariate effects might be hard to interpret in the bivariate setting. Specifically, to estimate the covariate-adjusted bivariate survival function at a fixed bivariate time point using such methods, one must assume two univariate regression models and an additional model for the dependence structure. 

In contrast, our approach directly models the covariate effects on the bivariate survival function at a fixed number of time points. 
This direct modeling approach does not require specifying two univariate regression models, and estimates fewer parameters than existing methods. 
If our interest is on a single bivariate event, such as the 5-year survival probability of both eyes in diabetic patients, then a binary (e.g., logistic) regression model is natural.
Such a model can also be extended to several time points, and can be used for prediction of the joint survival probability, at these time points, for different covariate values. However, fitting a binary regression model to survival data requires adjustments for censoring, which is a challenging task. As we show below, the pseudo-observations approach
offers a modification that enables fitting such regression models to bivariate censored data.

\section{Methods}\label{sec2:methods}
\subsection{The setting}\label{sec2.1:setting}
Let $(T_1,T_2)$ be a bivariate random variable measuring the survival times, and let $(C_1,C_2)$ be the bivariate censoring times. For $j=1,2$, denote by $Y_j=\min(T_j,C_j)$ the minimum between the survival and censoring times, and by $\Delta_j=I(T_j \leq C_j)$ the corresponding indicator. The observed data consists of $n$ i.i.d. copies $\{(Y_{11},\Delta_{11},Y_{21}, \Delta_{21}, Z_1),\ldots,(Y_{1n},\Delta_{1n},Y_{2n}, \Delta_{2n}, Z_n)\}$, where $Z_i \in \mathbb{R}^p$ are the covariates.
Throughout this work we assume that the pair $(C_1,C_2)$ is independent of both $(T_1,T_2)$ and $Z$, which is the bivariate extension of the standard independent censoring assumption in the pseudo-observations literature. In some cases, we may have a single variable $C$ censoring both failure times (see Section~\ref{sec2.3:NPE}).
We denote by $S_{T_1,T_2}(t_1,t_2)$ the bivariate survival function of the pair $(T_1,T_2)$ at the point $(t_1,t_2)$, and by $S_{T_1}(t_1)$ and $S_{T_2}(t_2)$ the marginal survival functions of $T_1$ and $T_2$, respectively.

Let $f(\cdot,\cdot)$ be some function of the bivariate survival times $(T_1,T_2)$, and let $\theta$ be a univariate parameter of the form $\theta=E[f(T_1,T_2)]$. We list below some examples of such functions $f$, and their corresponding $\theta$'s,  that may be of interest in the bivariate survival setting.\\
1) Consider the function $f(T_1,T_2)=I(T_1>t_1^0,T_2>t_2^0)$, which corresponds to $\theta=E[I(T_1>t_1^0,T_2>t_2^0)]=P(T_1>t_1^0,T_2>t_2^0)=S_{T_1,T_2}(t_1^0,t_2^0)$.  
This is our main example throughout this paper, which corresponds to the bivariate survival at the point $(t_1^0,t_2^0)$. \\
2) When $f(T_1,T_2)=\min(T_1,T_2,\tau)$ then $\theta=E[\min(T_1,T_2,\tau)]$ is the $\tau$-restricted mean of the minimum survival time. A similar relationship can be defined for the maximal survival time.\\
Note also that univariate functions such as $f(T_1,T_2) = f(T_1)=I(T_1>t_0)$ and $f(T_1,T_2)= f(T_1)=\min(T_1,\tau)$, which correspond to the marginal quantities $S_{T_1}(t_0)$ and $\int_0^\tau S_{T_1}(t)dt$, respectively, can also be considered in our setting.  However, it is more natural to estimate such parameters using the marginal data $(Y_1,\Delta_1)$.

The bivariate pseudo-observations approach includes three main steps: (i) estimation of the joint survival function, (ii) obtaining an approximately unbiased estimator (satisfying conditions defined in Web Appendix A) of $\theta$ that is based on the estimate of the joint survival function, and (iii) calculating the corresponding pseudo-observations and fitting an appropriate regression model. We now describe each of these three steps.

\subsection{Nonparametric estimators of the bivariate survival function}\label{sec2.3:NPE}
We consider two common nonparametric estimators of the joint survival function: (1) a simple estimator due to \cite{lin_simple_1993} for the case of a single censoring variable, and (2) the popular estimator of \cite{dabrowska_kaplan-meier_1988}. Both of these estimators are approximately unbiased as they are plug-in estimators of smooth functionals (Web Appendix A). 

(1) The estimator of Lin and Ying was developed for the specific case where both failure times are censored by the same univariate censoring variable. It is given by
\begin{equation}\label{eq:LY}
	\hat{S}_{T_1,T_2}^{\text{LY}}(t_1,t_2)=\frac{\hat{P}(Y_1>t_1,Y_2>t_2)}{\hat{G}(\max(t_1,t_2))},
\end{equation}
where $\hat{P}(Y_1>t_1,Y_2>t_2)=\frac{1}{n}\sum_{i=1}^{n}I(Y_{1i}>t_1, Y_{2i}>t_2)$ and $\hat{G}(\cdot)$ is the Kaplan-Meier estimator \citep{kaplan_nonparametric_1958} of the censoring survival function.

This estimator has a simple representation, converges weakly to a zero-mean Gaussian process \citep{lin_simple_1993}, and is computationally fast to implement.
%\sout{ The estimator was developed for the specific case where both failure times are censored by the same independent univariate censoring variable. This estimator}
It can be easily extended to cases where a simple consistent estimator for the censoring survival is available. For example, to the case of two independent censoring variables, by substituting ${\hat{G}(\max(t_1,t_2))}$ in the denominator with $\hat{G_1}(t_1)\hat{G_2}(t_2)$, where $\hat{G_1}(\cdot)$ and  $\hat{G_2}(\cdot)$ are the Kaplan-Meier estimators of the survival functions of the censoring variables $C_1$ and $C_2$, respectively.

(2) The Dabrowska estimator corresponds to the more general setting defined in Section~\ref{sec2.1:setting}, of two censoring times $C_1$ and $C_2$ that may be correlated. This estimator has a product integral representation that depends on cumulative hazard functions corresponding to single and double failures. The Dabrowska estimator possesses many nice properties such as almost sure consistency and asymptotic efficiency under independent censoring, has good practical performance for even small sample sizes, and is considered the gold standard method for nonparametric estimation of the bivariate survival function \citep{prentice_nonparametric_2018}. This estimator is more complicated than that of Lin and Ying, but
can be easily applied in R using the package mhazard \citep{mhazard}. 

\subsection{Estimators of $\theta$ that are based on the estimator of the joint survival function}\label{sec2.4:univariate}
The next step is to obtain an (approximately) unbiased estimator of the parameter $\theta$ that is based on $\hat{S}_{T_1,T_2}(t_1,t_2)$.
For example, consider the restricted mean parameters $\theta_1=E\left[\min(T_1,T_2,\tau)\right]$ and $\theta_2=E\left[\min\left(\max(T_1,T_2),\tau\right)\right]$.
Then,
$\theta_1=E\left[\min(T_1,T_2,\tau)\right]=\int_{0}^{\tau}S_{T_1,T_2}(t,t)dt$, and $\theta_2=E\left[\min\left(\max(T_1,T_2),\tau\right)\right]=\int_{0}^{\tau}\left[S_{T_1,T_2}(t,0)+S_{T_1,T_2}(0,t)-S_{T_1,T_2}(t,t)\right]dt$.
Hence, an approximately unbiased estimator of either $\theta_1$ or $\theta_2$ can be obtained by plugging-in the estimator of the bivariate survival function $\hat{S}_{T_1,T_2}(t,t)$.
In the following, we consider the simplest case where $\theta=S_{T_1,T_2}(t_1^0,t_2^0)$ is the bivariate survival probability at a fixed time point $(t_1^0,t_2^0)$.

\subsection{The bivariate pseudo-observations approach}\label{sec2.2:approach}
\begin{definition}
	The $i$th pseudo-observation for $f(T_{1i},T_{2i})$ is
	\[
	\hat{\theta_i}=n\hat{\theta}-(n-1)\hat{\theta}^{-i},
	\]
	where $\hat{\theta}^{-i}$ is the leave-one-out estimator applied to the sample of size $n-1$ obtained by removing
	the $i$th observation from the total sample.
\end{definition}

As in \cite{andersen_generalised_2003}, the pseudo-observations $\hat{\theta_i}$, $(i=1,\ldots,n)$, are used as the response in a generalized linear model.
Given a link function $g$, our goal is to estimate the regression coefficients in the following model
\begin{equation}\label{eq:GLM}
	g(E[f(T_1,T_2)\mid Z])=\beta_0+\sum_{j=1}^{p}\beta_jZ_j.
\end{equation}

In our main example, where $f(T_1,T_2)=I(T_1>t_1^0,T_2>t_2^0)$ and $\theta=S_{T_1,T_2}(t_1^0,t_2^0)$, we use the logit link function defined by $g(x)=\log\left(\frac{x}{1-x}\right)$ to obtain the model
\[
\log\left(\frac{S_{T_1,T_2}(t_1^0,t_2^0\mid Z)}{1-S_{T_1,T_2}(t_1^0,t_2^0\mid Z)}\right)=\beta_0+\sum_{j=1}^{p}\beta_jZ_j.
\]
Similar logistic regression models were proposed by \cite{andersen_generalised_2003} for state probabilities in multi-state models.
The odds satisfy the relationship
\[
\text{Odds}(Z=z)=\frac{S(t_1^0,t_2^0\mid Z=z)}{1-S(t_1^0,t_2^0\mid Z=z)}=\exp\left(\sum_{j=1}^{p}\beta_jz_j\right)\frac{S^0(t_1^0,t_2^0)}{1-S^0(t_1^0,t_2^0)}\equiv \exp\left(\sum_{j=1}^{p}\beta_jz_j\right)\text{Odds}(Z=0),
\]
where $S^0(t_1^0,t_2^0)=S(t_1^0,t_2^0\mid Z=0)=\left[1+\exp(-\beta_0)\right]^{-1}$ is the baseline bivariate survival function evaluated at $(t_1^0,t_2^0)$.
This is a proportional odds model, originally proposed by \cite{bennett_log-logistic_1983} for univariate survival data, and has since been studied as an alternative to the popular Cox model, which assumes proportional hazards. When the covariate indicates a binary group assignment, the proportional odds assumption can be checked by comparing the empirical (un-adjusted) odds of the two groups over different time points. 
The same interpretation of the proportional odds model used
in the univariate case can be used in our bivariate case. For a covariate value $z$, denote by $\tau(z)=\exp\left(-\sum_{j=1}^{p}\beta_jz_j\right)$ the joint covariate effect. For two covariate values $z_1$ and $z_2$ such that $\text{Odds}(Z=z_1)>\text{Odds}(Z=z_2)$ we have that $\tau(z_1)<\tau(z_2)$ and thus $S(t_1^0,t_2^0\mid z_1)>S(t_1^0,t_2^0\mid z_2)$. Hence,  when a covariate increases the odds, it also increases the survival. 

As in \cite{andersen_generalised_2003}, it is possible to generalize the model in Equation~\eqref{eq:GLM} to $k$ time points $\{(t_1^j,t_2^j)\}$, $j=1,\ldots,k$, and to estimate the model parameters simultaneously, with joint covariate coefficients, and a different intercept for each time point. 
We define $f(T_1,T_2)=(f_{1}(T_1,T_2),\ldots, f_{k}(T_1,T_2))^T$, such that
for the specific choice of $f_{j}(T_1,T_2)=I(T_1>t_1^j,T_2>t_2^j)$, we have that $\hat{\theta_i}=(\hat{S}^i_{T_1,T_2}(t_1^1,t_2^1),\ldots, \hat{S}^i_{T_1,T_2}(t_1^k,t_2^k))^T$, where $\hat{S}^i_{T_1,T_2}(t_1^j,t_2^j)$ is the $i$th pseudo-observation at time $(t_1^j,t_2^j)$. More generally, $\hat{\theta}_{ij}$ represents $f_{j}(T_{1i},T_{2i})$ and the model is 
\[
g(E[f_{j}(T_{1i},T_{2i})\mid Z_i])=\beta^T Z_{ij}^*,
\]
where $\beta=(\beta_{01},\ldots,\beta_{0k},\beta_1,\ldots,\beta_p)^T$, and $Z_{ij}^*$ is $(0,...,1,...0,Z_i^T)^T$ where the $1$ is in the $j$'th place. 
Estimates of the ($k + p$-dimensional) $\beta$ are based on the generalized estimating equations \citep{liang_longitudinal_1986}
\begin{equation}\label{eq:GEE}
	\sum_{i=1}^{n}\left(\frac{\partial}{\partial \beta}g^{-1}(\beta^T Z_i^*)\right)^T V_i^{-1}\left(\hat{\theta_i}-g^{-1}(\beta^T Z_i^*)\right)=\sum_{i=1}^n U_i(\beta)\equiv U(\beta)=0,
\end{equation}
where $V_i\in \mathbb{R}^{k\times k}$ is a working covariance matrix, and $g^{-1}(\beta^T Z_i^*)=(g^{-1}(\beta^TZ^*_{i1}),\ldots,g^{-1}(\beta^TZ^*_{ik}))^T$.% is short for $g^{-1}(\beta^T Z_{ij}^*)$, $j=1\ldots, k$.

As can be seen, some choices that need to be made include the type of link function $g$, the number of time points $k$ and their locations, and the working covariance matrix $V_i$. It is common to choose $V_i=I_{k\times k}$, though the efficiency of the estimators can be further improved when $V_i$ is chosen to resemble the true covariance \citep{liang_longitudinal_1986}.

Denote by $\hat{\beta}$ the solution to the generalized estimating equations \eqref{eq:GEE}. A sandwich estimator is commonly used to estimate the covariance matrix of $\hat{\beta}$ \citep{andersen_generalised_2003}; Let 
\[
I(\beta)=\sum_{i=1}^{n}\left(\frac{\partial}{\partial \beta}g^{-1}(\beta^T Z_i^*)\right)^T V_i^{-1}\left(\frac{\partial}{\partial \beta}g^{-1}(\beta^T Z_i^*)\right),
\]
and let
\begin{equation}\label{eq:sigma_hat}
	\hat{\Sigma}=\sum_{i=1}^n U_i(\hat{\beta}) U_i(\hat{\beta})^T.
\end{equation}
Then
\begin{equation}\label{eq:sandwich}
	\hat{\text{var}}(\hat{\beta})=I(\hat{\beta})^{-1}\hat{\Sigma}I(\hat{\beta})^{-1}.
\end{equation}
However, recent studies \citep{jacobsen_note_2016,overgaard_asymptotic_2017} found that this ordinary sandwich estimator may, in certain cases, be biased upward. An improved consistent variance estimator, together with theoretical properties of the estimated regression parameters $\hat{\beta}$, are discussed in Section~\ref{sec:theory} below.

\subsubsection{Theoretical properties}\label{sec:theory}
Assume Equation~\eqref{eq:GLM} holds for some true, but unknown parameter vector $\beta^*=(\beta_0,\beta_1,\ldots,\beta_p)^T$.
We estimate the parameter vector $\beta^*$ by $\hat{\beta}$, where $\hat{\beta}$ is the solution to the generalized estimating equations \eqref{eq:GEE}, where for simplicity we assume a single bivariate time point ($k=1$). Below we show that for our main example, where $f(T_1,T_2)=I(T_1>t_1^0, T_2>t_2^0)$, for both the Lin and Ying estiamtor and the Dabrowska estimator, $\hat{\beta}$ is consistent and asymptotically normal.

\begin{theorem}\label{thm}
	Let $\hat{\beta}$ be the solution to the generalized estimating equations \eqref{eq:GEE},
	where the pseudo-observations $\hat{\theta}_i$ are based on either the Dabrowska estimator or the Lin and Ying estimator of the bivariate survival function. Then, under suitable regularity conditions (stated in \citealt[Theorem 3.4]{overgaard_asymptotic_2017}), $\hat{\beta}\rightarrow \beta^*$ in probability and 
	\[
	\sqrt{n}(\hat{\beta}-\beta^*) \xrightarrow{d} N(0,M^{-1}\Sigma M^{-1}),
	\]
	where $M=N(\beta^*)$, $N(\beta)=E\left[\left(\frac{\partial}{\partial \beta}g^{-1}(\beta^T Z_i)\right)^T V_i^{-1}\frac{\partial}{\partial \beta}g^{-1}(\beta^T Z_i)\right]$, and $\Sigma$ is defined in Web Appendix A. 
\end{theorem}

The proof is based on Theorem 3.1 and  Theorem 3.4 of \cite{overgaard_asymptotic_2017} ; the technical details are provided in Web Appendix A.

\cite{jacobsen_note_2016}, \cite{overgaard_asymptotic_2017}, and \cite{furberg_bivariate_2023} showed that the ordinary variance estimator $\hat{\Sigma}$ given by Equation~\eqref{eq:sigma_hat} may over-estimate the true variance in some settings, and consequently will lead to conservative confidence bands. The ordinary variance estimator $\hat{\Sigma}$ can be improved to produce a consistent estimator of the true variance $\Sigma$ (see Web Appendix A). However, the improved variance estimator is more difficult to calculate as it depends on the second order influence function, leading to a cumbersome estimation procedure. Additionally, the difference between the ordinary sandwich estimator, and the improved sandwich estimator is usually negligible \citep{jacobsen_note_2016,overgaard_asymptotic_2017}. Consequently, we suggest using the ordinary sandwich estimator given by Equation~\eqref{eq:sandwich}, for both the Lin and Ying estimator and the Dabrowska estimator. This  ordinary sandwich estimator is studied further in our simulations below.

\section{Simulations}\label{sec3:simulations}
We test our approach on simulated data where the goals are to estimate (i) the regression parameters and (ii) the covariate-adjusted bivariate survival probability.
We consider a univariate covariate and the logit link function, and test the approach using both the Dabrowska estimator and the Lin and Ying estimator of the bivariate survival function.
Additionally, we consider both univariate censoring and independent bivariate censoring, in which case we modify the Lin and Ying estimator in Equation~\eqref{eq:LY} by substituting ${\hat{G}(\max(t_1,t_2))}$ in the denominator with $\hat{G_1}(t_1)\hat{G_2}(t_2)$, as explained in Section~\ref{sec2.3:NPE}.
Analyses were implemented in R \citep{R} with the packages mhazard \citep{mhazard}, pseudo \citep{pseudo}, geepack \citep{geepack,geepack2,geepack3}, ggplot2 \citep{ggplot}, and plotly \citep{plotly}.

The data generating process of the bivariate logistic failure-times model is as follows. Let $\beta_0(t_1,t_2)=-\kappa \log(a_1t_1+a_2t_2)$, where we use $\kappa=1$, $a_1=1$, and $a_2=3$. We generate a random vector $(Z_1,\ldots,Z_n)^T$ from a uniform $U(0.5,1.5)$ distribution. For each $i=1,\ldots,n$, we define the logistic bivariate survival function by 
\begin{equation}\label{eq:logistic}
	S_{T_1,T_2}(t_1,t_2\mid Z_i)=\frac{\exp(\beta_0(t_1,t_2)+\beta_1Z_i)}{1+\exp(\beta_0(t_1,t_2)+\beta_1Z_i)},
\end{equation}
where we use the value $\beta_1=2$. Note that the definition of $\beta_0(t_1,t_2)$ guarantees that this is a valid bivariate survival function. In order to sample from this model, we first sample $T_1\sim F_{T_1}(t_1\mid Z)$ marginally using the inverse transform method, where $F_{T_1}(t_1\mid Z)=1-S(t_1,0\mid Z)$ is the conditional cumulative distribution function (CDF) of $T_1$. Then, we sample $T_2$ conditional on $T_1=t_1$ using the conditional CDF $F_{T_2|T_1=t_1}(t_2\mid Z)$.

For the univariate censoring scenario, we generate $n$ independent univariate censoring variables from an exponential distribution with rate $\lambda=0.3$, which corresponds to about $70\%$ censoring of $T_1$, and about $50\%$ censoring of $T_2$.  For the bivariate independent censoring scenario, we generate $n$ independent univariate censoring variables from an exponential distribution with rate $\lambda=0.3$, and $n$ independent univariate censoring variables from an exponential distribution with rate $\lambda=0.2$, which correspond to about $72\%$ censoring of $T_1$, and about $48\%$ censoring of $T_2$. In both cases, the observed data consists of $(Y_{i1},\Delta_{1i},Y_{21},\Delta_{2i},Z_i)$, as described in Section~\ref{sec2.1:setting}.

\subsection{Logistic failure times: a single regression model for six time points}\label{sec3.1:single_reg_logistic}
We use a sample of size $n=200$ and consider the six time points $$\{(0.5,0.7),(1,0.7),(0.5,1.2),(1,1.2),(0.5,1.5),(1,1.5)\}$$ and the functions $f_j(T_1,T_2)=I(T_1>t_1^j, T_2>t_2^j)$, $j=1,\ldots, 6$, such that $f(T_1,T_2)=\left(f_1(T_1,T_2),\ldots, f_6(T_1,T_2)\right)$. Note that these six time points were chosen randomly and before observing the simulated data. We then verified that these time points were not too extreme by checking that the empirical bivariate CDF of the observed times at these six time points was not too close to either zero or one, where the bivariate empirical CDF of the observed times is defined by $\hat{F}_{Y_1,Y_2}(t_1,t_2)=\frac{1}{n}\sum_{i=1}^{n}I(Y_{1i}\leq t_1,Y_{2i}\leq t_2)$. Specifically, we found that all six time points were between the 10 and 30 empirical percentiles of the observed bivariate times; that is $\hat{F}_{Y_1,Y_2}(t_1^j,t_2^j)\in (0.10,0.30)$ for $1\leq j\leq 6$.
The corresponding regression model
\begin{equation}\label{eq:single_reg}
	g\left(S(t_1^j,t_2^j\mid Z)\right)=\beta_0^j+\beta_1Z, \ j=1,\ldots, 6,
\end{equation}
where $g$ is the logit link function.
This regression model assumes a different intercept $\beta_0^j$ for each time point, and a single slope parameter $\beta_1$ (which is equal for all time points). In Web Appendix B we repeat the analysis using six different regression models, with a different intercept $\beta_0^j$ and a different slope $\beta_1^j$ for each time point, where $j=1,\ldots,6$.

For the response, we use the bivariate pseudo-observations that are based on estimates of the bivariate survival function using either (1) the Lin and Ying estimator or (2) the Dabrowska estimator.

Table~\ref{table:logistic_new} presents the true values and the estimated values of the regression coefficients based on $m=500$ replications (including standard deviations, standard errors, and coverage), for both the univariate censoring mechanism (top), and the bivariate censoring mechanism (bottom). The standard deviation is based on the empirical variance of the estimated regression parameters.
The standard error is the square root of the average of the ordinary sandwich estimators of the variance, each given by Equation~\eqref{eq:sandwich}; that is,  $\text{se}=\sqrt{\frac{1}{m}\sum_{i=1}^m\hat{\text{var}}(\hat{\beta}_i)}$. The coverage is the proportion of times that the true regression parameter is included in the confidence interval.

For the univariate censoring scenario, the bivariate pseudo-observations approach manages to correctly estimate the true value of the regression parameters, with a negligible insignificant bias. Interestingly, for the univariate censoring scenario, the estimates of the regression parameters are very close, for both the Dabrowska estimator and the Lin and Ying estimator. However, the standard deviation and standard error of the estimates based on the Dabrowska estimator are lower than those based on the Lin and Ying estimator. Additionally, for each estimator separately, the standard deviations are larger than the standard errors. Note also that the coverage, which is based on the standard errors, is very close to 95\%, for both the Dabrowska estimator and for the Lin and Ying estimator. However, there are two instances where the coverage falls slightly below 95\% for the Lin and Ying estimator. 

For the bivariate censoring scenario, the Dabrowska estimator seems to perform better than the modified estimator of Lin and Ying, both in terms of lower bias and lower variance. The coverage is generally slightly lower than 95\%. This phenomenon is more pronounced for the Lin and Ying estimator.

\begin{table}[p]
	\begin{subtable}{1\textwidth}
		\centering
		\subcaption{Univariate censoring}
		\begin{tabular}{rrrrrrrrrr}
			\hline
			\multicolumn{2}{r}{} & 	\multicolumn{4}{c}{Dabrowska} & \multicolumn{4}{c}{Lin and Ying} \\	
			\hline
			Parameter & True value	& $\hat{\beta}$ & sd & se & cov (\%) & $\hat{\beta}$ & sd & se & cov (\%) \\ 
			\hline
			\hline
			$\beta_0(0.5,0.7)$ & -0.96 & -0.99 & 0.63 & 0.60 & 95.60 & -0.99 & 0.70 & 0.65 & 95.20  \\ 
			$\beta_0(1,0.7)$ & -1.13  & -1.16 & 0.63 & 0.60 & 95.00 & -1.16 & 0.70 & 0.66 & 95.80  \\ 
			$\beta_0(0.5,1.2)$ & -1.41 & -1.45 & 0.64 & 0.60 & 95.00 & -1.45 & 0.71 & 0.66 & 94.80 \\ 
			$\beta_0(1,1.2)$ & -1.53 & -1.56 & 0.64 & 0.61 & 95.80 & -1.56 & 0.71 & 0.67 & 95.40  \\ 
			$\beta_0(0.5,1.5)$ & -1.61 & -1.64 & 0.65 & 0.61 & 95.40 & -1.64 & 0.71 & 0.67 & 95.40  \\ 
			$\beta_0(1,1.5)$ & -1.70 & -1.74 & 0.65 & 0.61 & 95.60 & -1.74 & 0.71 & 0.67 & 95.20  \\ 
			$\beta_1$ & 2.00 & 2.05 & 0.64 & 0.60 & 95.80 & 2.05 & 0.71 & 0.66 & 94.60\\ 
			\hline
		\end{tabular}
	\end{subtable}
	\hfill
	\hfill\\
	\hfill\\
	\begin{subtable}{1\textwidth}
		\centering
		\subcaption{Independent bivariate censoring and a modification of Lin and Ying}
		\begin{tabular}{lrrrrrrrrr}
			\hline
			\multicolumn{2}{r}{} & 	\multicolumn{4}{c}{Dabrowska} & \multicolumn{4}{c}{Lin and Ying} \\	
			\hline
			Parameter & True value	& $\hat{\beta}$ & sd & se & cov (\%) & $\hat{\beta}$ & sd & se & cov (\%) \\ 
			\hline
			$\beta_0(0.5,0.7)$ & -0.96 & -0.98 & 0.61 & 0.58 & 94.20 & -1.00 & 0.71 & 0.67 & 94.20  \\ 
			$\beta_0(1,0.7)$ & -1.13 & -1.16 & 0.61 & 0.58 & 94.20 & -1.17 & 0.72 & 0.68 & 94.80  \\ 
			$\beta_0(0.5,1.2)$ & -1.41 & -1.44 & 0.61 & 0.59 & 95.00 & -1.46 & 0.72 & 0.69 & 94.80  \\ 
			$\beta_0(1,1.2)$ & -1.53 & -1.56 & 0.61 & 0.59 & 94.80 & -1.58 & 0.73 & 0.69 & 94.20  \\ 
			$\beta_0(0.5,1.5)$ & -1.61 & -1.64 & 0.61 & 0.59 & 95.00 & -1.67 & 0.72 & 0.69 & 94.80  \\ 
			$\beta_0(1,1.5)$ & -1.70 & -1.74 & 0.61 & 0.59 & 94.40 & -1.76 & 0.73 & 0.70 & 94.60  \\ 
			$\beta_1$ & 2.00 & 2.04 & 0.61 & 0.58 & 95.20 & 2.07 & 0.73 & 0.69 & 94.80 \\ 
			\hline
		\end{tabular}
	\end{subtable}
	\hfill\\
	\caption{Top: univariate censoring. Bottom: independent bivariate censoring. Mean, sd, se, and coverage of the estimated regression parameters $\beta_0(t_1,t_2)$ and $\beta_1$, for both the Dabrowska estimator, and the Lin and Ying estimator. These estimates are based on a single regression model for the six time points, and on 500 simulations from the bivariate logistic model. The estimate $\hat{\beta}$ is the average of 500 parameter estimates; sd, square root of empirical variance of 500 simulation replications; se, square root of average estimated variance based on \eqref{eq:sandwich}.
	}\label{table:logistic_new}
\end{table}

Additionally, for each time point, iteration, and type of estimator, we estimated the joint survival probability $\hat{S}(t_1^j,t_2^j\mid Z_i)$ by plugging the estimates in Equation~\eqref{eq:logistic}, and compared it to the true survival probability. We then averaged the absolute difference between $\hat{S}(t_1^j,t_2^j\mid Z)$ and the true $S(t_1^j,t_2^j\mid Z)$, over all $Z_i$, $i=1,\ldots,n$. That is, for each time point, estimator and iteration, we report the mean absolute error (MAE) defined by $${\rm MAE}=\frac{1}{n}\sum_{i=1}^{n}|\hat{S}(t_1^j,t_2^j\mid Z_i)-S(t_1^j,t_2^j\mid Z_i)|.$$ Additionally, we present the standardized MAEs, which are defined by $${\rm standardized \; MAE}=\frac{1}{n}\sum_{i=1}^{n}\frac{|\hat{S}(t_1^j,t_2^j\mid Z_i)-S(t_1^j,t_2^j\mid Z_i)|}{S(t_1^j,t_2^j\mid Z_i)}.$$

Web Figure~1 and Web Figure~2 present the boxplot of the MAEs (top) and the standardized
MAEs (bottom) for the univariate censoring scenario and the bivariate censoring
scenario, respectively. For both censoring mechanisms, the MAEs show that
the bivariate pseudo-observations approach estimates quite well the covariate-adjusted joint
survival probability, with a slight preference to the Dabrowska estimator over the Lin and
Ying estimator. Finally, Web Figure~3 presents the boxplot of the MAEs as a function of the
sample size $n$. As expected, the MAEs decrease as the sample size increases. Specifically, for a sample size of $n=200$, the majority of MAEs are below 0.05, whereas for a sample size of $n=800$, the majority of MAEs are below 0.025.

\subsection{The choice of time points}\label{sec:time_points}
In Section~\ref{sec3.1:single_reg_logistic} we pre-specified six time points that were selected before observing the data. This is known as the fixed time point selection approach, and was found to work well in the pseudo-observations literature (see, for example, \citealt{furberg_bivariate_2023}). Other common approaches include placing the time points at percentiles of the observed times, or at equidistant values over the range of observed times \citep{andersen_generalised_2003}; these different approaches usually result in similar unbiased regression estimates \citep{furberg_bivariate_2023}. For bivariate data, it is not clear how to uniquely define such time point selection approaches. We suggest the following procedure that adapts the idea of equidistant points to the bivariate case.
For each coordinate separately, calculate the 20th and 90th empirical percentiles of the observed time points.
Pair the two lower and upper percentiles to obtain the bivariate time points $(t_1^{20\%},t_2^{20\%})$ and $(t_1^{90\%},t_2^{90\%})$, and calculate the bivariate empirical CDF, $\hat{F}_{Y_1,Y_2}(t_1,t_2)$, at these two bivariate time points.
Note that the value of the bivariate empirical CDF at the pair of the low percentiles $(t_1^{20\%},t_2^{20\%})$ can be quite small. If it is lower than 0.1, replace the 20th percentile with the 25th or 30th  percentile, and choose the value for which the bivariate empirical CDF first reaches the value 0.1. This value will define the two lower limits. Then, take $k$ equidistant time points between the lower and upper limits, for each coordinate separately. Finally, pair these points together to obtain the selected bivariate time points.  We study the performance of this approach in Web Appendix C.

In summary, we found that the number of time points and their selection method has hardly any effect in the univariate censoring case, for both the Dabrowska and the Lin and Ying estimators (Web Table~2). Additionally, the Dabrowska estimator outperforms the Lin and Ying estimator, both in terms of lower bias and lower variance. In the bivariate censoring case, the Dabrowska estimator performs equally well for all time point selection methods. However, the modified Lin and Ying estimator is more sensitive to the choice of time points. Nevertheless, in the bivariate censoring scenario with a large enough sample size, the choice of time points has little effect on the performance of the method with either estimator, and the resulting regression estimates are quite good (Web Table~2). Additionally, our simulations suggest that the modified Lin and Ying estimator works well when the amount of censoring is relatively low (Web Table~3). Finally, in our simulated example, we found that lowering the upper percentile, which defines the upper limit of our suggested time point selection procedure, improves estimation for the Lin and Ying estimator, both in terms of lower bias and lower variance. However, it has hardly any effect on the results for the Dabrowska estimator (Web Table~3). To sum up, we recommend using the Dabrowska estimator over the Lin and Ying estimator.

\subsection{Model misspecification}
To test for the effect of model misspecification, we consider also data that are generated
from a log-normal distribution in which case the logistic regression model is misspecified.
Web Appendix D and Web Figures 14-17 present the data generating mechanism and the
MAEs for the log-normal bivariate data. In general, the MAEs show that the
pseudo-observations approach manages to provide good predictions of the covariate-adjusted
bivariate survival function, under this model misspecification.

\section{Data analysis}\label{sec4:data-analysis}
We demonstrate the method on two data sets, the diabetic retinopathy study that is discussed below, and the burn victim data that is presented in Web Appendix F.
The diabetic retinopathy study describes the time to blindness of 197 patients with diabetic retinopathy, where one eye of each patient is randomly assigned to laser treatment, and time is measured in months from initiation of treatment. This dataset was first described by \cite{huster_modelling_1989}, and is available from the survival R package \citep{survival-book,survival-package}. 101 out of the 197 untreated eyes suffered from vision loss during the follow-up period, compared to only 54 vision loss cases out of the 197 treated eyes. All survival times include a built-in lag time of approximately 6 months and the maximal follow-up time was 75 months. Censoring was caused by either death, dropout, or end of the study; thus, the setting is of a single censoring variable for both survival times. The covariates include the age at diagnosis of diabetes, the type of diabetes (juvenile or adult) and a risk score of each eye (a value between 6 and 12). As the risk scores of both eyes are usually identical or very similar, we define a patient-specific risk score by averaging the risk scores of both eyes. Age at diagnosis of diabetes may be correlated with death (which is one of the possible causes of censoring); however, since the largest follow-up time is 75 months, and since all patients in this dataset were diagnosed with diabetes before the age of 58, we assume that the correlation between age at diagnosis of diabetes and death during the course of the study is negligible, and that the assumption of independent censoring holds. 

We denote by $T_1$ the time to blindness (in months) of the treated eye, and by $T_2$ the time to blindness of the untreated eye. Our first goal here is to use the proportional odds model to estimate the five year survival probability of both eyes given the covariates age at diagnosis of diabetes, type of diabetes, and the patient-specific risk score.
As a preliminary analysis, we first compared the empirical un-adjusted odds for juvenile versus adult-onset diabetes, for all bivariate time points on the diagonal $(t, t)$.
Web Figure~18 shows that the vertical displacement between the two curves is reasonably
close to being constant for all $t>30$ (where $t$ denotes the number of months since the
laser treatment), thus providing some confirmation for the proportional odds assumption, at least for some bivariate time points. 
We use our bivariate pseudo-observations approach with the Dabrowska estimator and the logit link function to estimate the joint probability that both eyes survive more than five years given the covariates. Figure~\ref{fig:joint_retinopathy} presents the estimates of the probability that both eyes survive more than five years, as a function of both age at diagnosis and risk score. As can be seen, survival decreases as age at diagnosis and risk score increase.  

\begin{figure}[p]
	\centering
	\includegraphics[width=0.95\linewidth]{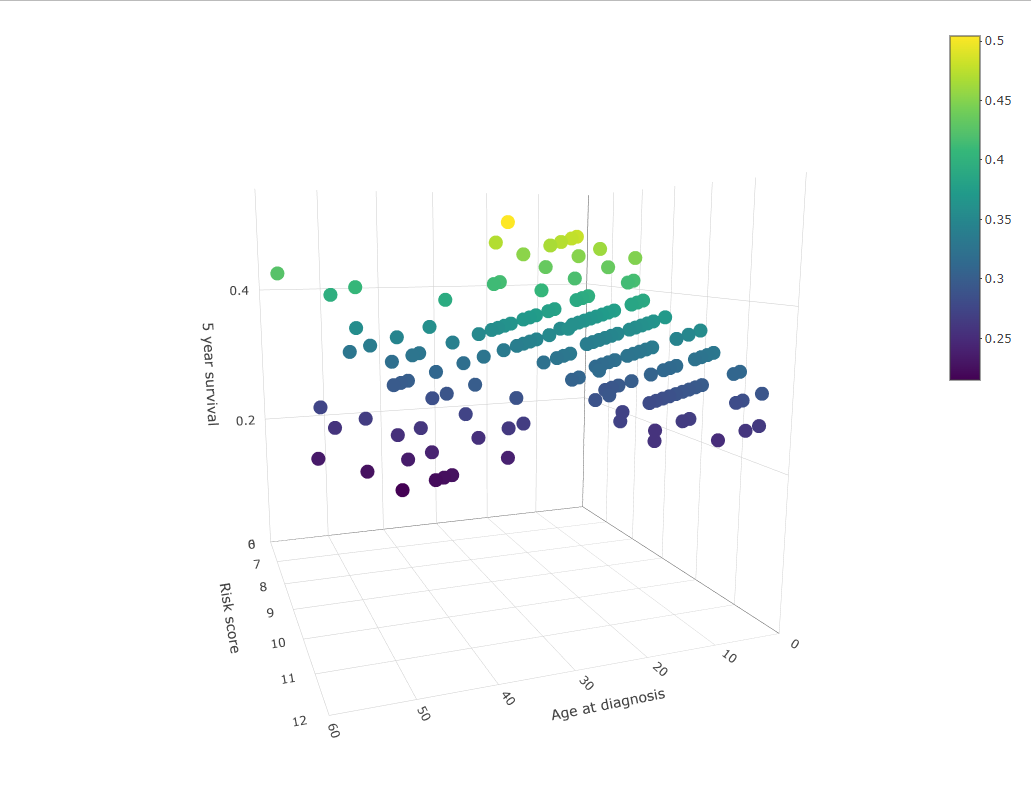}
	\caption{Diabetic retinopathy study. Covariate-adjusted joint survival probability $S_{T_1,T_2}(60,60\mid Z)=P(T_1>60, T_2 > 60 \mid Z)$. Estimates of the 5 year joint survival probabilities as a function of age at diagnosis (x axis) and risk score (y axis). The estimates are based on the Dabrowska estimator. The z axis is the value of the 5-year (60-month) survival probability.}
	\label{fig:joint_retinopathy}
\end{figure}

In addition to studying the bivariate survival function, we consider the effect of the covariates~$Z$ on the conditional survival probabilities $S_{T_1| T_2\leq 36}(60\mid Z)=P(T_1>60\mid T_2 \leq 36, Z)$ and $S_{T_1| T_2>36}(60\mid Z)=P(T_1>60\mid T_2 >36, Z)$. That is, we are interested in the conditional probability that the treated eye survives more than five years, given that the untreated eye has either failed or survived during the first three years, given the covariates. For this purpose, we use our bivariate pseudo-observations approach based on the Dabrowska estimator, and the logit link function, to estimate the covariate-adjusted bivariate survival probabilities at the three time points $\{(60,36), (60,0), (0,36)\}$
Consequently, we obtain estimates of the two conditional probabilities $S_{T_1| T_2\leq 36}(60\mid Z)$ and $S_{T_1| T_2>36}(60\mid Z)$ using:
\begin{align*}
	S_{T_1| T_2\leq 36}(60\mid Z)&=P(T_1>60 \mid T_2\leq 36, Z)=\frac{S(60,0 \mid Z)-S(60,36 \mid Z)}{1-S(0,36 \mid Z)}, \\
	S_{T_1| T_2>36}(60\mid Z)&=P(T_1>60 \mid T_2> 36, Z)=\frac{S(60,36 \mid Z)}{S(0,36 \mid Z)}.
\end{align*}
Confidence intervals for such conditional probabilities can be obtained from Theorem~\ref{thm} together with the delta method (Web Appendix G). 
Web Table~4 presents the estimated regression coefficients of the regression model based on either a single time point, or three time points. Web Table~4 shows that the estimates of all three regression parameters are negative, meaning that an increase in these covariate values decreases the bivariate survival probability. Additionally, the estimates based on either a single time point or three time points are very close, providing some additional confirmation for the correctness of the proportional odds model. Finally, the estimated regression parameter of juvenile diabetes is negative, indicating that when the risk score is fixed and when age at diagnosis is close to 20 years (the threshold for diabetes type), the bivariate survival for adult onset diabetes is larger than the survival for juvenile diabetes.

Figure~\ref{fig:retinopathy} presents the estimates of the conditional probabilities $S_{T_1| T_2\leq 36}(60\mid Z)$ (top) and $S_{T_1| T_2>36}(60\mid Z)$ (bottom), as a function of age at diabetes diagnosis and risk score. When the untreated eye fails within the first three years, the probability of the treated eye surviving beyond five years ranges from 0.50 to 0.77. This probability is highest when both age at diagnosis and risk score are low, and gradually decreases as these covariates increase. Conversely, if the untreated eye survives beyond three years, the treated eye has a substantially higher chance of surviving for more than five years, with probabilities ranging from 0.73 to 0.84. While this conditional survival probability decreases slightly as the risk score increases, it remains high even as the age at diagnosis increases, showing minimal sensitivity to this covariate. In summary, a comparison of the two conditional probabilities reveals that the effectiveness of laser treatment also depends on the health of the untreated eye. If the untreated eye fails within the first three years, the 5-year survival probability of the treated eye varies based on covariate values. However, if the untreated eye survives beyond three years, the 5-year survival probability for the treated eye remains high and is less influenced by these covariates.
Finally, note that additional unmeasured confounders may exist and result in biased effect estimates. An important example of such a confounder is a variable measuring the time from diagnosis to the time of treatment, which unfortunately is not available in this dataset.

\begin{figure}[p]
	\centering
	\begin{subfigure}[b]{0.9\textwidth}
		\centering
		\includegraphics[width=0.95\linewidth]{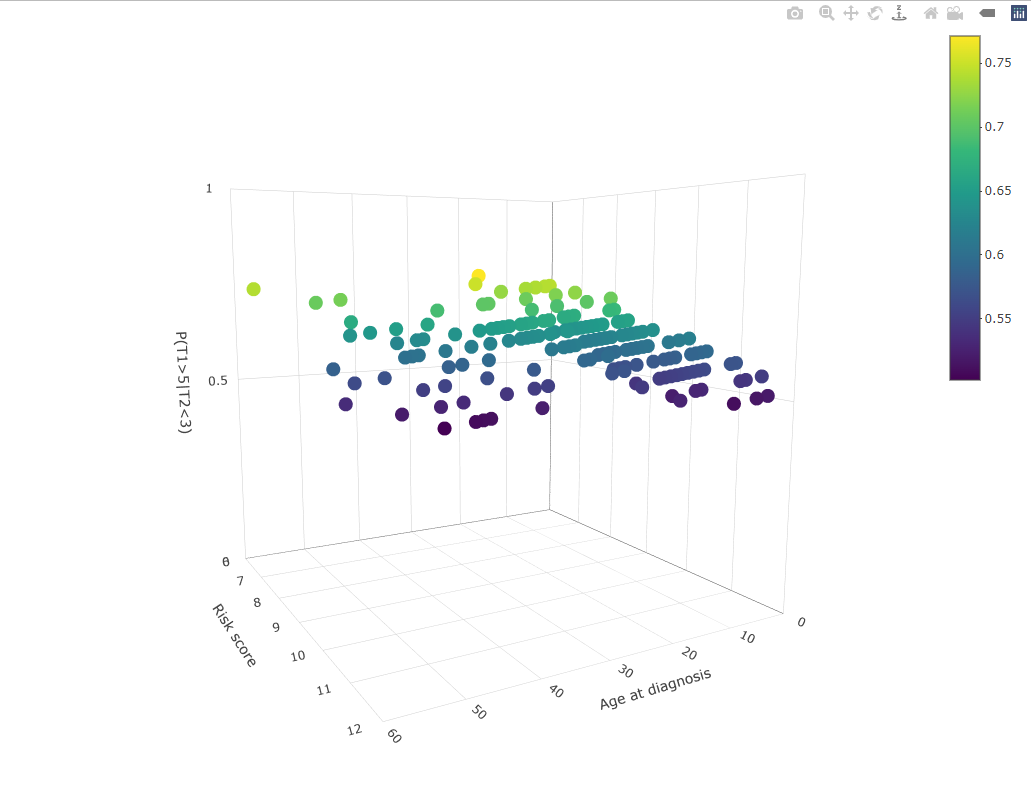}
		\caption{Conditional survival probability $S_{T_1| T_2\leq 36}(60\mid Z)=P(T_1>60\mid T_2 \leq 36, Z)$} \label{fig:retinopathy_1}
	\end{subfigure}
	\hfill
	\begin{subfigure}[b]{0.8\textwidth}
		\centering
		\includegraphics[width=0.95\linewidth]{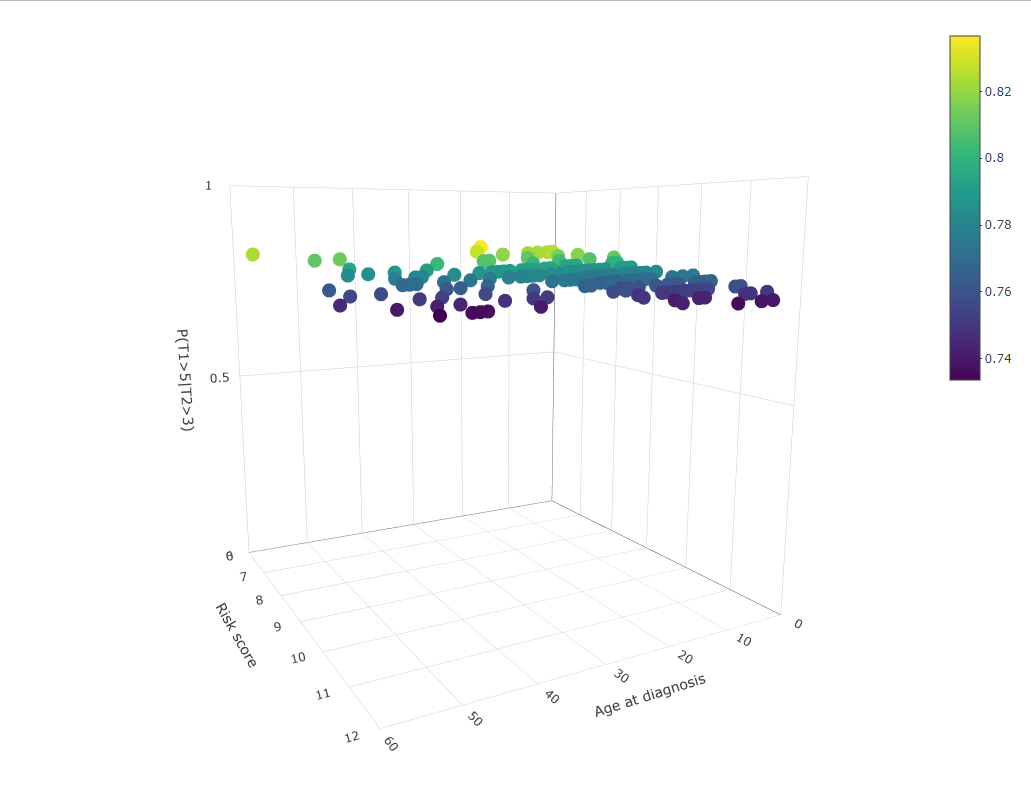}
		\caption{Conditional survival probability $S_{T_1| T_2> 36}(60\mid Z)=P(T_1>60\mid T_2 > 36, Z)$} \label{fig:retinopathy_2}
	\end{subfigure}
	
	\caption{Diabetic retinopathy study. Estimates of the two conditional survival probabilities as a function of age at diagnosis (x axis) and risk score (y axis). The estimates are based on the bivariate pseudo-observations approach using the Dabrowska estimator. The z axis is the value of the 5-year (60-month) survival probability.}
	\label{fig:retinopathy}
\end{figure}

\section{Discussion}\label{sec5:discussion}
We proposed an extension of the pseudo-observations approach for estimation of covariate effects on bivariate survival data with two aims: (i) estimation of regression parameters and (ii)  estimation of the covariate-adjusted bivariate survival function for the purpose of prediction.
We study two non-parametric estimators of the survival function and recommend using the Dabrowska estimator and the time point selection procedure presented in Section~\ref{sec:time_points}. Finally, we found that the ordinary sandwich estimator usually obtains correct coverage, for both estimators and for both censoring scenarios. 

Additionally, we proved that 
the estimated regression coefficient is both consistent and asymptotically normal under the assumption of independent censoring.
Interestingly, under this assumption, \cite{gill_inefficient_1995} showed that the influence function of the Dabrowska estimator coincides with the efficient influence function. That is, in the setting considered in this paper, the Dabrowska estimator is efficient.

As a referee noted, the bivariate extension of the proportional odds model, which is used as our main example throughout the paper, only holds under specific parametric assumptions. Consequently, its
approximations of the covariate-adjusted bivariate survival probability for different data-generating mechanisms is limited. However, for a single bivariate time point, such a model seems natural, and the proportional odds assumption only needs to hold for that specific time point. Additionally, when interest lies in estimating the covariate-adjusted bivariate survival probability at $k$ different bivariate time points, one can use either a single regression model for simultaneous estimation of $p+k$ parameters ($p$ time-fixed regression estimates and $k$ different intercepts), or fit $k$ separate models, estimating $k+pk$ parameters (a unique intercept and $p$-dimensional slope for each time point; see Web Appendix B). If the slopes are similar across the $k$ models, the proportional odds assumption is supported, allowing a single model for all $k$ probabilities.
Furthermore, in our simulations, we found that a single proportional odds model predicted quite well the true covariate-adjusted bivariate survival probabilities simultaneously at a number of fixed time points, even under model misspecification (Section 3.3 and Web Appendix D).
Finally,
one can use other link functions, such as the probit link or the cloglog link as alternatives to the proportional odds model. Comparing predictions from different GLMs can assess robustness to model misspecification.

In this work, we were mainly interested in the effects of covariates on the bivariate survival probability. However, as described in Sections~\ref{sec2.1:setting} and \ref{sec2.4:univariate},  other parameters of interest can also be considered in a similar manner. An interesting direction for future research is using our bivariate pseudo-observations approach for simultaneous estimation of covariate effects on several bivariate parameters of interest. We intend to study this direction in future work.

The current study offers an extension of the pseudo-observations approach to bivariate survival data. We note that generalization to trivariate, or, more generally, multi-variate survival data is also possible. This can be achieved by first estimating the joint (multi-variate) survival function (using, for example, the multi-variate Kaplan-Meier estimator of \citealt{prentice_nonparametric_2018}), and then using it to define the relevant pseudo-observations.  
\section*{Acknowledgments}

This research was partially supported by the ISRAEL SCIENCE FOUNDATION (grant No. 1147/20). Y. T.-L. was supported by the Drs. Eva \& Shelby Kashket Memorial Fellowship.

\section*{Supplementary Materials}
The code for the simulations and for the data-analysis is available at https://github.com/Yael-Travis-Lumer/Bivariate-pseudo-observations.

\newpage
\bibliographystyle{plainnat}
\bibliography{bivariatePObib}

\appendix
\setcounter{table}{0}
\setcounter{figure}{0}
\setcounter{equation}{0}
\renewcommand{\theequation}{S\arabic{equation}}
\renewcommand{\figurename}{Web Figure}
\renewcommand{\tablename}{Web Table}
\renewcommand{\thesection}{Web Appendix \Alph{section}}
\renewcommand{\thesubsection}{A.\arabic{subsection}}
\section{}
In this section our goal is to prove Theorem~\ref{thm}. Note that the proof of Theorem~1 is based on showing that the conditions stated in Theorem 3.1 and  Theorem 3.4 of \cite{overgaard_asymptotic_2017} are met. Before presenting the proof, we present below some key ideas, and express the pseudo-observations in terms of estimating functionals and influence functions. We also define more formally the conditions that an approximately unbiased estimator needs to satisfy.

Specifically in our context, an estimating functional of an estimator of $\theta=S_{T_1,T_2}(t_1^0,t_2^0)$ is a function $\phi:D\mapsto \Theta$, where $D$ is some Banach space consisting of vectors of functions that includes also a general sample average denoted here by $F_n$ and its limit $F$, and $\Theta \subset \mathbb{R}$ contains the parameter $\theta=S_{T_1,T_2}(t_1^0,t_2^0)$. The estimating functional is such that $\phi(F)=\theta$ returns the true parameter (which in our case is the joint bivariate survival $S_{T_1,T_2}(t_1^0,t_2^0)$), and $\phi(F_n)=\hat{\theta}$ returns the estimator $\hat{S}_{T_1,T_2}(t_1^0,t_2^0)$.

In general, similarly to the Taylor expansion of functions, smooth functionals can be von-Mises expanded in a similar manner, using Fréchet derivatives. It can be shown (see, for example, \citealt{graw_pseudo-values_2009,jacobsen_note_2016,overgaard_asymptotic_2017}) that the second-order von-Mises expansion of a smooth functional $\psi(F_n)$ around $F$ is
\begin{equation}\label{eq:von-Mises}
	\psi(F_n)= \psi(F)+\psi_F'(F_n-F)+\frac{1}{2} \psi_F''(F_n-F,F_n-F)+R_n,
\end{equation}
where $\psi_F'(F_n-F)$ is the Fréchet derivative of $\psi$ at $F$ in the direction of $F_n-F$, $\psi_F''(F_n-F,F_n-F)$ is the second-order Fréchet derivative, and where $R_n$ is the remainder. The (first order) influence function is defined by $\dot{\psi}(X)=\psi_F'(\delta_X-F)$, where $\delta_X$ is usually a vector of indicators such that $F_n=\frac{1}{n}\sum_{i=1}^{n}\delta_{X_i}$. The second-order influence function is defined by $\ddot{\psi}(X_1,X_2)=\psi_F''(\delta_{X_1}-F,\delta_{X_2}-F)$. By properties of these influence functions (see, for example, \citealt{vaart_asymptotic_1998, overgaard_asymptotic_2017}), Equation~\eqref{eq:von-Mises} can be written as
\begin{equation*}
	\psi(F_n)=\psi(F)+\frac{1}{n}\sum_{i=1}^{n}\dot{\psi}(X_i)+\frac{1}{2n^2}\sum_i\sum_j \ddot{\psi}(X_i,X_j)+R_n.
\end{equation*}
Recall that in our setting, $\phi$ denotes the estimating functional of an estimator of $\theta=S_{T_1,T_2}^0(t_1^0,t_2^0)$. Consequently, if $\phi$ is sufficiently smooth, then we can write 
\begin{equation*}
	\phi(F_n)= \phi(F)+\frac{1}{n}\sum_{i=1}^{n}\dot{\phi}(X_i)+\frac{1}{2n^2}\sum_i\sum_j \ddot{\phi}(X_i,X_j)+R_n,
\end{equation*}
where $\dot{\phi}(X)=\phi_F'(\delta_X-F)$ is its first order influence function, and where $\ddot{\phi}(X_1,X_2)=\phi_F''(\delta_{X_1}-F,\delta_{X_2}-F)$ is its second order influence function. Note that such an estimator is approximately unbiased in the sense that
\begin{align*}
	E\left[\hat{\theta}\right]=E\left[\phi(F_n)\right]&=E\left[\phi(F)\right]+\frac{1}{n}\sum_{i=1}^{n}E\left[\dot{\phi}(X_i)\right]+\frac{1}{2n^2}\sum_i\sum_j E\left[\ddot{\phi}(X_i,X_j)\right]+E[R_n]\\
	&=\phi(F)+\frac{1}{n}\sum_{i=1}^{n}0+\frac{1}{2n^2}\sum_i\sum_j 0+E[R_n]=\theta+E[R_n],
\end{align*}
which follows from the fact that $E\left[\dot{\phi}(X)\right]=0$ and $E\left[\ddot{\phi}(X,y)\right]=0$ for every $y$ (see \citealt{vaart_asymptotic_1998}, Chapter 20). That is, $E[\hat{\theta}]=\theta+E[R_n]$
which means that a plug-in estimator of a smooth functional is approximately unbiased, as the remainder term converges to zero. 

Similarly to the expansion above, for the average $F_n^{(k)}=\frac{1}{n-1}\sum_{i\neq k}\delta_{X_i}$ based on a sample of size $n-1$, we have that $\phi(F_n^{(k)})=\hat{\theta}^{-k}$ and that 
\begin{equation*}
	\phi(F_n^{(k)})= \phi(F)+\frac{1}{n-1}\sum_{i\neq k}\dot{\phi}(X_i)+\frac{1}{2(n-1)^2}\sum_{i\neq k}\sum_{j\neq k} \ddot{\phi}(X_i,X_j)+R_n^{-k}.
\end{equation*}

Hence, the pseudo-observation $\hat{\theta}_k=n\hat{\theta}-(n-1)\hat{\theta}^{-k}$ can be obtained by
\[
\hat{\theta}_k=n\phi(F_n)-(n-1)\phi(F_n^{(k)})=\phi(F)+\dot{\phi}(X_k)+\frac{1}{n-1}\sum_{i\neq k}\ddot{\phi}(X_k,X_i)+R_{n,k},
\]
where the remainder term $R_{n,k}$ includes also some of the second order terms. For more information see \cite{overgaard_asymptotic_2017} and the references therein.

For the theory of pseudo-observations to hold, the estimating functional $\phi$ needs to be two times (Fréchet) differentiable with a Lipschitz continuous second-order derivative, and the norm of $F_n-F$ needs to converge at a certain rate. This ensures that all remainder terms $R_{n,k}$ in the von-Mises expansion of the pseudo-observations converge (in probability) to zero, and corresponds to Theorem 3.1 of \cite{overgaard_asymptotic_2017}. The second condition is that the conditional expectation of the influence function $E\left[\dot{\phi}(X)|Z\right]$ is equal to $S_{T_1,T_2}(t_1,t_2|Z)-S_{T_1,T_2}(t_1,t_2)$, which corresponds to Equation 3.36 in Theorem 3.4 of \cite{overgaard_asymptotic_2017} for the specific case where $\theta=S_{T_1,T_2}(t_1,t_2)$. In Section~\ref{sec:LY} we prove these conditions for the Lin and Ying estimator, and in Section~\ref{sec:Dab} for the Dabrowska estimator. Finally, in Section~\ref{sec:covariance} we state the result and present the covariance matrix of the limiting Gaussian distribution.

\subsection{The Lin and Ying estimator}\label{sec:LY}
In Section~\ref{sec:LY_functional} we present the estimating functional of the Lin and Ying estimator, and discuss the convergence of $\norm{F_n-F}$. In Section~\ref{sec:IF} we discuss the smoothness of this functional and derive the associated first order influence function. In Section~\ref{sec:cond_expec} we prove that $E\left[\dot{\phi}(X)|Z\right]=S_{T_1,T_2}(t_1,t_2|Z)-S_{T_1,T_2}(t_1,t_2)$. 

\subsubsection{The estimating functional of the Lin and Ying estimator}\label{sec:LY_functional}
Recall that we observe $n$ i.i.d. quintuples $\{(Y_{11},\Delta_{11},Y_{21}, \Delta_{21}, Z_1),\ldots,(Y_{1n},\Delta_{1n},Y_{2n}, \Delta_{2n}, Z_n)\}$, where for $j=1,2$, $Y_j=\min(T_j,C_j)$ is the minimum between the survival and censoring times, and $\Delta_j=I(T_j \leq C_j)$ is the corresponding indicator. We denote by $\{X_1,\ldots, X_n\}$ the observed data without the covariates; that is, for $i=1,\ldots,n$, $X_i=(Y_{1i},\Delta_{1i},Y_{2i}, \Delta_{2i})$.
Using the notation of \cite{overgaard_asymptotic_2017}, we define similarly:
\begin{enumerate}
	\item $Y:\mathbb{R}_+^2\mapsto \{0,1\}$ such that  $Y(t_1,t_2)=1(Y_1>t_1,Y_2>t_2)$
	\item $Y_X:\mathbb{R}_+^2\mapsto \{0,1\}$ such that  $Y_X(t_1,t_2)=1(C^*\geq \max(t_1,t_2))$, where $C^*=\max(Y_1,Y_2)$
	\item $N_{X,1}:\mathbb{R}_+^2\mapsto \{0,1\}$ such that $N_{X,1}(t_1,t_2)=1(C^*\leq \max(t_1,t_2), \Delta=1)$, where $\Delta=1-\Delta_1\Delta_2$
\end{enumerate}
Define $\delta_X=(Y, Y_X, N_{X,1})$, and $F_n=\frac{1}{n}\sum_{i=1}^n\delta_{X_i}$. Note that
\[
F=\lim_{n\rightarrow \infty}F_n=(H,H_*,H_1)
\] where $H(t_1,t_2)=P(Y_1>t_1,Y_2>t_2)$, $H_*(t_1,t_2)=P(C^*\geq \max(t_1,t_2))$ and $H_1(t_1,t_2)=P(C^*\leq \max(t_1,t_2), \Delta=1)$.

Let $F_n, F \in D$, where $D$ is some space of the form $D=\{\tilde{h}: \tilde{h}=(h, h_*,h_1)\}$. That is, $D$ is a space of vectors $\tilde{h}=(h, h_*,h_1)$; the first element of the vector, $h$, is a function of a bivariate time point $(t_1,t_2)$, whereas the two other elements $h_*$ and $h_1$ are actually functions of a univariate time point $t^*=\max(t_1,t_2)$. Note also that by Lemma 2.1 in the supplement of \cite{overgaard_asymptotic_2017}, we have that $\norm{F_n-F}$ converges to zero sufficiently fast.
Next, as in \cite{overgaard_asymptotic_2017}, define the Nelson-Aalen estimating functional $\psi:D\mapsto \mathbb{R}$ by $$\psi(\tilde{h})=\int_0^{(\cdot)} \frac{1(h_*(s)>0)}{h_*(s)}dh_1(s),$$ and where we are only using the second and third coordinates of $\tilde{h}$. The corresponding Kaplan-Meier estimating functional is 
$$\chi(\tilde{h})=\prod_{0}^{(\cdot)}(1-\psi(ds;\tilde{h})),$$ where $\psi(ds;\tilde{h})=\psi(\tilde{h})(ds)$.

The estimating functional for the Lin and Ying estimator is given by 
\[
\phi(\tilde{h})=\frac{h}{\chi(\tilde{h})}.
\]

We note that 
\[
\phi(F)(t_1,t_2)=\frac{H(t_1,t_2)}{\chi(F)(\max(t_1,t_2))}=\frac{P(Y_1>t_1,Y_2>t_2)}{P(C>\max(t_1,t_2))}=S_{T_1,T_2}(t_1,t_2)
\]
is the true parameter of interest, and 
\[
\phi(F_n)(t_1,t_2)=\frac{\frac{1}{n}\sum_{i=1}^{n}I(Y_{1i}>t_1, Y_{2i}>t_2)}{\hat{G}(\max(t_1,t_2))}
\]
is the Lin and Ying estimator, and where $\hat{G}$ is the Kaplan-Meier estimator of the survival function of the univariate censoring variable.

\subsubsection{The influence function of the Lin and Ying estimator}\label{sec:IF}
To obtain the Fréchet derivative of $\phi$, we re-express $\phi(\tilde{h})$ as a quotient of two functionals
\[
\phi(\tilde{h})=\frac{\nu(\tilde{h})}{\chi(\tilde{h})},
\]
where $\nu (\tilde{h})=h$ returns the first coordinate of $\tilde{h}$. This representation of $\phi$, as the quotient of two functionals, will allow us to use the product rule for differentiable functionals. 

Note that the coordinate projection $\nu (\tilde{h})=h$ is differentiable of any order, and that the Kaplan-Meier functional $\chi(\tilde{h})$ is also differentiable of any order in a neighborhood of $F$ (\citealt[Section 4.1]{overgaard_asymptotic_2017}). Hence, the quotient $\phi(\tilde{h})=\frac{\nu(\tilde{h})}{\chi(\tilde{h})}$ is differentiable of any order in a neighborhood of $F$ (which is assumed to be bounded away from zero). Specifically, our estimating functional $\phi(\tilde{h})$ is two times (Fréchet) differentiable with a Lipschitz continuous second-order derivative (since by Theorem 1.4 of the supplement of \citealt{overgaard_asymptotic_2017}, existence of the (k+1)-order derivative guarantees Lipschitz continuity of the k-order derivative).

Consequently, the Fréchet derivative of $\phi$ at $\tilde{h}$ in the direction of $\tilde{g}$ can be obtained using the product rule (which also holds for the Fréchet derivative)
\[
\phi_{\tilde{h}}'(\tilde{g})=\frac{\nu_{\tilde{h}}'(\tilde{g})}{\chi(\tilde{h})}-\frac{\nu (\tilde{h})\chi_{\tilde{h}}'(\tilde{g})}{\chi^2(\tilde{h})}.
\]
Next, note that similarly to \cite{overgaard_asymptotic_2017}, $\chi_{\tilde{h}}'(\tilde{g})=-\chi(\tilde{h})\psi_{\tilde{h}}'(\tilde{g})$, and $$\psi_{\tilde{h}}'(\tilde{g})=\int_{0}^{(\cdot)}\frac{1}{h_*(s)}dg_1(s)-\int_{0}^{(\cdot)}\frac{g_*(s)}{[h_*(s)]^2}dh_1(s).$$
We propose $\nu_{\tilde{h}}'(\tilde{g})=g$ as the candidate for the Fréchet derivative of $\nu$ at $\tilde{h}$ in the direction of $\tilde{g}$. This is because 
\[
\frac{\nu(\tilde{h}+u\tilde{g})-\nu(\tilde{h})}{u}=\frac{h+ug-h}{u}=g\equiv \nu(\tilde{g}). 
\]
Consequently,
\[
\phi_{\tilde{h}}'(\tilde{g})=\frac{\nu(\tilde{g})}{\chi(\tilde{h})}-\frac{\nu (\tilde{h})\chi_{\tilde{h}}'(\tilde{g})}{\chi^2(\tilde{h})}.
\]
Specifically, note that the Fréchet derivative of $\phi$ at $F=(H,H_*,H_1)$ in the direction of $\tilde{g}$ is
\[
\phi_{F}'(\tilde{g})=\frac{\nu(\tilde{g})}{\chi(F)}-\frac{\nu(F)\chi_{F}'(\tilde{g})}{\chi^2(F)}=\frac{g}{\chi(F)}-\frac{H\cdot \chi_{F}'(\tilde{g})}{\chi^2(F)},
\]
where the last equality is due to the fact that $\nu_{\tilde{h}}'(\tilde{g})=g=\nu(\tilde{g})$.

Note also that $\chi(F)(\max(t_1,t_2))=P(C>\max(t_1,t_2))=G(\max(t_1,t_2))$,  $\chi_{F}'(\tilde{g})=-\chi(F)\psi_{F}'(\tilde{g})$, and that $$\psi_{F}'(\tilde{g})=\int_{0}^{(\cdot)}\frac{1}{H_*(s)}dg_1(s)-\int_{0}^{(\cdot)}\frac{g_*(s)}{[H_*(s)]^2}dH_1(s).$$ 

Consequently, for a bivariate time point $(t_1,t_2)$ with $t^*=\max(t_1,t_2)$ we have that
\[
\phi_{F}'(\tilde{g})(t_1,t_2)=\frac{g(t_1,t_2)}{G(t^*)}-\frac{H(t_1,t_2)\cdot \left(-G(t^*)\right)\left[\int_{0}^{t^*}\frac{1}{H_*(s)}dg_1(s)-\int_{0}^{t^*}\frac{g_*(s)}{[H_*(s)]^2}dH_1(s)\right]}{[G(t^*)]^2},
\]
where $\tilde{g}=(g,g_*,g_1)$, $F=(H,H_*,H_1)$, and where $H(t_1,t_2)=P(Y_1>t_1,Y_2>t_2)$, $H_*(t^*)=P(C^*\geq t^*)$ and $H_1(t^*)=P(C^*\leq t^*, \Delta=1)$.  Hence, after dividing by common factors, we obtain that
\[
\phi_{F}'(\tilde{g})(t_1,t_2)=\frac{g(t_1,t_2)}{G(t^*)}+\frac{H(t_1,t_2)\left[\int_{0}^{t^*}\frac{1}{H_*(s)}dg_1(s)-\int_{0}^{t^*}\frac{g_*(s)}{[H_*(s)]^2}dH_1(s)\right]}{G(t^*)}.
\]

Finally, recall that $\delta_X-F=(Y-H,Y_X-H_*,N_{X,1}-H_1)$, and consequently the influence function, defined by $\dot{\phi}(X)=\phi_F'(\delta_X-F)$ is 
\[
\dot{\phi}(X)=\phi_{F}'(\delta_X-F)=\frac{Y-H}{G}+\frac{H\left[\int_{0}^{(\cdot)}\frac{1}{H_*(s)}dM_{X,1}(s)\right]}{G},
\]
where as in \cite{overgaard_asymptotic_2017}, it can be shown that 
\[
\int_{0}^{t^*}\frac{1}{H_*(s)}d(N_{X,1}-H_1)(s)-\int_{0}^{t^*}\frac{Y_X(s)-H_*(s)}{[H_*(s)]^2}dH_1(s)=\int_{0}^{t^*}\frac{1}{H_*(s)}dM_{X,1}(s),
\]
where $M_{X,1}=N_{X,1}-\int_{0}^{(\cdot)}Y_X(s)d\Lambda_{C^*}(s)$ and where $\Lambda_{C^*}(t)=\int_0^t\frac{1}{H_*(s)}dH_1(s)$ is the cumulative hazard of the variable $C^*=\min(C,\max(T_1,T_2))$.

Interestingly, the influence function can be represented as
\[
\dot{\phi}(X)=\frac{Y}{G}-\frac{H}{G}\left[1-\int_{0}^{(\cdot)}\frac{1}{H_*(s)}dM_{X,1}(s)\right],
\]
where $\frac{H}{G}(t_1,t_2)=\frac{P(Y_1>t_1,Y_2>t_2)}{G(\max(t_1,t_2))}=S_{T_1,T_2}(t_1,t_2)$ is exactly our parameter of interest $\theta$, and  $\frac{Y}{G}(t_1,t_2)=\frac{1(Y_1>t_1,Y_2>t_2)}{G(\max(t_1,t_2))}$. Hence,  for $(t_1,t_2)$ with $t^*=\max(t_1,t_2)$ we have that 
\begin{equation}\label{eq:IF1}
	\dot{\phi}(X)(t_1,t_2)=\frac{1(Y_1>t_1,Y_2>t_2)}{G(\max(t_1,t_2))}-S_{T_1,T_2}(t_1,t_2)\left[1-\int_{0}^{t^*}\frac{1}{H_*(s)}dM_{X,1}(s)\right].
\end{equation}
This presentation of the influence function is equivalent to Equation (2.3) of \cite{huang_empirical_2018}, which in itself relies on Equation (A.2) of the original paper by \cite{lin_simple_1993}.
Specifically, \cite{lin_simple_1993} showed that $\sqrt{n}(\hat{S}(t_1,t_2)-S_{T_1,T_2}(t_1,t_2))$ is asymptotically equivalent to $\frac{1}{\sqrt{n}}\sum_{i=1}^{n}T_i$, where
\begin{equation}\label{eq:IF12}
	T_i=\frac{1(Y_{1i}>t_1,Y_{2i}>t_2)}{G(\max(t_1,t_2))}-S_{T_1,T_2}(t_1,t_2)\left[1-\int_{0}^{t^*}\frac{\Delta_i d1(C_i^*\leq s)-I(C_i^*\geq s)d\Lambda_{C}(s)}{G(s)P(\max(T_1,T_2))\geq s)}\right],
\end{equation}
are i.i.d. random variables, and where $\Lambda_{C}(s)$ is the cumulative hazard of the censoring variable $C$. 
Hence, in order to prove the equivalence between Equations~\eqref{eq:IF1} and \eqref{eq:IF12}, we need to prove that 
\[
\int_{0}^{t^*}\frac{1}{H_*(s)}dM_{X,1}(s)=\int_{0}^{t^*}\frac{\Delta d1(C^*\leq s)-I(C^*\geq s)d\Lambda_{C}(s)}{G(s)P(\max(T_1,T_2))\geq s)}.
\]
\subsubsection*{Proving the equivalence between Equations~\eqref{eq:IF1} and \eqref{eq:IF12}}
Indeed, $M_{X,1}=N_{X,1}-\int_{0}^{(\cdot)}Y_X(s)d\Lambda_{C^*}(s)$ and thus,
\begin{align*}
	\int_{0}^{t^*}\frac{1}{H_*(s)}dM_{X,1}(s)&=\int_{0}^{t^*}\frac{1}{H_*(s)}dN_{X,1}(s)-\int_{0}^{t^*}\frac{Y_X(s)}{H_*(s)}d\Lambda_{C^*}(s)\\
	&=\int_{0}^{t^*}\frac{1}{P(C^*\geq s)}d1(C^*\leq s, \Delta=1)-\int_{0}^{t^*}\frac{1(C^*\geq s)}{P(C^*\geq s)}d\Lambda_{C^*}(s)\\
	&\equiv A-B.
\end{align*}

Next, note that
\[
A\equiv\int_{0}^{t^*}\frac{1}{P(C^*\geq s)}d1(C^*\leq s, \Delta=1)=\int_{0}^{t^*}\frac{\Delta d1(C^*\leq  s)}{P(C^*\geq s)}.
\]
Note also that since $C$ is an independent censoring variable and that $C^*=\min(C,\max(T_1,T_2))$, we have that $$P(C^*\geq s)=P(\min(C,\max(T_1,T_2))\geq s)=P(C\geq s)P(\max(T_1,T_2)\geq s)=G(s^-)P(\max(T_1,T_2)\geq s).$$
Consequently,
\[
A=\int_{0}^{t^*}\frac{\Delta d1(C^* \leq s)}{G(s^-)P(\max(T_1,T_2)\geq s)},
\]
and 
\[
B\equiv \int_{0}^{t^*}\frac{1(C^*\geq s)}{P(C^*\geq s)}d\Lambda_{C^*}(s)=\int_{0}^{t^*}\frac{1(C^*\geq s)}{G(s^-)P(\max(T_1,T_2)\geq s)}d\Lambda_{C^*}(s).
\]
Finally, note that 
\[
\Lambda_{C^*}(t)=\int_{0}^{t}\frac{1}{H_*(s)}dH_1(s)=\int_{0}^{t}\frac{1}{P(C^*\geq s)}dP(C^*\leq s, \Delta=1),
\]
and that by the independent censoring assumption,
\[
dP(C^*\leq s, \Delta=1)=P(C^*\in ds, \Delta=1)=P(C\in ds, \max(T_1,T_2)\geq s)=P(C\in ds)P(\max(T_1,T_2)\geq s),
\]
and
\[
P(C^*\geq s)=P(C\geq s,\max(T_1,T_2)\geq s)=P(C\geq s)P(\max(T_1,T_2)\geq s).
\]
Consequently,
\begin{align*}
	\Lambda_{C^*}(t)&=\int_{0}^{t}\frac{1}{P(C^*\geq s)}dP(C^*\leq s, \Delta=1)=\int_{0}^{t}\frac{P(\max(T_1,T_2)>s)}{P(C\geq s)P(\max(T_1,T_2)\geq s)}dP(C\leq s)\\
	&=\int_{0}^{t}\frac{1}{P(C\geq s)}dP(C\leq s)\equiv \Lambda_{C}(t).
\end{align*}
That is, the hazard of the observed censoring time $C^*$ is equivalent to the hazard of the true censoring time $C$.

In summary, we showed that 
\begin{align*}
	\int_{0}^{t^*}\frac{1}{H_*(s)}dM_{X,1}(s)&=\int_{0}^{t^*}\frac{\Delta d1(C^* \leq s)}{G(s^-)P(\max(T_1,T_2)\geq s)}-\int_{0}^{t^*}\frac{1(C^*\geq s)}{G(s^-)P(\max(T_1,T_2)\geq s)}d\Lambda_{C}(s),
\end{align*}
which is exactly what we wanted to prove (providing that $C$ is a continuous random variable such that $G(s^-)=G(s)$).

\subsubsection{The conditional expectation of the Lin and Ying influence function}\label{sec:cond_expec}
Our last step is to show that when adding covariates $Z$, we have that $$E[\dot{\phi}(X)\mid Z]=S_{T_1,T_2}(t_1,t_2\mid Z)-S_{T_1,T_2}(t_1,t_2),$$ which corresponds to \citealt[Assumption 3.36, Therorem 3.4]{overgaard_asymptotic_2017}. 
We first recall the definition of $M_{X,1}$: 
\begin{align*}
	M_{X,1}(s)=N_{X,1}(s)-\int_{0}^{s}Y_X(u)d\Lambda_{C^*}(u)&=1(C^*\leq s, \Delta=1)-\int_0^s 1(C^* \geq u)d\Lambda_{C^*}(u)\\
	&=1(C^*\leq s, \Delta=1)-\int_0^s 1(C^* \geq u)d\Lambda_{C}(u),
\end{align*}
where the last equality is due to the equality between the hazards of the observed censoring time $C^*$ and the true censoring time $C$.
Consequently,
\begin{align*}
	E[dM_{X,1}(s)\mid Z]&=E[1(C^*\in ds, \Delta=1)\mid Z]-E[1(C^*\geq s)\mid Z]d\Lambda_{C}(s)\\
	&=P(C^*\in ds, \Delta=1 \mid Z)-P(C^*\geq s \mid Z)d\Lambda_{C}(s).
\end{align*}
Next, by the independent censoring assumption,
\begin{align*}
	P(C^*\in ds, \Delta=1 \mid Z)&=P(C\in ds, \max(T_1,T_2)\geq s\mid  Z)=P(C\in ds)P(\max(T_1,T_2)\geq s \mid Z),
\end{align*}
and 
\begin{align*}
	P(C^*\geq s \mid Z)=P(C\geq s, \max(T_1,T_2)\geq s \mid Z)=P(C\geq s)P(\max(T_1,T_2)\geq s \mid Z)
\end{align*}

Consequently,
\begin{align*}
	E[dM_{X,1}(s)\mid Z]=P(C\in ds)P(\max(T_1,T_2)\geq s \mid Z)-P(C\geq s)P(\max(T_1,T_2)\geq s \mid Z)d\Lambda_{C}(s).
\end{align*}

Recall also that 
\[
\Lambda_{C}(t)\equiv \int_{0}^{t}\frac{1}{P(C\geq s)}dP(C\leq s),
\]
and thus,
\begin{align*}
	E[dM_{X,1}(s)\mid Z]&=P(C\in ds)P(\max(T_1,T_2)\geq s \mid Z)-P(C\geq s)P(\max(T_1,T_2)\geq s \mid Z)d\Lambda_{C}(s)\\
	&=P(C\in ds)P(\max(T_1,T_2)\geq s \mid Z)-\frac{P(C\geq s)P(\max(T_1,T_2)\geq s \mid Z)}{P(C\geq s)}dP(C \leq s)\\
	&=P(C\in ds)P(\max(T_1,T_2)\geq s \mid Z)-P(\max(T_1,T_2)\geq s \mid Z)dP(C \leq s)=0.
\end{align*}
That is, conditional expectations of integrals with respect to $dM_{X,1}(s)$ will also be equal to zero. Specifically, $E\left[\int_{0}^{(\cdot)}\frac{1}{H_*(s)}dM_{X,1}(s)\mid Z\right]=0$.
Consequently,
\begin{align*}
	E[\dot{\phi}(X)\mid Z]&=E\left[\frac{Y}{G}-\frac{H}{G}\left[1-\int_{0}^{(\cdot)}\frac{1}{H_*(s)}dM_{X,1}(s)\right]\mid Z\right]\\
	&=E\left[\frac{Y}{G}-\frac{H}{G}\mid Z\right]=\frac{1}{G}E[Y-H\mid Z]\\
	&=\frac{1}{G}\left(E[Y\mid Z]-H\right),
\end{align*}
where the last equality is due to the fact that $H(\cdot,\cdot)=P(Y_1>\cdot, Y_2>\cdot)$ is not a random variable.

Finally, for $(t_1,t_2)$ with $t^*=\max(t_1,t_2)$ we have that 
\begin{align*}
	E[\dot{\phi}(X)(t_1,t_2)\mid Z]&=\frac{1}{G(t^*)}\left(E[Y(t_1,t_2)\mid Z]-H(t_1,t_2)\right)\\
	&=\frac{1}{G(t^*)}\left(E[1(Y_1>t_1,Y_2>t_2)\mid Z]-P(Y_1>t_1,Y_2>t_2)\right)\\
	&=\frac{P(Y_1>t_1,Y_2>t_2\mid Z)}{G(t^*)}-\frac{P(Y_1>t_1,Y_2>t_2)}{G(t^*)}\\
	&=S_{T_1,T_2}(t_1,t_2\mid Z)-S_{T_1,T_2}(t_1,t_2),
\end{align*}
where the last equality is due to the fact that $C\indep (T_1,T_2,Z)$.

\subsection{The Dabrowska estimator}\label{sec:Dab}
In Section~\ref{sec:Dab_pre} we first define the bivariate hazard vector, and then present the Dabrowska estimator. In Section~\ref{sec:Dab_functional} we present the estimating functional of the Dabrowska estimator, and discuss the convergence of $\norm{F_n-F}$. In Section~\ref{sec:IF_Dab} we discuss the smoothness of this functional and derive the associated first order influence function. Note that both the estimating functional of the Dabrowska estimator and its influence function have already been presented in \cite{gill_inefficient_1995}, but are also derived here for completeness. In Section~\ref{sec:cond_expec_Dab} we prove that $E\left[\dot{\phi}(X)\mid Z\right]=S_{T_1,T_2}(t_1,t_2\mid Z)-S_{T_1,T_2}(t_1,t_2)$. 
\subsubsection{Preliminaries}\label{sec:Dab_pre}
As before, we denote by $S_{T_1,T_2}(t_1,t_2)$ the bivariate survival function of the pair $(T_1,T_2)$, and by $S_{T_1}(t_1)$ and $S_{T_2}(t_2)$ the separate marginal survival functions of $T_1$ and $T_2$, respectively.
We write $P(T_1 \in dt_1, T_2\in dt_2)$ instead of $P(t_1\leq T_1 \leq t_1+dt_1, t_2\leq T_2 \leq t_2+dt_2)$. Denote by $$\Lambda_{11}(ds,dt)=P\left(T_1 \in ds, T_2\in dt \mid T_1\geq s, T_2\geq t\right)=\frac{P(T_1 \in ds, T_2\in dt)}{P(T_1\geq s, T_2\geq t)}=\frac{S_{T_1,T_2}(ds,st)}{S_{T_1,T_2}(s^-,t^-)}$$ the instantaneous risk of a double failure, and by $$\Lambda_{10}(ds,t)=P\left(T_1 \in ds, T_2>t \mid T_1\geq s, T_2> t\right)=\frac{P(T_1 \in ds, T_2>t)}{P(T_1\geq s, T_2>t)}=\frac{P(T_1 \in ds \mid T_2>t)}{P(T_1\geq s \mid T_2>t)}$$ and $$ \Lambda_{01}(s,dt)=P\left(T_1 >s, T_2\in dt \mid T_1> s, T_2\geq t\right)=\frac{P(T_1 >s, T_2\in dt)}{P(T_1> s, T_2\geq t)}=\frac{P(T_2 \in dt \mid T_1>s)}{P(T_2\geq t \mid T_1>s)}$$ the instantaneous risks of single failures. Note that $\Lambda_{10}(ds,t)$ and $\Lambda_{01}(s,dt)$ can also be interpreted as conditional univariate hazards, given the events $T_2>t$ and $T_1>s$, respectively. We denote by $\Lambda(s,t)=(\Lambda_{10}(s,t),\Lambda_{01}(s,t),\Lambda_{11}(s,t))$ the bivariate cumulative hazard vector, consisting of the three bivariate cumulative hazard functions.

Set $A(t_1,t_2)=\log(S_{T_1,T_2}(t_1,t_2))$. Then, as in \cite{dabrowska_kaplan-meier_1988}, for $(t_1,t_2)\in [0,\tau_1]\times [0,\tau_2]$ such that $S_{T_1,T_2}(\tau_1,\tau_2)>0$, we have that
\begin{equation}\label{eq1}
	\begin{aligned}
		S_{T_1,T_2}(t_1,t_2)&=\exp\{A(t_1,t_2)\}=\exp\left\{\int_{0}^{t_1}\int_{0}^{t_2}A(du,dv)+A(t_1,0)+A(0,t_2)\right\}\\
		&=\exp\{A(t_1,0)\}\exp\{A(0,t_2)\}\exp\left\{\int_{0}^{t_1}\int_{0}^{t_2}A(du,dv)\right\}\\
		&=S_{T_1}(t_1)S_{T_2}(t_2)\exp\left\{\int_{0}^{t_1}\int_{0}^{t_2}A(du,dv)\right\},
	\end{aligned}
\end{equation}
where $\exp\left\{\int_{0}^{t_1}\int_{0}^{t_2}A(du,dv)\right\}=\frac{S_{T_1,T_2}(t_1,t_2)}{S_{T_1}(t_1)S_{T_2}(t_2)}$ is the cross-ratio term.
Next, as in \cite{dabrowska_kaplan-meier_1988}, each of the three terms in this factorization of the bivariate survival function can be defined using product integration:
\begin{align*}
	S_{T_1,T_2}(t_1,t_2)&=\prod_{(0,t_1]}\{1-\Lambda_{10}(du,0)\}\prod_{(0,t_2]}\{1-\Lambda_{01}(0,dv)\}
	\prod_{\substack{
			(0,t_1] \\
			(0,t_2]}}\left(1+L(du,dv)\right),
\end{align*}
where $\prod$ denotes the product integral, \[
L(ds_1,ds_2)=\frac{\Lambda_{11}(ds_1,ds_2)-\Lambda_{10}(ds_1,s_2)\Lambda_{01}(s_1,ds_2)}{\left(1-\Lambda_{10}(\Delta s_1, s_2^-)\right)\left(1-\Lambda_{01}( s_1^-, \Delta s_2)\right)},
\]
and where $\Lambda_{10}(\Delta s_1, s_2^-)=\Lambda_{10}(s_1, s_2^-)-\Lambda_{10}(s_1^-, s_2^-)$ and $\Lambda_{01}( s_1^-, \Delta s_2)=\Lambda_{01}(s_1^-, s_2)-\Lambda_{01}(s_1^-, s_2^-)$.
That is, the bivariate survival function is presented in terms of the bivariate hazard vector. Consequently, the Dabrowska estimator is given by \[
\hat{S}_{T_1,T_2}(t_1,t_2)=\prod_{(0,t_1]}\{1-\hat{\Lambda}_{10}(du,0)\}\prod_{(0,t_2]}\{1-\hat{\Lambda}_{01}(0,dv)\}
\prod_{\substack{
		(0,t_1] \\
		(0,t_2]}}\left(1+\hat{L}(du,dv)\right),
\]
where we just plug-in the estimates of the bivariate hazard vector, as will be be discussed below (for more information see \citealt{dabrowska_kaplan-meier_1988}). Finally, note that $\prod_{(0,t_1]}\{1-\hat{\Lambda}_{10}(du,0)\}$ and $\prod_{(0,t_2]}\{1-\hat{\Lambda}_{01}(0,dv)\}$ are just the univariate Kaplan-Meier estimators of the marginal survival functions $S_{T_1}(t_1)$ and $S_{T_2}(t_2)$, respectively. Hence, for an element $\tilde{h} \in D$, the estimating functional of the Dabrowska estimator is some functional $\phi:D\mapsto\mathbb{R}$ which is given by the product of three estimating functionals
\[
\phi(\tilde{h})=\chi_1(\tilde{h})\chi_2(\tilde{h})\Gamma(\rho(\tilde{h})),
\]
where for $j=1,2$, $\chi_j(\tilde{h})$ denotes the estimating functional of the Kaplan-Meier estimator, and $\Gamma(\rho(\tilde{h}))$ denotes the estimating functional of the cross-ratio term $\prod_{\substack{
		(0,t_1] \\
		(0,t_2]}}\left(1+\hat{L}(du,dv)\right)$. In Section~\ref{sec:Dab_functional} below we present each of these three estimating functionals. 

\subsubsection{The estimating functional of the Dabrowska estimator}\label{sec:Dab_functional}
Recall that the observed data consists of $n$ i.i.d. copies $\left\{X_1,\ldots, X_n\right\}$.
Define the following functions 
\begin{enumerate}
	\item $Y:\mathbb{R}^2_+\mapsto \{0,1\}$ such that  $Y(t_1,t_2)=1(Y_1> t_1, Y_2> t_2)$
	\item $N_{10}:\mathbb{R}^2_+\mapsto \{0,1\}$ such that $N_{10}(t_1,t_2)=1(Y_1\leq t_1, Y_2> t_2, \Delta_1=1)$
	\item $N_{01}:\mathbb{R}^2_+\mapsto \{0,1\}$ such that $N_{01}(t_1,t_2)=1(Y_1> t_1, Y_2\leq t_2, \Delta_2=1)$
	\item $N_{11}:\mathbb{R}^2_+\mapsto \{0,1\}$ such that $N_{11}(t_1,t_2)=1(Y_1\leq t_1, Y_2\leq t_2, \Delta_1=1,\Delta_2=1)$
\end{enumerate}
Define $\delta_X=(Y,N_{10},N_{01},N_{11})$, and define the sample average of $\delta_X$ by $F_n=\frac{1}{n}\sum_{i=1}^n
\delta_{X_i}$. Note that $F_n$ is a vector where each coordinate is the mean of the corresponding indicator function. That is, $F_n=(\hat{H}^n,\hat{H}_{10}^n,\hat{H}_{01}^n,\hat{H}_{11}^n)$, where 
\begin{enumerate}
	\item $\hat{H}^n=\frac{1}{n}\sum _{i=1}^n 1(Y_{1i}> t_1, Y_{2i}> t_2)$
	\item $\hat{H}_{10}^n=\frac{1}{n}\sum _{i=1}^n 1(Y_{1i}\leq t_1, Y_{2i}> t_2, \Delta_{1i}=1)$
	\item $\hat{H}_{01}^n=\frac{1}{n}\sum _{i=1}^n 1(Y_{1i}> t_1, Y_{2i}\leq t_2, \Delta_{2i}=1)$
	\item $\hat{H}_{11}^n=\frac{1}{n}\sum _{i=1}^n 1(Y_{1i}\leq t_1, Y_{2i}\leq t_2, \Delta_{1i}=1,\Delta_{2i}=1)$
\end{enumerate}
and its limit is $F=\lim_{n\rightarrow \infty}F_n=(H,H_{10},H_{01},H_{11})$ where $H(t_1,t_2)=P(Y_1> t_1,Y_2> t_2)$,  $H_{10}(t_1,t_2)=P(Y_1\leq t_1, Y_2> t_2, \Delta_1=1)$,  $H_{01}(t_1,t_2)=P(Y_1> t_1, Y_2\leq t_2, \Delta_2=1)$, and  $H_{11}(t_1,t_2)=P(Y_1\leq t_1, Y_2\leq t_2, \Delta_1=1, \Delta_2=1)$. Let $F_n, F \in D$, where $D$ is some space of the form $D=\{\tilde{h}: \tilde{h}=(h,h_{10},h_{01},h_{11})\}$. Define the two univariate Nelson-Aalen estimating functionals $\psi_1:D\mapsto \mathbb{R}$ and $\psi_2:D\mapsto \mathbb{R}$ by $$\psi_1(\tilde{h})=\int_0^{(\cdot)} \frac{1(h(s^-,0)>0)}{h(s^-,0)}dh_{10}(s,0) \ \text{and} \ \psi_2(\tilde{h})=\int_0^{(\cdot)} \frac{1(h(0,t^-)>0)}{h(0,t^-)}dh_{01}(0,t).$$

Note that $$\psi_1(F)=\int_0^{(\cdot)} \frac{1(H(s^-,0)>0)}{H(s^-,0)}dH_{10}(s,0)=\int_0^{(\cdot)} \frac{1}{H(s^-,0)}dH_{10}(s,0)=\Lambda_1(\cdot)=\Lambda_{10}(\cdot,0)$$ 
and
$$\psi_2(F)=\int_0^{(\cdot)} \frac{1(H(0,t^-)>0)}{H(0,t^-)}dH_{01}(0,t)=\int_0^{(\cdot)} \frac{1}{H(0,t^-)}dH_{01}(0,t)=\Lambda_2(\cdot)=\Lambda_{01}(0,\cdot)$$
return the univariate cumulative hazard functions, and that  $$\psi_1(F_n)=\int_0^{(\cdot)} \frac{1(Y^1(s)>0)}{Y^1(s)}dN_{1}(s)=\hat{\Lambda}_1(\cdot)$$ is the Nelson-Aalen estimator of the cumulative hazard of $T_1$, and 
$$\psi_2(F_n)=\int_0^{(\cdot)} \frac{1(Y^2(t)>0)}{Y^2(t)}dN_{2}(t)=\hat{\Lambda}_2(\cdot)$$ is the Nelson-Aalen estimator of the cumulative hazard of $T_2$, where we used $$Y^1(s)=\#\{i: Y_{1i}\geq s\}=\sum _{i=1}^n 1(Y_{1i}\geq s)=\sum _{i=1}^n 1(Y_{1i}> s^-, Y_{2i}> 0),$$
$$Y^2(t)=\#\{i: Y_{2i}\geq t\}=\sum _{i=1}^n 1(Y_{2i}\geq t)=\sum _{i=1}^n 1(Y_{1i}> 0, Y_{2i}> t^-),$$
$$N_1(s)=\#\{i: Y_{1i}\leq s, \Delta_{1i}=1\}=\sum _{i=1}^n 1(Y_{1i}\leq s,\Delta_{1i}=1)=\sum_{i=1}^n 1(Y_{1i}\leq s, Y_{2i} > 0, \Delta_{1i}=1),$$
and
$$N_2(t)=\#\{i: Y_{2i}\leq t, \Delta_{2i}=1\}=\sum _{i=1}^n 1(Y_{2i}\leq t,\Delta_{2i}=1)=\sum_{i=1}^n 1(Y_{1i}> 0, Y_{2i} \leq t, \Delta_{2i}=1)$$
(note the relationship to $F_n$).

Next, when analyzing each failure time separately, then for $j=1,2$ we can define the Kaplan-Meier estimating functionals $\chi_j:D\mapsto\mathbb{R}$ by 
$$\chi_1(\tilde{h})=\prod_{0}^{(\cdot)}(1-\psi_1(ds;\tilde{h})). \ \text{and} \ \chi_2(\tilde{h})=\prod_{0}^{(\cdot)}(1-\psi_2(dt;\tilde{h})),$$
where $\psi_j(ds;\tilde{h})=\psi_j(\tilde{h})(ds)$.
We note that for $j=1,2$, we have that
\[
\chi_j(F)=\prod_{0}^{(\cdot)}(1-\psi_j(ds;F))=\prod_{0}^{(\cdot)}(1-d\Lambda_j(s))
\]
is the true survival function of the failure time $T_j$, and that 
\[
\chi_j(F_n)=\prod_{0}^{(\cdot)}(1-\psi_j(ds;F_n))=\prod_{0}^{(\cdot)}(1-d\hat{\Lambda}_j(s))
\]
is the Kaplan-Meier estimator of the survival function of $T_j$ (where $\hat{\Lambda}_j$ is the Nelson-Aalen estimator). That is, $\chi_1$ and $\chi_2$ are the two estimating functionals of the Kaplan-Meier estimator of $T_1$ and $T_2$, respectively.

Next, as in \cite{dabrowska_kaplan-meier_1988}, the estimator of the cumulative hazards vector $\Lambda(s,t)=(\Lambda_{10}(s,t), \Lambda_{01}(s,t), \Lambda_{11}(s,t))$ is given by 
\begin{equation*}
	\begin{aligned}
		\hat{\Lambda}_{11}(s,t)&=\int_{0}^{s}\int_{0}^{t} \frac{\hat{H}_{11}^n(du,dv)}{\hat{H}^n(u^-,v^-)}\\
		\hat{\Lambda}_{10}(s,t)&=\int_{0}^{s} \frac{\hat{H}_{10}^n(du,t)}{\hat{H}^n(u^-,t)}\\
		\hat{\Lambda}_{01}(s,t)&=\int_{0}^{t} \frac{\hat{H}_{01}^n(s,dv)}{\hat{H}^n(s,v^-)}.
	\end{aligned}
\end{equation*}

More generally, for an element $\tilde{h} \in D=\{\tilde{h}: \tilde{h}=(h,h_{10},h_{01},h_{11})\}$ we can define the functional $\rho(\tilde{h})(s,t)=\left(\int_{0}^{s}\int_{0}^{t} \frac{h_{11}(du,dv)}{h(u^-,v^-)}, \int_{0}^{s} \frac{h_{10}(du,t)}{h(u^-,t)}, \int_{0}^{t} \frac{h_{01}(s,dv)}{h(s,v^-)}\right)$ such that
\[
\rho(F_n)(s,t)=\left(\int_{0}^{s}\int_{0}^{t} \frac{\hat{H}_{11}^n(du,dv)}{\hat{H}^n(u^-,v^-)}, \int_{0}^{s} \frac{\hat{H}_{10}^n(du,t)}{\hat{H}^n(u^-,t)}, \int_{0}^{t} \frac{\hat{H}_{01}^n(s,dv)}{\hat{H}^n(s,v^-)}\right)=\left(\hat{\Lambda}_{11}(s,t), \hat{\Lambda}_{10}(s,t), \hat{\Lambda}_{01}(s,t)\right),
\]
and for  $F=\lim_{n\rightarrow \infty}F_n=(H,H_{10},H_{01},H_{11})$ where $H(t_1,t_2)=P(Y_1> t_1,Y_2> t_2)$,  $H_{10}(t_1,t_2)=P(Y_1\leq t_1, Y_2> t_2, \Delta_1=1)$,  $H_{01}(t_1,t_2)=P(Y_1> t_1, Y_2\leq t_2, \Delta_2=1)$, and  $H_{11}(t_1,t_2)=P(Y_1\leq t_1, Y_2\leq t_2, \Delta_1=1, \Delta_2=1)$, we have that
\[
\rho(F)(s,t)=\left(\int_{0}^{s}\int_{0}^{t} \frac{H_{11}(du,dv)}{H(u^-,v^-)}, \int_{0}^{s} \frac{H_{10}(du,t)}{H(u^-,t)}, \int_{0}^{t} \frac{H_{01}(s,dv)}{H(s,v^-)}\right)=\left(\Lambda_{11}(s,t), \Lambda_{10}(s,t), \Lambda_{01}(s,t)\right).
\]
That is, $\rho$ is the estimating functional for the Dabrowska estimator of the cumulative hazard vector.
For ease of notation, we will write $\rho(\tilde{h})=\left(\rho_1(\tilde{h}),\rho_2(\tilde{h}),\rho_2(\tilde{h})\right)$,
where $\rho_1(\tilde{h})$ is the estimating functional of the double failure rate $\Lambda_{11}$, and $\rho_2(\tilde{h})$ and $\rho_3(\tilde{h})$ are the estimating functionals of the two single-failure rates $\Lambda_{10}$ and $\Lambda_{01}$, respectively.

Finally, we can now define the estimating functional for the cross-ratio of $S$ over $(0,t_1]\times (0,t_2]$ by 
\begin{align*}
	\Gamma(\rho(\tilde{h}))(t_1,t_2)=\prod_{\substack{
			(0,t_1] \\
			(0,t_2]} }\left(1+\frac{\rho_1(\tilde{h})(ds_1,ds_2)-\rho_2(\tilde{h})(ds_1,s_2)\rho_3(\tilde{h})(s_1,ds_2)}{\left(1-\rho_2(\tilde{h})(\Delta s_1, s_2^-)\right)\left(1-\rho_3(\tilde{h})( s_1^-, \Delta s_2)\right)}\right).
\end{align*}

In summary, the estimating functional of the Dabrowska estimator is
\[
\phi(\tilde{h})=\chi_1(\tilde{h})\chi_2(\tilde{h})\Gamma(\rho(\tilde{h})),
\]
where $\chi_1(\tilde{h})\ \text{and} \ \chi_2(\tilde{h})$ are the univariate Kaplan-Meier estimating functionals, and where $\Gamma(\rho(\tilde{h}))$ is the estimating functional of the cross-ratio term $\prod_{\substack{
		(0,t_1] \\
		(0,t_2]}}\left(1+\hat{L}(du,dv)\right)$.

\subsubsection{The influence function of the Dabrowska estimator}\label{sec:IF_Dab}
Recall that the (first order) influence function is defined by $\dot{\phi}(X)=\phi_F'(\delta_X-F)$. Hence, in order to find the influence function, we need to first find the Fréchet derivative of $\phi$, where
\[
\phi(\tilde{h})=\chi_1(\tilde{h})\chi_2(\tilde{h})\Gamma(\rho(\tilde{h})).
\]

Note that the Kaplan-Meier functional $\chi_j(\tilde{h})$ is differentiable of any order in a neighborhood of $F$ (\citealt[Section 4.1]{overgaard_asymptotic_2017}) as it is a product integral. Similarly, the estimating functional of the cross-ratio term, which is also a product integral, is differentiable of any order. Hence, the product $\phi(\tilde{h})=\chi_1(\tilde{h})\chi_2(\tilde{h})\Gamma(\rho(\tilde{h}))$ is differentiable of any order in a neighborhood of $F$. Specifically, our estimating functional $\phi(\tilde{h})$ is two times (Fréchet) differentiable with a Lipschitz continuous second-order derivative (since by Proposition 1.4 of the supplement of \cite{overgaard_asymptotic_2017}, existence of the (k+1)-order derivative guarantees Lipschitz continuity of the k-order derivative).

By the product rule, 
\[
\phi_{\tilde{h}}'(\tilde{g})=\left[{\chi_1}_{\tilde{h}}'(\tilde{g})\chi_2(\tilde{h})+\chi_1(\tilde{h}){\chi_2}_{\tilde{h}}'(\tilde{g})\right](\Gamma\circ \rho)(\tilde{h})+\chi_1(\tilde{h})\chi_2(\tilde{h})(\Gamma\circ \rho)_{\tilde{h}}'(\tilde{g}),
\]
where by the chain rule, $(\Gamma\circ \rho)_{\tilde{h}}'(\tilde{g})=\Gamma_{\rho(\tilde{h})}'\left(\rho_{\tilde{h}}'(\tilde{g})\right)$. We start by finding $\rho_{\tilde{h}}'(\tilde{g})=\left({\rho_1}_{\tilde{h}}'(\tilde{g}),{\rho_2}_{\tilde{h}}'(\tilde{g}),{\rho_3}_{\tilde{h}}'(\tilde{g})\right)$, where for $j=1,2,3$ and for $\tilde{h},\tilde{g} \in D$,

\[
{\rho_j}_{\tilde{h}}'(\tilde{g})=\lim_{\epsilon\rightarrow 0}\frac{\rho_j\left(\tilde{h}+\epsilon\tilde{g}\right)-\rho_j(\tilde{h})}{\epsilon}.
\]
First, note that
\begin{align*}
	&\rho_2\left(\tilde{h}+\epsilon\tilde{g}\right)-\rho_2(\tilde{h})=\int_{0}^{s}\frac{1}{h(u^-,t)+\epsilon g(u^-,t)}d(h_{10}+\epsilon g_{10})(u,t)-\int_{0}^{s}\frac{1}{h(u^-,t)}dh_{10}(u,t)\\
	=&\int_{0}^{s}\frac{1}{h(u^-,t)+\epsilon g(u^-,t)}h_{10}(du,t)+\int_{0}^{s}\frac{\epsilon}{h(u^-,t)+\epsilon g(u^-,t)}g_{10}(du,t)-\int_{0}^{s}\frac{1}{h(u^-,t)}h_{10}(du,t)\\
	=&\int_{0}^{s}\left[\frac{1}{h(u^-,t)+\epsilon g(u^-,t)}-\frac{1}{h(u^-,t)}\right]h_{10}(du,t)+\int_{0}^{s}\frac{\epsilon}{h(u^-,t)+\epsilon g(u^-,t)}g_{10}(du,t)\\
	=&-\epsilon\int_{0}^{s}\left[\frac{g(u^-,t)}{h(u^-,t)\left(h(u^-,t)+\epsilon g(u^-,t)\right)}\right]h_{10}(du,t)+\epsilon\int_{0}^{s}\frac{1}{h(u^-,t)+\epsilon g(u^-,t)}g_{10}(du,t)
\end{align*}
Hence, if we divide by $\epsilon$ and take the limit as $\epsilon$ approaches zero, we obtain that
\[
{\rho_2}_{\tilde{h}}'(\tilde{g})=\lim_{\epsilon\rightarrow 0}\frac{\rho_2\left(\tilde{h}+\epsilon\tilde{g}\right)-\rho_2(\tilde{h})}{\epsilon}=
\int_{0}^{s}\frac{1}{h(u^-,t)}g_{10}(du,t)-
\int_{0}^{s}\frac{g(u^-,t)}{h^2(u^-,t)}h_{10}(du,t).
\]

Note also that the influence function associated with the Dabrowska estimator of $\Lambda_{10}$ is
\begin{align*}
	\dot{\rho}_2(X)\equiv {\rho_2}_{F}'(\delta_X-F)&=\int_{0}^{s}\frac{1}{H(u^-,t)}d(N_{10}-H_{10})(u,t)-
	\int_{0}^{s}\frac{Y(u^-,t)-H(u^-,t)}{H^2(u^-,t)}H_{10}(du,t)\\
	&=\int_{0}^{s}\frac{1}{H(u^-,t)}N_{10}(du,t)-\int_{0}^{s}\frac{1}{H(u^-,t)}H_{10}(du,t)\\
	&-
	\int_{0}^{s}\frac{Y(u^-,t)}{H^2(u^-,t)}H_{10}(du,t)+\int_{0}^{s}\frac{H(u^-,t)}{H^2(u^-,t)}H_{10}(du,t)\\
	&=\int_{0}^{s}\frac{1}{H(u^-,t)}N_{10}(du,t)-\int_{0}^{s}\frac{Y(u^-,t)}{H^2(u^-,t)}H_{10}(du,t),
\end{align*}
where $X=(Y_1,Y_2,\Delta_1,\Delta_2)$ denotes the observed data without the covariates, and where $\delta_X=(Y,N_{10},N_{01},N_{11})$ and $F=(H,H_{10},H_{01},H_{11})$. Recall also that $\int_{0}^{s}\frac{1}{H^(u^-,t)}H_{10}(du,t)=\Lambda_{10}(s,t)$ and hence $$\int_{0}^{s}\frac{Y(u^-,t)}{H^2(u^-,t)}H_{10}(du,t)=\int_{0}^{s}\frac{Y(u^-,t)}{H(u^-,t)}\Lambda_{10}(du,t).$$
Consequently,
\[
\dot{\rho}_2(X)=\int_{0}^{s}\frac{1}{H(u^-,t)}N_{10}(du,t)-\int_{0}^{s}\frac{Y(u^-,t)}{H(u^-,t)}\Lambda_{10}(du,t)=\int_{0}^{s}\frac{1}{H(u^-,t)}M_{10}(du,t),
\]
where $M_{10}(s,t)=N_{10}(s,t)-\int_{0}^{s}Y(u^-,t)\Lambda_{10}(du,t)$.

A similar analysis for $\rho_3$ gives
\[
{\rho_3}_{\tilde{h}}'(\tilde{g})=\lim_{\epsilon\rightarrow 0}\frac{\rho_3\left(\tilde{h}+\epsilon\tilde{g}\right)-\rho_3(\tilde{h})}{\epsilon}=
\int_{0}^{t}\frac{1}{h(s,v^-)}g_{01}(s,dv)-
\int_{0}^{t}\frac{g(s,v^-)}{h^2(s,v^-)}h_{01}(s,dv),
\]
and the corresponding influence function for $\Lambda_{01}$ is
\[
\dot{\rho}_3(X)=\int_{0}^{t}\frac{1}{H(s,v^-)}M_{01}(s,dv),
\]
where $M_{01}(s,t)=N_{01}(s,t)-\int_{0}^{t}Y(s,v^-)\Lambda_{01}(s,dv)$.

As for $\rho_1$,
\[
{\rho_1}_{\tilde{h}}'(\tilde{g})=\lim_{\epsilon\rightarrow 0}\frac{\rho_1\left(\tilde{h}+\epsilon\tilde{g}\right)-\rho_1(\tilde{h})}{\epsilon}=
\int_0^s\int_{0}^{t}\frac{1}{h(u^-,v^-)}g_{11}(du,dv)-
\int_0^s\int_{0}^{t}\frac{g(u^-,v^-)}{h^2(u^-,v^-)}h_{11}(du,dv),
\]
and the corresponding influence function for the Dabrowska estimator of $\Lambda_{11}$ is
\[
\dot{\rho}_1(X)=\int_0^s\int_{0}^{t}\frac{1}{H(u^-,v^-)}M_{11}(du,dv),
\]
where  $M_{11}(s,t)=N_{11}(s,t)-\int_0^s\int_{0}^{t}Y(u^-,v^-)\Lambda_{11}(du,dv)$.

Next, we will now find the derivative of the functional composition $(\Gamma\circ \rho)_{\tilde{h}}'(\tilde{g})=\Gamma_{\rho(\tilde{h})}'\left(\rho_{\tilde{h}}'(\tilde{g})\right)$. Note that $\rho_{\tilde{h}}'(\tilde{g})=\left({\rho_1}_{\tilde{h}}'(\tilde{g}),{\rho_2}_{\tilde{h}}'(\tilde{g}),{\rho_3}_{\tilde{h}}'(\tilde{g})\right)$ is a vector where each coordinate is a bivariate function. We will denote the space of such vectors by $R^3$. Note also that $\rho(\tilde{h})=(\rho_1(\tilde{h}),\rho_2(\tilde{h}), \rho_3(\tilde{h}))$ is also a member of $R^3$. Let $r,c\in R^3$ be such vectors, that is $r=(r_1(\cdot,\cdot),r_2(\cdot,\cdot),r_3(\cdot,\cdot))$ and $c=(c_1(\cdot,\cdot),c_2(\cdot,\cdot),c_3(\cdot,\cdot))$. Define the functional $\nu:R^3\rightarrow\Theta$ by
\[
\nu(r)(t_1,t_2)=\int_0^{t_1}\int_0^{t_2}\frac{r_1(ds_1,ds_2)-r_2(ds_1,s_2)r_3(s_1,ds_2)}{[1-r_2(\Delta s_1,s_2)][1-r_3(s_1,\Delta s_2)]},
\]
and note that \[
\Gamma(r)=\prod_{\substack{
		(0,t_1] \\
		(0,t_2]}}\left(1+\nu(r)(ds_1,ds_2)\right)=\exp\left\{\int_{0}^{t_1}\int_{0}^{t_2}\nu(r)(ds_1,ds_2)\right\}=\exp\left\{\nu(r)(t_1,t_2)\right\}
\]
is the bivariate product integral. Consequently, its Fréchet derivative is
\[
\Gamma_r'(c)=\Gamma(r)\nu_r'(c).
\]
Hence, we need to find the Fréchett derivative of $\nu$. As in \cite{gill_inefficient_1995}, we will assume for simplicity that the bivariate times are continuous, in which case we have that
\[
\nu(r)(t_1,t_2)=\int_0^{t_1}\int_0^{t_2}r_1(ds_1,ds_2)-r_2(ds_1,s_2)r_3(s_1,ds_2),
\]
and 
\[
\nu(r)(ds_1,ds_2)=r_1(ds_1,ds_2)-r_2(ds_1,s_2)r_3(s_1,ds_2).
\]
Hence,
\begin{align*}
	\nu(r+\epsilon c)-\nu(r)=\epsilon\left[c_1(ds_1,ds_2)-c_2(ds_1,s_2)r_3(s_1,ds_2)-r_2(ds_1,s_2)c_3(s_1,ds_2)-\epsilon c_2(ds_1,s_2)c_3(s_1,ds_2)\right],
\end{align*}
and after dividing by $\epsilon$ and taking the limit as $\epsilon\rightarrow 0$, we obtain that
\[
\nu_r'(c)=\lim_{\epsilon\rightarrow 0}\frac{\nu(r+\epsilon c)-\nu(r)}{\epsilon}=c_1(ds_1,ds_2)-c_2(ds_1,s_2)r_3(s_1,ds_2)-r_2(ds_1,s_2)c_3(s_1,ds_2).
\]

Specifically, for $r=\rho(\tilde{h})$ and $c=\rho_{\tilde{h}}'(\tilde{g})$, we obtain that
\begin{align*}
	\nu_{\rho(\tilde{h})}'(\rho_{\tilde{h}}'(\tilde{g}))&=\lim_{\epsilon\rightarrow 0}\frac{\nu(\rho(\tilde{h})+\epsilon \rho_{\tilde{h}}'(\tilde{g}))-\nu(\rho(\tilde{h}))}{\epsilon}\\
	&={\rho_1}_{\tilde{h}}'(\tilde{g})(ds_1,ds_2)-{\rho_2}_{\tilde{h}}'(\tilde{g})(ds_1,s_2)\rho_3(\tilde{h})(s_1,ds_2)-\rho_2(\tilde{h})(ds_1,s_2){\rho_3}_{\tilde{h}}'(\tilde{g})(s_1,ds_2)\\
	&=\frac{1}{h(s_1^-,s_2^-)}g_{11}(ds_1,ds_2)-
	\frac{g(s_1^-,s_2^-)}{h^2(s_1^-,s_2^-)}h_{11}(ds_1,ds_2)\\
	&-\left(\frac{1}{h(s_1^-,s_2)}g_{10}(ds_1,s_2)-
	\frac{g(s_1^-,s_2)}{h^2(s_1^-,s_2)}h_{10}(ds_1,s_2)\right)\left(\frac{h_{01}(s_1,ds_2)}{h(s_1,s_2^-)}\right)\\
	&-\left(\frac{h_{10}(ds_1,s_2)}{h(s_1^-,s_2)}\right)\left(\frac{1}{h(s_1,s_2^-)}g_{01}(s_1,ds_2)-
	\frac{g(s_1,s_2^-)}{h^2(s_1,s_2^-)}h_{10}(s_1,ds_2)\right).
\end{align*}

Finally, for $\tilde{h}=F=(H,H_{10},H_{01},H_{11})$, $\tilde{g}=\delta_X-F=(Y-H,N_{10}-H_{10},N_{01}-H_{01},N_{11}-H_{11})$,  $M_{10}(s,t)=N_{10}(s,t)-\int_{0}^{s}Y(u^-,t)\Lambda_{10}(du,t)$, $M_{01}(s,t)=N_{01}(s,t)-\int_{0}^{t}Y(s,v^-)\Lambda_{01}(u,dv)$,  and $M_{11}(s,t)=N_{11}(s,t)-\int_0^s\int_{0}^{t}Y(u^-,v^-)\Lambda_{11}(du,dv)$, we have that

\begin{align*}
	\nu_{\rho(F)}'(\rho_{F}'(\delta_X-F))=\nu_{\vec{\Lambda}}'(IF_1,IF_2,IF_3),
\end{align*}
where $\rho(F)=\vec{\Lambda}=(\Lambda_{11},\Lambda_{10},\Lambda_{01})$ is the true bivariate hazard vector, and where 
\begin{align*}
	IF_1(s,t)&={\rho_1}_{F}'(\delta_X-F)(s,t)=\dot{\rho}_1(X)(s,t)=\int_0^{s}\int_0^{t}\frac{1}{H(u^-,v^-)}M_{11}(du,dv),\\
	IF_2(s,t)&={\rho_2}_{F}'(\delta_X-F)(s,t)=\dot{\rho}_2(X)(s,t)=\int_0^{s}\frac{1}{H(u^-,t)}M_{10}(du,t),\\
	IF_3(s,t)&={\rho_3}_{F}'(\delta_X-F)(s,t)=\dot{\rho}_3(X)(s,t)=\int_0^{t}\frac{1}{H(s,v^-)}M_{01}(s,dv),
\end{align*}
are the influence functions (IF) of the estimators of the three bivariate hazards $\hat{\Lambda}_{11}$, $\hat{\Lambda}_{10}$,  and $\hat{\Lambda}_{01}$.
That is,
\begin{align*}
	&\nu_{\rho(F)}'(\rho_{F}'(\delta_X-F))(s,t)\\
	=&\int_0^{s}\int_0^{t}\frac{1}{H(s_1^-,s_2^-)}M_{11}(ds_1,ds_2)
	-\frac{1}{H(s_1^-,s_2)}M_{10}(ds_1,s_2)\Lambda_{01}(s_1,ds_2)-\frac{1}{H(s_1,s_2^-)}M_{01}(s_1,ds_2)\Lambda_{10}(ds_1,s_2).
\end{align*}
Summing it up, we obtain that,
\begin{align*}
	&\Gamma_{\rho(F)}'(\rho_{F}'(\delta_X-F))=\Gamma(\rho(F))\nu_{\rho(F)}'(\rho_{F}'(\delta_X-F))=\text{C-R}\cdot \nu_{\rho(F)}'(\rho_{F}'(\delta_X-F))\\
	=&\text{C-R} \cdot \left(\int_{0}^{(\cdot)}\int_{0}^{(\cdot)}\frac{M_{11}(ds_1,ds_2)}{H(s_1^-,s_2^-)}
	-\frac{1}{H(s_1^-,s_2)}M_{10}(ds_1,s_2)\Lambda_{01}(s_1,ds_2)-\frac{1}{H(s_1,s_2^-)}M_{01}(s_1,ds_2)\Lambda_{10}(ds_1,s_2)\right),
\end{align*}
where $\text{C-R}$ is the true cross ratio term
\[
\text{C-R}=\Gamma(\rho(F))=\frac{S(\cdot,\cdot)}{S_{T_1}(\cdot)S_{T_2}(\cdot)},
\]
and where we write from now on $S(t_1,t_2)$ instead of $S_{T_1,T_2}(t_1,t_2)$.

Next, note that as in \cite{overgaard_asymptotic_2017}, the influence functions of the two Kaplan-Meier estimators are given by
\begin{align*}
	IF(KM_1)&={\chi_1}_{F}'(\delta_X-F)=-\chi_1(F){\psi_1}_F'(\delta_X-F)=-S_{T_1}(\cdot)IF(\Lambda_{1}(\cdot)),\\	IF(KM_2)&={\chi_2}_{F}'(\delta_X-F)=-\chi_2(F){\psi_2}_F'(\delta_X-F)=-S_{T_2}(\cdot)IF(\Lambda_{2}(\cdot)),
\end{align*}
where $IF(\Lambda_{1}(\cdot))=IF(\Lambda_{10}(\cdot,0))=\dot{\rho}_2(X)(\cdot,0)=\int_{0}^{(\cdot)}\frac{1}{H(s_1^-,0)}M_{10}(ds_1,0)$, and $IF(\Lambda_{2}(\cdot))=IF(\Lambda_{01}(0,\cdot))=\dot{\rho}_3(X)(0,\cdot)=\int_{0}^{(\cdot)}\frac{1}{H(0,s_2^-)}M_{01}(0,ds_2)$.

Hence, we can write
\begin{align*}
	\phi_{\tilde{h}}'(\tilde{g})&=\left[{\chi_1}_{\tilde{h}}'(\tilde{g})\chi_2(\tilde{h})+\chi_1(\tilde{h}){\chi_2}_{\tilde{h}}'(\tilde{g})\right](\Gamma\circ \rho)(\tilde{h})+\chi_1(\tilde{h})\chi_2(\tilde{h})(\Gamma\circ \rho)_{\tilde{h}}'(\tilde{g})\\
	&=\left[-\chi_1(\tilde{h}){\psi_1}_{\tilde{h}}'(\tilde{g})\chi_2(\tilde{h})-\chi_1(\tilde{h})\chi_2(\tilde{h}){\psi_2}_{\tilde{h}}'(\tilde{g})\right](\Gamma\circ \rho)(\tilde{h})+\chi_1(\tilde{h})\chi_2(\tilde{h})(\Gamma\circ \rho)_{\tilde{h}}'(\tilde{g})\\
	&=\left[-\chi_1(\tilde{h}){\psi_1}_{\tilde{h}}'(\tilde{g})\chi_2(\tilde{h})-\chi_1(\tilde{h})\chi_2(\tilde{h}){\psi_2}_{\tilde{h}}'(\tilde{g})\right](\Gamma\circ \rho)(\tilde{h})+\chi_1(\tilde{h})\chi_2(\tilde{h})\Gamma_{\rho(\tilde{h})}'\left(\rho_{\tilde{h}}'(\tilde{g})\right),
\end{align*}
where in the last equality we used the chain rule $(\Gamma\circ \rho)_{\tilde{h}}'(\tilde{g})=\Gamma_{\rho(\tilde{h})}'\left(\rho_{\tilde{h}}'(\tilde{g})\right)$.

Thus,
\begin{align*}
	\dot{\phi}(X)&=\phi_{F}'(\delta_X-F)\\
	&=\left[-\chi_1(F){\psi_1}_{F}'(\delta_X-F)\chi_2(F)-\chi_1(F)\chi_2(F){\psi_2}_{F}'(\delta_X-F)\right](\Gamma\circ \rho)(F)+\chi_1(F)\chi_2(F)\Gamma_{\rho(F)}'\left(\rho_{F}'(\delta_X-F)\right)\\
	&=\left[-S_{T_1}(\cdot) IF(\Lambda_{1}(\cdot))S_{T_2}(\cdot)-S_{T_1}(\cdot)S_{T_2}(\cdot)IF(\Lambda_{2}(\cdot))\right]\text{C-R}+S_{T_1}(\cdot)S_{T_2}(\cdot)\text{C-R}\cdot \nu_{\rho(F)}'(\rho_{F}'(\delta_X-F))\\
	&=-S_{T_1}(\cdot)S_{T_2}(\cdot)\text{C-R}\left[IF(\Lambda_{1}(\cdot))+IF(\Lambda_{1}(\cdot))-\nu_{\rho(F)}'(\rho_{F}'(\delta_X-F))\right],
\end{align*}
where $\nu_{\rho(F)}'(\rho_{F}'(\delta_X-F))=\int_{0}^{(\cdot)}\int_{0}^{(\cdot)}\frac{M_{11}(ds_1,ds_2)}{H(s_1^-,s_2^-)}
-\frac{M_{10}(ds_1,s_2)}{H(s_1^-,s_2)}\Lambda_{01}(s_1,ds_2)-\frac{M_{01}(s_1,ds_2)}{H(s_1,s_2^-)}\Lambda_{10}(ds_1,s_2)$.
Note also that $S(s,t)=S_{T_1}(s)S_{T_2}(t)\text{C-R}(s,t)$. Consequently,

\begin{align*}
	&	\dot{\phi}(X)(s,t)=\phi_{F}'(\delta_X-F)(s,t)=-S(s,t)\left[IF(\Lambda_{1}(s))+IF(\Lambda_{2}(t))-\nu_{\rho(F)}'(\rho_{F}'(\delta_X-F))(s,t)\right]\\
	=&-S(s,t)\left[\int_{0}^{s}\frac{M_{1}(du)}{H_1(u)}+\int_{0}^{t}\frac{M_{2}(dv)}{H_2(v)}-\int_{0}^s\int_{0}^t\frac{M_{11}(du,dv)-M_{10}(du,v)\Lambda_{01}(u,dv)-M_{01}(u,dv)\Lambda_{10}(du,v)}{H(u^-,v^-)}\right],
\end{align*}
which is equivalent to the influence function given by \cite{gill_inefficient_1995}.
\subsubsection{The conditional expectation of the Dabrowska influence function}\label{sec:cond_expec_Dab}

Our goal here is to show that $E[\dot{\phi}(X)\mid Z]=S(t_1,t_2\mid Z)-S(t_1,t_2)$.
The influence function of the Dabrowska estimator is 
\begin{equation}
	\begin{aligned}
		&\dot{\phi}(X)(s,t)\\
		=&-S(s,t)\left[\int_{0}^{s}\frac{M_{1}(du)}{H_1(u)}+\int_{0}^{t}\frac{M_{2}(dv)}{H_2(v)}-\int_{0}^s\int_{0}^t\frac{M_{11}(du,dv)-M_{10}(du,v)\Lambda_{01}(u,dv)-M_{01}(u,dv)\Lambda_{10}(du,v)}{H(u^-,v^-)}\right],
	\end{aligned}
\end{equation}
where $H_1(u)=P(Y_1\geq u)=P(Y_1> u^-)=P(Y_1> u^-,Y_2> 0)=H(u^-,0)$, $H_2(v)=H(0,v^-)$, 
$M_1(s)=N_{1}(s)-\int_{0}^{s}1(Y_1\geq u)\Lambda_{1}(du)=N_{10}(s,0)-\int_{0}^{s}Y(u^-,0)\Lambda_{10}(du,0)=M_{10}(s,0)$, and similarly $M_2(t)=M_{01}(0,t)$. 

Lemma~\ref{lem:1} and Proposition~\ref{prop:1} below are due to \cite{jacobsen_note_2016}. For completeness, we also present here their proofs. Lemma~\ref{lem:2d} and Lemma~\ref{prop:2} are the bivariate extensions of Lemma~\ref{lem:1}, and Proposition~\ref{prop:3} is the bivariate extension of Proposition~\ref{prop:1}.

\begin{lem}[\cite{jacobsen_note_2016}]\label{lem:1}
	For $j=1,2$, denote by $\delta_j^Z(\cdot)=\lambda_j(\cdot \mid Z)-\lambda_j(\cdot)$, the difference between the conditional and marginal univariate hazards, by $y_j(u)=1(Y_j\geq u)$ an indicator function, by $H_j^Z(\cdot)=P(Y_j\geq \cdot\mid Z)$ and $S_j^Z(\cdot)=P(T_j\geq \cdot\mid Z)$ conditional survival probabilities, and by $M_j^Z(du)=1(Y_j \in du, \Delta_j=1)-I(Y_1\geq u)\Lambda_j(du\mid Z)$ the conditional martingale. Then, 
	\begin{enumerate}
		\item $M_j(du)=M_j^Z(du)+\delta_j^Z(u)y_j(u)du$
		\item Denote $q_j^Z(\cdot)=\frac{H_j^Z(\cdot)}{H_j(\cdot)}\delta_j^Z(\cdot)$. Then,
		\[
		\frac{\partial }{\partial t}\left(\frac{H_j^Z(t)}{H_j(t)}\right)=\frac{\partial }{\partial t}\left(\frac{S_j^Z(t)}{S_j(t)}\right)=-q_j^Z(t).
		\]
	\end{enumerate}
\end{lem}

\begin{proof}[Proof of Lemma~\ref{lem:1}]
	\begin{enumerate}
		\item Recall that $M_j^Z(du)=1(Y_j \in du, \Delta_j=1)-I(Y_1\geq u)\Lambda_j(du\mid Z)$.  Note also that
		\[
		\delta_j^Z(u)y_j(u)du=1(Y_j\geq u)\left[\lambda_j(u \mid Z)-\lambda_j(u)\right]du=1(Y_j\geq u)\lambda_j(u \mid Z)du-\lambda_j(u)1(Y_j\geq u)du.
		\]
		Hence,
		\begin{align*}
			M_j^Z(du)+\delta_j^Z(u)y_j(u)du&=1(Y_j \in du, \Delta_j=1)-I(Y_j\geq u)\Lambda_j(du\mid Z)\\
			&+1(Y_j\geq u)\lambda_j(u \mid Z)du-\lambda_j(u)1(Y_j\geq u)du\\
			&=1(Y_j \in du, \Delta_j=1)-\lambda_j(u)1(Y_j\geq u)du\equiv M_j(du),
		\end{align*}
		where $\Lambda_j(du)=\lambda_j(u)du$.
		\item First note that
		\[
		\frac{H_j^Z(t)}{H_j(t)}=\frac{P(Y_j\geq t \mid Z )}{P(Y_j\geq t)}=\frac{P(\min(T_j,C_j)\geq t \mid Z)}{P(\min(T_j,C_j)\geq t )}=\frac{P(T_j\geq t \mid Z)P(C_j\geq t)}{P(T_j\geq t )P(C_j\geq t)}=\frac{S_j^Z(t)}{S_j(t)}
		\]
		where the equality before last is due to the independent censoring assumption. Consequently, the fractions $\frac{H_j^Z(t)}{H_j(t)}$ and $\frac{S_j^Z(t)}{S_j(t)}$ have the same derivative. The derivative is
		\begin{align*}
			\frac{\partial }{\partial t}\left(\frac{S_j^Z(t)}{S_j(t)}\right)&=\frac{-f_j(t\mid Z)}{S_j(t)}-\frac{S_j(t\mid Z)}{(S_j(t))^2}(-f_j(t))=\frac{1}{S_j(t)}\left[\frac{f_j(t)}{S_j(t)}S_j(t\mid Z)-f_j(t\mid Z)\right]\\
			&=\frac{1}{S_j(t)}\left[\lambda_j(t)S_j(t\mid Z)-f_j(t\mid Z)\right]=\frac{S_j(t\mid Z)}{S_j(t)}\left[\lambda_j(t)-\frac{f_j(t\mid Z)}{S_j(t\mid Z)}\right]\\
			&=\frac{S_j(t\mid Z)}{S_j(t)}\left[\lambda_j(t)-\lambda_j(t\mid Z)\right]= -\frac{S_j(t\mid Z)}{S_j(t)}[\lambda_j(t\mid Z)-\lambda_j(t)]\\
			&=-\frac{H_j(t\mid Z)}{H_j(t)}[\lambda_j(t\mid Z)-\lambda_j(t)]\equiv-q_j^Z(t).
		\end{align*}
	\end{enumerate}
\end{proof}

\begin{prop}[\cite{jacobsen_note_2016}]\label{prop:1}
	For $j=1,2$,
	\[
	E\left[\int_0^{t_0}\frac{M_j(du)}{H_j(u)}\mid Z\right]=1-\frac{S_j(t_0\mid Z)}{S_j(t_0)}.
	\]
\end{prop}

\begin{proof}[Proof of Proposition~\ref{prop:1}]
	By Lemma~\ref{lem:1} we have that $M_j(du)=M_j^Z(du)+\delta_j^Z(u)y_j(u)du$. Consequently, 
	\[
	E\left[\int_0^{t_0}\frac{M_j(du)}{H_j(u)}\mid Z\right]=\underbrace{E\left[\int_0^{t_0}\frac{M_j^Z(du)}{H_j(u)}\mid Z \right]}_A+\underbrace{E\left[\int_0^{t_0}\frac{\delta_j^Z(u)y_j(u)}{H_j(u)}du \mid Z\right]}_B.	
	\]
	We have that $A=0$ since $M_j^Z(u)$ is a martingale.
	Note also that
	\begin{align*}
		B&=E\left[\int_0^{t_0}\frac{\delta_j^Z(u)y_j(u)}{H_j(u)}du \mid Z\right]=\int_0^{t_0}\frac{\delta_j^Z(u)P(Y_j\geq u \mid Z)}{H_j(u)}du\\
		&=\int_0^{t_0}\frac{H_j^Z(u)\delta_j^Z(u)}{H_j(u)}du=\int_0^{t_0}q_j^Z(u)du=-\int_0^{t_0}-q_j^Z(u)du.
	\end{align*}
	Next, by Lemma~\ref{lem:1} (part 2), we have that $-q_j^Z(t)=\frac{\partial }{\partial t}\left(\frac{S_j^Z(t)}{S_j(t)}\right)=\frac{\partial }{\partial t}\left(\frac{H_j^Z(t)}{H_j(t)}\right)$. Hence,
	\begin{align*}
		B&=-\int_0^{t_0}-q_j^Z(u)du=-\int_0^{t_0}\frac{\partial }{\partial u}\left(\frac{S_j^Z(u)}{S_j(u)}\right)du=-\frac{S_j^Z(u)}{S_j(u)}\bigg\rvert_0^{t_0}=1-\frac{S_j^Z(t_0)}{S_j(t_0)}.
	\end{align*}
	Summing it up, we obtain that
	\[
	E\left[\int_0^{t_0}\frac{M_j(du)}{H_j(u)}\mid Z\right]=\underbrace{E\left[\int_0^{t_0}\frac{M_j^Z(du)}{H_j(u)}\mid Z \right]}_A+\underbrace{E\left[\int_0^{t_0}\frac{\delta_j^Z(u)y_1(u)}{H_j(u)}du \mid Z\right]}_B=1-\frac{S_j^Z(t_0)}{S_j(t_0)}.	
	\]
\end{proof}

\begin{lem}\label{lem:2d}
	Denote by $\delta_{10}^Z(u,v)=\Lambda_{10}(du,v\mid Z)-\Lambda_{10}(du,v)$, $\delta_{01}^Z(u,v)=\Lambda_{01}(u,dv\mid Z)-\Lambda_{01}(u,dv)$, and $\delta_{11}^Z(u,v)=\Lambda_{11}(du,dv\mid Z)-\Lambda_{11}(du,dv)$ the difference between the conditional and marginal bivariate hazards. Then,
	\begin{enumerate}
		\item $M_{10}(du,v)=M_{10}^Z(du,v)+\delta_{10}^Z(u,v)Y(u^-,v)du$
		\item $M_{01}(u,dv)=M_{01}^Z(u,dv)+\delta_{01}^Z(u,v)Y(u,v^-)dv$
		\item $M_{11}(du,dv)=M_{11}^Z(du,dv)+\delta_{11}^Z(u,v)Y(u^-,v^-)dudv$
	\end{enumerate}
\end{lem}
The proof is similar to the proof of Lemma~\ref{lem:1}.
\begin{lem}\label{prop:2}
	Denote by $H^Z(u,v)=P(Y_1>u,Y_2>v \mid Z)$ and $S^Z(u,v)=P(T_1>u,T_2>v \mid Z)$ the conditional joint survival functions of the observed times $(Y_1,Y_2)$, and the actual failure-times $(T_1,T_2)$, respectively. Denote further $q_{11}^Z(u,v)=\frac{H^Z(u,v)}{H(u,v)}\delta_{11}^Z(u,v)$, $q_{10}^Z(u,v)=\frac{H^Z(u,v)}{H(u,v)}\delta_{10}^Z(u,v)$, and $q_{01}^Z(u,v)=\frac{H^Z(u,v)}{H(u,v)}\delta_{01}^Z(u,v)$. Then,
	\[
	\frac{\partial^2 }{\partial s\partial t}\left(\frac{H^Z(s,t)}{H(s,t)}\right)=\frac{\partial^2 }{\partial s\partial t}\left(\frac{S^Z(s,t)}{S(s,t)}\right)=q_{11}^Z(s,t)-q_{10}^Z(u,v)\lambda_{01}(s,t)-q_{01}^Z(u,v)\lambda_{10}(s,t)
	\]
\end{lem}

\begin{proof}[Proof of Lemma~\ref{prop:2}]
	First note that
	\begin{align*}
		\frac{H^Z(s,t)}{H(s,t)}&=\frac{P(Y_1\geq s,Y_2 \geq t \mid Z )}{P(Y_1\geq s,Y_2 \geq t)}=\frac{P(\min(T_1,C_1)\geq s, \min(T_2,C_2)\geq t\mid Z)}{P(\min(T_1,C_1)\geq s, \min(T_2,C_2)\geq t)}\\
		&=\frac{P(T_1\geq t, T_2\geq t \mid Z)P(C_1\geq t, C_2\geq t)}{P(T_1\geq t, T_2\geq t)P(C_1\geq t, C_2\geq t)}=\frac{S^Z(s,t)}{S(s,t)},
	\end{align*}
	where again, the equality before last is due to the independent censoring assumption. Consequently, 
	\[
	\frac{\partial^2 }{\partial s\partial t}\left(\frac{H^Z(s,t)}{H(s,t)}\right)=\frac{\partial^2 }{\partial s\partial t}\left(\frac{S^Z(s,t)}{S(s,t)}\right). 
	\]
	It can be shown that the cross partial derivative is equal to
	\begin{align*}
		&\frac{\partial^2 }{\partial s\partial t}\left(\frac{S^Z(s,t)}{S(s,t)}\right)=\frac{f(s,t \mid Z)}{S(s,t)}-\frac{\frac{\partial S^Z(s,t)}{\partial s}}{S(s,t)}\frac{\frac{\partial S(s,t)}{\partial t}}{S(s,t)}-\frac{\frac{\partial}{\partial s}S(s,t)\frac{\partial}{\partial t}S^z(s,t)}{S^2(s,t)}+\frac{S^Z(s,t)\frac{\partial}{\partial s}S(s,t)\frac{\partial}{\partial t}S(s,t)}{S^3(s,t)}\\
		-&\frac{S^Z(s,t)f(s,t)}{S^2(s,t)}+\frac{s^Z(s,t)\frac{\partial}{\partial s}S(s,t)\frac{\partial}{\partial t}S(s,t)}{S^3(s,t)}\\
		=&\frac{f(s,t \mid Z)}{S(s,t)}-\frac{\frac{\partial S^Z(s,t)}{\partial s}}{S(s,t)}\frac{\frac{\partial S(s,t)}{\partial t}}{S(s,t)}-\frac{\frac{\partial}{\partial s}S(s,t)\frac{\partial}{\partial t}S^z(s,t)}{S^2(s,t)}-\frac{S^Z(s,t)f(s,t)}{S^2(s,t)}+2\frac{S^Z(s,t)\frac{\partial}{\partial s}S(s,t)\frac{\partial}{\partial t}S(s,t)}{S^3(s,t)}\\
		=&\frac{S^Z(s,t)}{S(s,t)}\left[\frac{f(s,t \mid Z)}{S^Z(s,t)}-\frac{\frac{\partial S^Z(s,t)}{\partial s}}{S^Z(s,t)}\frac{\frac{\partial S(s,t)}{\partial t}}{S(s,t)}-\frac{\frac{\partial}{\partial s}S(s,t)}{S(s,t)}\frac{\frac{\partial}{\partial t}S^z(s,t)}{S^z(s,t)}-\frac{f(s,t)}{S(s,t)}+2\frac{\frac{\partial}{\partial s}S(s,t)\frac{\partial}{\partial t}S(s,t)}{S^2(s,t)}\right]\\
		=&\frac{S^Z(s,t)}{S(s,t)}\left[\lambda_{11}(s,t \mid Z)-\lambda_{10}(s,t\mid Z)\lambda_{01}(s,t)-\lambda_{10}(s,t)\lambda_{01}(s,t\mid Z)-\lambda_{11}(s,t)+2\lambda_{10}(s,t)\lambda_{01}(s,t)\right]\\
		=&\frac{S^Z(s,t)}{S(s,t)}\left[\lambda_{11}(s,t \mid Z)-\lambda_{11}(s,t)\right]-\frac{S^Z(s,t)\lambda_{01}(s,t)}{S(s,t)}\left[\lambda_{10}(s,t \mid Z)-\lambda_{10}(s,t)\right]\\
		-&\frac{S^Z(s,t)\lambda_{10}(s,t)}{S(s,t)}\left[\lambda_{01}(s,t \mid Z)-\lambda_{01}(s,t)\right]\\
		=&\frac{S^Z(s,t)}{S(s,t)}\left[\delta_{11}^Z(s,t)-\delta_{10}^Z(s,t)\lambda_{01}(s,t)-\delta_{01}^Z(s,t)\lambda_{10}(s,t)\right]\\
		=&\frac{H^Z(s,t)}{H(s,t)}\left[\delta_{11}^Z(s,t)-\delta_{10}^Z(s,t)\lambda_{01}(s,t)-\delta_{01}^Z(s,t)\lambda_{10}(s,t)\right]\\
		=&q_{11}^Z(s,t)-q_{10}^Z(u,v)\lambda_{01}(s,t)-q_{01}^Z(u,v)\lambda_{10}(s,t).
	\end{align*}
\end{proof}
\begin{prop}\label{prop:3}
	\begin{enumerate}
		\item \[
		E\left[\int_{0}^s\int_{0}^t\frac{M_{11}(du,dv)}{H(u,v)} \mid Z\right]=\int_{0}^s\int_{0}^tq_{11}^Z(u,v)dudv
		\]
		\item \[
		E\left[\int_{0}^s\int_{0}^t\frac{M_{10}(du,v)\Lambda_{01}(u,dv)}{H(u,v)} \mid Z\right]=\int_{0}^s\int_{0}^tq_{10}^Z(u,v)\lambda_{01}(u,v)dudv
		\]
		\item \[
		E\left[\int_{0}^s\int_{0}^t\frac{M_{01}(u,dv)\Lambda_{10}(du,v)}{H(u,v)} \mid Z\right]=\int_{0}^s\int_{0}^tq_{01}^Z(u,v)\lambda_{10}(u,v)dudv
		\]
	\end{enumerate}
\end{prop}
\begin{proof}[Proof of Proposition~\ref{prop:3}]
	Based on Lemma~\ref{lem:2d}, we have that
	\begin{enumerate}
		\item 	\begin{align*}
			E\left[\int_{0}^s\int_{0}^t\frac{M_{11}(du,dv)}{H(u,v)} \mid Z\right]&=\underbrace{E\left[\int_{0}^s\int_{0}^t\frac{M_{11}^Z(du,dv)}{H(u,v)} \mid Z\right]}_a+\underbrace{E\left[\int_{0}^s\int_{0}^t\frac{\delta_{11}^Z(u,v)Y(u^-,v^-)}{H(u,v)}dudv \mid Z\right]}_b,
		\end{align*}
		where $M_{11}^Z(du,dv)=1(Y_1\in du, Y_2 \in dv, \Delta_1=1,\Delta_2=1)-1(Y_1\geq u, Y_2 \geq v)\Lambda_{11}(du,dv\mid Z)$. Hence, $a=0$ since $M_{11}^Z(u,v)$ is a martingale.
		Note also that 
		\begin{align*}
			b&=E\left[\int_{0}^s\int_{0}^t\frac{\delta_{11}^Z(u,v)Y(u^-,v^-)}{H(u,v)}dudv \mid Z\right]=\int_{0}^s\int_{0}^t\frac{\delta_{11}^Z(u,v)P(Y_1\geq u, Y_2 \geq v \mid Z)}{H(u,v)}dudv\\
			&=\int_{0}^s\int_{0}^t\frac{H^Z(u,v)\delta_{11}^Z(u,v)}{H(u,v)}dudv=\int_{0}^s\int_{0}^tq_{11}^Z(u,v)dudv.
		\end{align*}
		Consequently, 
		\[
		E\left[\int_{0}^s\int_{0}^t\frac{M_{11}(du,dv)}{H(u,v)} \mid Z\right]=a+b=0+b=\int_{0}^s\int_{0}^tq_{11}^Z(u,v)dudv.
		\]
		\item \begin{align*}
			&E\left[\int_{0}^s\int_{0}^t\frac{M_{10}(du,v)\Lambda_{01}(u,dv)}{H(u,v)} \mid Z\right]\\
			=&\underbrace{E\left[\int_{0}^s\int_{0}^t\frac{M_{10}^Z(du,v)\Lambda_{01}(u,dv)}{H(u,v)} \mid Z\right]}_a+\underbrace{E\left[\int_{0}^s\int_{0}^t\frac{\delta_{10}^Z(u,v)Y(u^-,v^-)\Lambda_{01}(u,dv)}{H(u,v)}dudv \mid Z\right]}_b,
		\end{align*}
		where $M_{10}^Z(du,v)=1(Y_1\in du, Y_2 \geq v, \Delta_1=1)-1(Y_1\geq u, Y_2 \geq v)\Lambda_{10}(du,v\mid Z)$. Hence, $a=0$ since $M_{10}^Z(u,v)$ is a martingale.
		Note also that 
		\begin{align*}
			b&=E\left[\int_{0}^s\int_{0}^t\frac{\delta_{10}^Z(u,v)Y(u^-,v^-)\Lambda_{01}(u,dv)}{H(u,v)}dudv \mid Z\right]\\
			&=\int_{0}^s\int_{0}^t\frac{\delta_{10}^Z(u,v)P(Y_1\geq u, Y_2 \geq v \mid Z)\Lambda_{01}(u,dv)}{H(u,v)}dudv\\
			&=\int_{0}^s\int_{0}^t\frac{H^Z(u,v)\delta_{10}^Z(u,v)\Lambda_{01}(u,dv)}{H(u,v)}dudv=\int_{0}^s\int_{0}^tq_{10}^Z(u,v)\lambda_{01}(u,v)dudv.
		\end{align*}
		Consequently, 
		\[
		E\left[\int_{0}^s\int_{0}^t\frac{M_{10}(du,v)\Lambda_{01}(u,dv)}{H(u,v)} \mid Z\right]=a+b=0+b=\int_{0}^s\int_{0}^tq_{10}^Z(u,v)\lambda_{01}(u,v)dudv.
		\]
		\item Use the same proof as above.
	\end{enumerate}
\end{proof}
\textbf{Consequence}\\
By Proposition~\ref{prop:3}, we have that
\begin{align*}
	&E\left[\int_{0}^s\int_{0}^t\frac{M_{11}(du,dv)-M_{10}(du,v)\Lambda_{01}(u,dv)-M_{01}(u,dv)\Lambda_{10}(du,v)}{H(u^-,v^-)}\mid Z\right]\\
	=&\int_{0}^s\int_{0}^tq_{11}^Z(u,v)dudv-\int_{0}^s\int_{0}^tq_{10}^Z(u,v)\lambda_{01}(u,v)dudv-\int_{0}^s\int_{0}^tq_{01}^Z(u,v)\lambda_{10}(u,v)dudv\\
	=&\int_{0}^s\int_{0}^t\left[q_{11}^Z(u,v)-q_{10}^Z(u,v)\lambda_{01}(u,v)-q_{01}^Z(u,v)\lambda_{10}(u,v)\right]dudv\\
	=&\int_{0}^s\int_{0}^t\frac{\partial^2 }{\partial u\partial v}\left(\frac{S^Z(u,v)}{S(u,v)}\right)dudv,
\end{align*}
where the last equality is due to Lemma~\ref{prop:2}. Note also that by the fundamental theorem of calculus (also known as the Newton–Leibniz theorem) on the plane, we have that
\begin{align*}
	\int_{0}^s\int_{0}^t\frac{\partial^2 }{\partial u\partial v}\left(\frac{S^Z(u,v)}{S(u,v)}\right)dudv&=\frac{S^Z(s,t)}{S(s,t)}-\frac{S^Z(s,0)}{S(s,0)}-\frac{S^Z(0,t)}{S(0,t)}+\frac{S^Z(0,0)}{S(0,0)}\\&=\frac{S^Z(s,t)}{S(s,t)}-\frac{S^Z(s,0)}{S(s,0)}-\frac{S^Z(0,t)}{S(0,t)}+1.
\end{align*}
Combining this with Proposition~\ref{prop:1}, we obtain that
\begin{align*}
	&E\left[\int_{0}^{s}\frac{M_{1}(du)}{H_1(u)}+\int_{0}^{t}\frac{M_{2}(dv)}{H_2(v)}-\int_{0}^s\int_{0}^t\frac{M_{11}(du,dv)-M_{10}(du,v)\Lambda_{01}(u,dv)-M_{01}(u,dv)\Lambda_{10}(du,v)}{H(u^-,v^-)}\mid Z\right]\\
	=&\left[1-\frac{S_{T_1}(s\mid Z)}{S_{T_1}(s)}\right]+\left[1-\frac{S_{T_2}(t\mid Z)}{S_{T_2}(t)}\right]-\left[\frac{S^Z(s,t)}{S(s,t)}-\frac{S^Z(s,0)}{S(s,0)}-\frac{S^Z(0,t)}{S(0,t)}+1\right]\\
	=& 1-\frac{S^Z(s,t)}{S(s,t)},
\end{align*}
where we used the fact that $\frac{S^Z(s,0)}{S(s,0)}=\frac{S_{T_1}(s\mid Z)}{S_{T_1}(s)}$ and $\frac{S^Z(0,t)}{S(0,t)}=\frac{S_{T_2}(t\mid Z)}{S_{T_2}(t)}$. In summary, we obtain that 
\begin{align*}
	&E\left[\dot{\phi}(X)(s,t)\mid Z\right]\\
	=&E\left[-S(s,t)\left[\int_{0}^{s}\frac{M_{1}(du)}{H_1(u)}+\int_{0}^{t}\frac{M_{2}(dv)}{H_2(v)}-\int_{0}^s\int_{0}^t\frac{M_{11}(du,dv)-M_{10}(du,v)\Lambda_{01}(u,dv)-M_{01}(u,dv)\Lambda_{10}(du,v)}{H(u^-,v^-)}\right]\mid Z\right]\\
	=&-S(s,t)\left[ 1-\frac{S^Z(s,t)}{S(s,t)}\right]=S^Z(s,t)-S(s,t),
\end{align*}
which is exactly what we needed to prove for the theory of pseudo-observations to hold.

\subsection{Conclusion}\label{sec:covariance}
Based on Theorem 3.4 of \cite{overgaard_asymptotic_2017}, we obtain that $\hat{\beta}$ is consistent and asymptotically normal, with covariance matrix given by $M^{-1}\Sigma M^{-1}$,
where $M=N(\beta^*)$, $N(\beta)=E\left[\left(\frac{\partial}{\partial \beta}g^{-1}(\beta^T Z_i)\right)^T V_i^{-1}\frac{\partial}{\partial \beta}g^{-1}(\beta^T Z_i)\right]$,
\begin{equation*}
	\Sigma=Var\left[\frac{\partial g^{-1}(\beta^TZ_i)}{\partial \beta}\Bigr|_{\beta=\beta^*} V_i^{-1}\left(\phi(F)+\dot{\phi}(X_i)-g^{-1}({\beta^*}^TZ_i)\right)+h_1(X_i)\right],
\end{equation*}
$\phi$ is the estimating functional of either the Lin and Ying estimator or the Dabrowska estimator, 
and where $h_1(x)=E\left[\frac{\partial g^{-1}(\beta^TZ_i)}{\partial \beta}\Bigr|_{\beta=\beta^*} V_i^{-1}\ddot{\phi}(x,X_i)\right]$.

\section{}
We repeat the analysis in Section 3.1 using six separate regression models for the same six time points $\{(0.5,0.7),(1,0.7),(0.5,1.2),(1,1.2),(0.5,1.5),(1,1.5)\}$. We use the same bivariate logistic data generating mechanism with both univariate and bivariate independent censoring.
For $j=1,\ldots, 6$, the corresponding regression models are
\begin{equation}\label{eq:six_reg}
	g(S(t_1^j,t_2^j\mid Z))=\beta_0^j+\beta_1^jZ,
\end{equation}
where both $\beta_0^j$ and $\beta_1^j$ depend on the $j$th time point, and $g$ is the logit link function. For the response, we use the bivariate pseudo-observations that are based on estimates of the bivariate survival function using either (1) the Lin and Ying estimator or (2) the Dabrowska estimator. We emphasize that for the case of bivariate censoring we first modify the Lin and Ying estimator, by substituting ${\hat{G}(\max(t_1,t_2))}$ in the denominator with $\hat{G_1}(t_1)\hat{G_2}(t_2)$.

We estimate $\beta_0$ and $\beta_1$ based on 500 simulations, for each of the six time points, for both types of estimators, and for both censoring mechanisms. Web Table~\ref{table:logistic_6_models} (top) presents the true values and the estimated values of the regression coefficients for the univariate censoring mechanism, and Web Table~\ref{table:logistic_6_models} (bottom) presents the true values and the estimated values of the regression coefficients for the bivariate censoring mechanism. As can be seen, for the univariate censoring scenario, the regression estimates that are based on either the Dabrowska estimator or the Lin and Ying estimator are relatively close to their true values, and the standard deviations of the Dabrowska estimator are lower than those based on the Lin and Ying estimator. The coverage is usually either correct or conservative. In the bivariate censoring scenario, the performance of the pseudo-observations approach based on the Dabrowska estimator is better than that of the modified Lin and Ying estimator, both in terms of lower bias and lower variance, and the coverage obtained from both estimators is generally correct. Interestingly, the bivariate pseudo-observations approach based on the Dabrowska estimator resulted in better estimates in the bivariate censoring scenario rather than in the univariate censoring scenario, both in terms of lower bias and lower variance.

\begin{table}[ht]
	\begin{subtable}{1\textwidth}
		\centering
		\subcaption{Univariate censoring}
		\begin{tabular}{llrrrrrrrr}
			\hline
			\multicolumn{2}{r}{} & 	\multicolumn{4}{c}{Dabrowska} & \multicolumn{4}{c}{Lin and Ying} \\	
			\hline
			\multicolumn{4}{r}{} & 	\multicolumn{2}{l}{$\beta_0$} & \multicolumn{4}{r}{} \\	
			\hline
			Time point & True value	& $\hat{\beta}$ & sd & se & cov (\%) & $\hat{\beta}$ & sd & se & cov (\%) \\ 
			\hline
			$(0.5,0.7)$ & -0.96 & -0.98 & 0.72 & 0.68 & 97.20 & -0.98 & 0.75 & 0.72 & 95.80 \\ 
			$(1,0.7)$ & -1.13 & -1.15 & 0.70 & 0.66 & 96.00 & -1.15 & 0.77 & 0.73 & 96.60 \\ 
			$(0.5,1.2)$ & -1.41 & -1.46 & 0.70 & 0.66 & 94.20 & -1.46 & 0.78 & 0.72 & 94.00 \\ 
			$(1,1.2)$  & -1.53 & -1.58 & 0.70 & 0.65 & 94.80 & -1.58 & 0.77 & 0.72 & 94.00 \\ 
			$(0.5,1.5)$ & -1.61 & -1.65 & 0.70 & 0.66 & 95.00 & -1.66 & 0.78 & 0.75 & 95.60 \\ 
			$(1,1.5)$ & -1.70 & -1.75 & 0.70 & 0.66 & 95.00 & -1.75 & 0.78 & 0.74 & 96.40 \\ 
			\hline
			\multicolumn{4}{r}{} & 	\multicolumn{2}{l}{$\beta_1$} & \multicolumn{4}{r}{} \\	
			\hline
			$(0.5,0.7)$ &  \multirow{6}{*}{2.00} & 2.04 & 0.78 & 0.71 & 97.00 & 2.05 & 0.80 & 0.75 & 95.60 \\ 
			$(1,0.7)$ &   & 2.03 & 0.73 & 0.68 & 95.60 & 2.04 & 0.80 & 0.75 & 96.40 \\ 
			$(0.5,1.2)$ &  & 2.06 & 0.71 & 0.66 & 95.00 & 2.07 & 0.79 & 0.73 & 94.80 \\ 
			$(1,1.2)$ &  &  2.06 & 0.69 & 0.65 & 95.60 & 2.07 & 0.77 & 0.72 & 96.00 \\ 
			$(0.5,1.5)$ &  & 2.06 & 0.70 & 0.66 & 95.80 & 2.07 & 0.78 & 0.74 & 95.40 \\ 
			$(1,1.5)$ &  & 2.06 & 0.69 & 0.65 & 94.80 & 2.06 & 0.78 & 0.74 & 95.80 \\ 
			\hline
		\end{tabular}
	\end{subtable}
	\hfill
	\hfill\\
	\hfill\\
	\begin{subtable}{1\textwidth}
		\centering
		\subcaption{Independent bivariate censoring and a modification of Lin and Ying}
		\begin{tabular}{llrrrrrrrr}
			\hline
			\multicolumn{2}{r}{} & 	\multicolumn{4}{c}{Dabrowska} & \multicolumn{4}{c}{Lin and Ying} \\	
			\hline
			\multicolumn{4}{r}{} & 	\multicolumn{2}{l}{$\beta_0$} & \multicolumn{4}{r}{} \\	
			\hline
			Time point & True value	& $\hat{\beta}$ & sd & se & cov (\%) & $\hat{\beta}$ & sd & se & cov (\%) \\ 
			\hline
			$(0.5,0.7)$  & -0.96 & -0.96 & 0.68 & 0.66 & 95.20 & -1.00 & 0.81 & 0.75 & 96.20 \\ 
			$(1,0.7)$ & -1.13 & -1.13 & 0.69 & 0.65 & 95.20 & -1.16 & 0.83 & 0.79 & 95.40 \\ 
			$(0.5,1.2)$ &  -1.41 & -1.46 & 0.65 & 0.63 & 94.00 & -1.47 & 0.75 & 0.73 & 95.00 \\ 
			$(1,1.2)$ &   -1.53 & -1.57 & 0.66 & 0.63 & 94.00 & -1.60 & 0.84 & 0.80 & 95.60 \\ 
			$(0.5,1.5)$ &  -1.61 & -1.67 & 0.66 & 0.63 & 94.60 & -1.69 & 0.77 & 0.73 & 94.60 \\ 
			$(1,1.5)$ &   -1.70 & -1.76 & 0.67 & 0.63 & 95.00 & -1.78 & 0.85 & 0.82 & 94.80 \\ 
			\hline
			\multicolumn{4}{r}{} & 	\multicolumn{2}{l}{$\beta_1$} & \multicolumn{4}{r}{} \\	
			\hline
			$(0.5,0.7)$ &  \multirow{6}{*}{2.00} & 2.01 & 0.70 & 0.69 & 95.40 & 2.06 & 0.88 & 0.80 & 95.80 \\ 
			$(1,0.7)$ &   &  2.01 & 0.71 & 0.67 & 95.00 & 2.05 & 0.86 & 0.83 & 95.20 \\ 
			$(0.5,1.2)$ &  & 2.05 & 0.65 & 0.64 & 94.80 & 2.07 & 0.77 & 0.74 & 95.20 \\ 
			$(1,1.2)$ &  & 2.05 & 0.66 & 0.63 & 95.00 & 2.09 & 0.86 & 0.82 & 95.60 \\ 
			$(0.5,1.5)$ &  & 2.06 & 0.66 & 0.63 & 94.60 & 2.09 & 0.78 & 0.73 & 95.20 \\ 
			$(1,1.5)$ &  & 2.06 & 0.66 & 0.62 & 94.60 & 2.09 & 0.87 & 0.82 & 94.80 \\ 
			\hline
		\end{tabular}
	\end{subtable}
	\hfill\\
	\caption{Top: univariate censoring. Bottom: independent bivariate censoring. Mean, sd, se, and coverage of the estimated regression parameters $\beta_0$ and $\beta_1$, for both the Dabrowska estimator, and the Lin and Ying estimator. These estimates are based on six separate regression models for the six time points, and on 500 simulations from the bivariate logistic model. The estimate $\hat{\beta}$ is the average of 500 parameter estimates; sd, square root of empirical variance of 500 simulation replications; se, square root of average estimated variance based on the ordinary sandwich estimator.
	}\label{table:logistic_6_models}
\end{table}

Web Figure~\ref{fig:MAEs_correct_model_6_reg} and Web Figure~\ref{fig:MAEs_correct_model_6_reg_bivar_cens} present the boxplot of the MAEs (top) and the standardized MAEs (bottom) for the univariate censoring scenario and the bivariate censoring scenario, respectively. For both censoring mechanisms, the relatively low MAEs show that the bivariate pseudo-observations approach estimates quite well the conditional joint survival probability, with a slight preference to the Dabrowska estimator over the Lin and Ying estimator.
Note also that the standardized MAEs are by definition larger than the MAEs, as we are dividing the MAEs with the joint survival probability (a value between 0 and 1).

\begin{figure}[p]
	\centering
	\begin{subfigure}[b]{1\textwidth}
		\centering
		\includegraphics[width=1\textwidth]{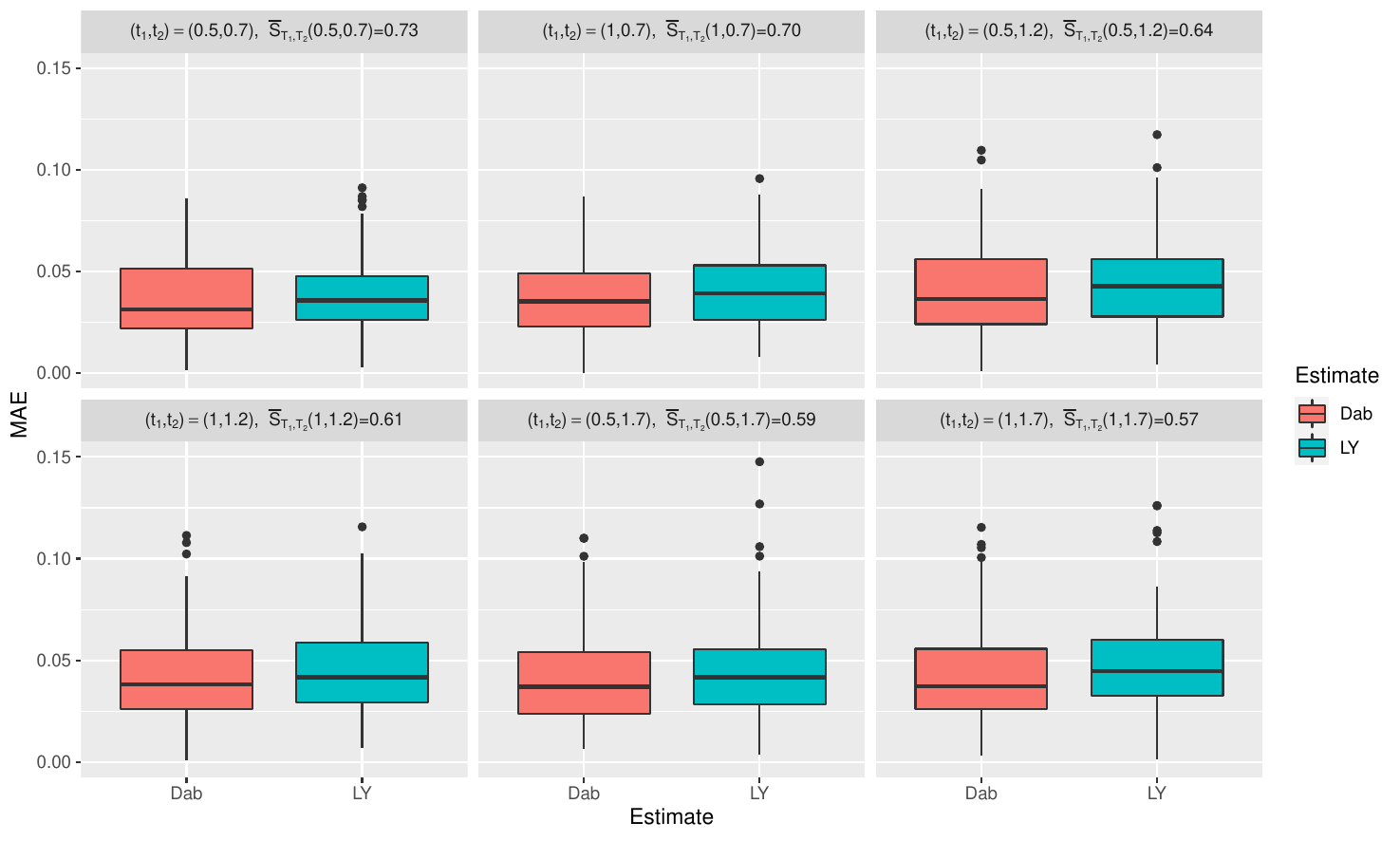}
		\caption{MAEs} \label{fig:MAES_1_reg_logis_data_m=500_n=200}
	\end{subfigure}
	\hfill
	\begin{subfigure}[b]{1\textwidth}
		\centering
		\includegraphics[width=1\textwidth]{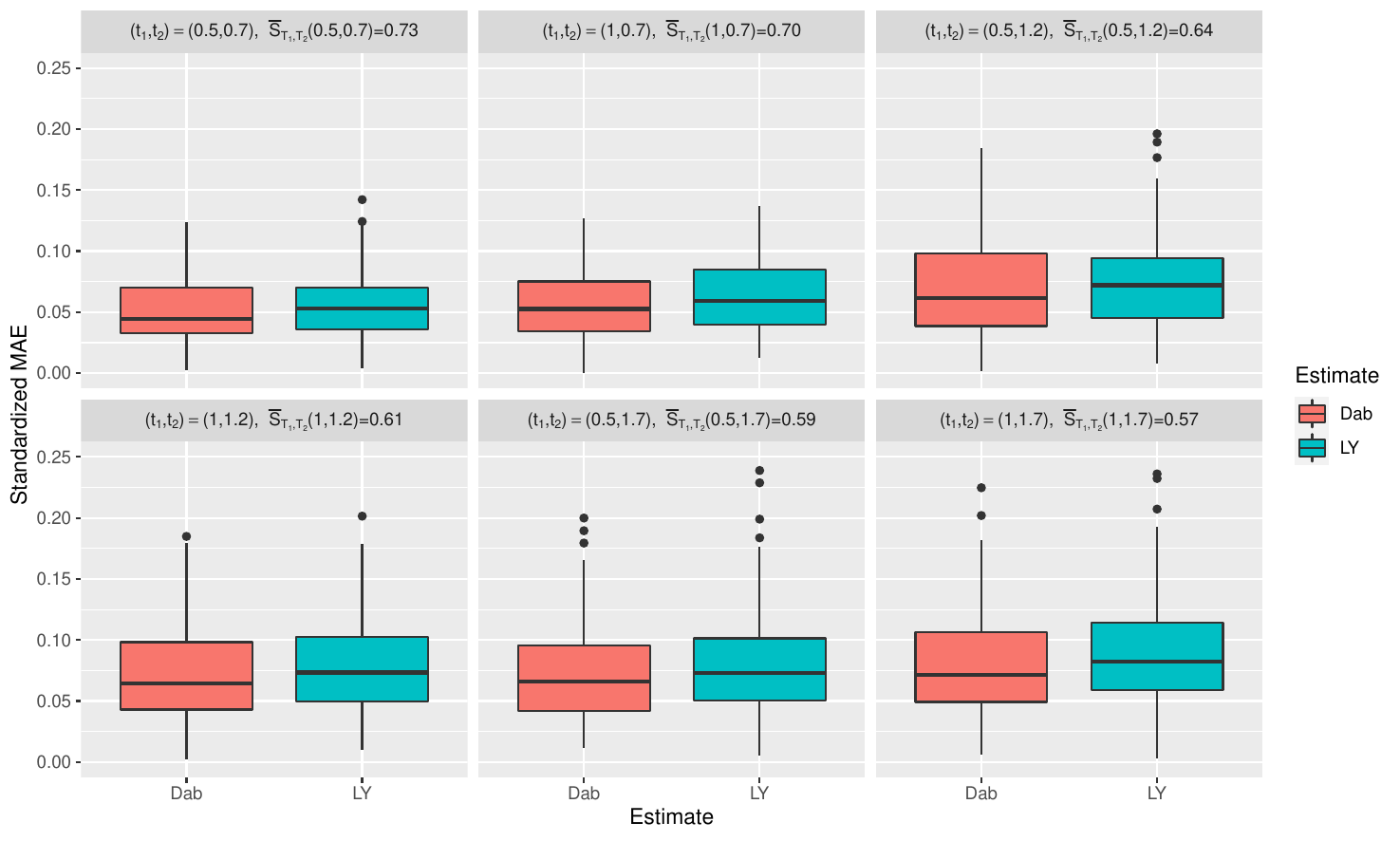}
		\caption{Standardized MAEs} \label{fig:SMAES_1_reg_logis_data_m=500_n=200}
	\end{subfigure}
	\caption{Bivariate logistic times, univariate censoring, $n=200$, six fixed time points, a single regression model. MAEs and standardized MAEs between the true joint survival of the bivariate logistic failure times, and the estimated survival that is based on a single regression model for all six time points, for both types of estimators and for 500 simulations from the univariate censoring scenario. The top row of each panel specifies the time point $(t_1^j,t_2^j)$, and the mean value of the true joint survival probability $\bar{S}\equiv \bar{S}_{T_1,T_2}(t_1^j,t_2^j)=\frac{1}{n}\sum_{i=1}^{n}S_{T_1,T_2}(t_1^j,t_2^j\mid Z_i)$.}
	\label{fig:MAEs_correct_model}
\end{figure}

\begin{figure}[p]
	\centering
	\begin{subfigure}[b]{1\textwidth}
		\centering
		\includegraphics[width=1\textwidth]{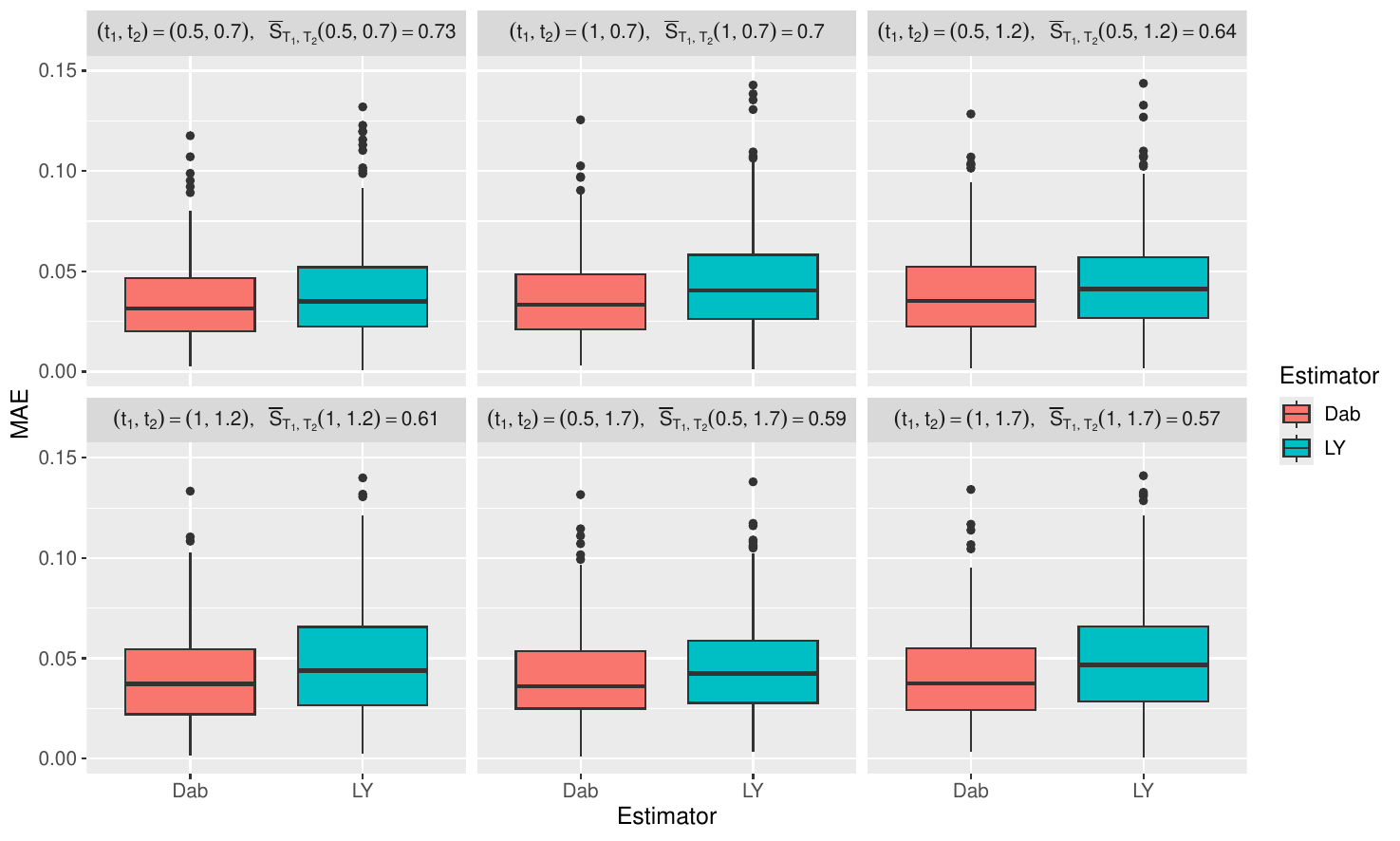}
		\caption{MAEs} \label{fig:MAES_1_reg_logis_data_m=500_n=200_bivar_cens}
	\end{subfigure}
	\hfill
	\begin{subfigure}[b]{1\textwidth}
		\centering
		\includegraphics[width=1\textwidth]{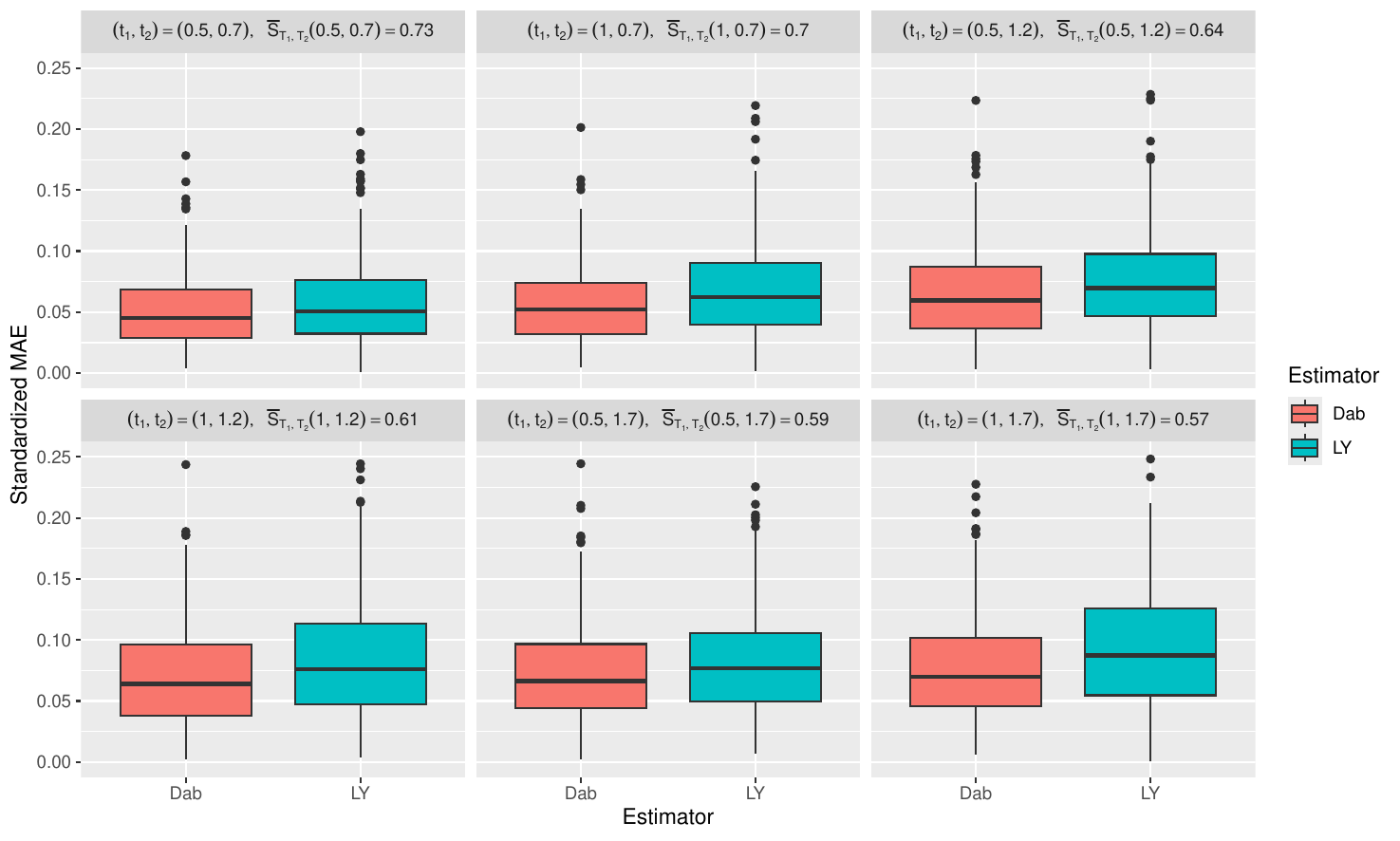}
		\caption{Standardized MAEs} \label{fig:SMAES_1_reg_logis_data_m=500_n=200_bivar_cens}
	\end{subfigure}
	\caption{Bivariate logistic times, bivariate censoring, $n=200$, six fixed time points, a single regression model. MAEs and standardized MAEs between the true joint survival of the bivariate logistic failure times, and the estimated survival that is based on a single regression model for all six time points, for both types of estimators and for 500 simulations from the bivariate censoring scenario. The top row of each panel specifies the time point $(t_1^j,t_2^j)$, and the mean value of the true joint survival probability $\bar{S}\equiv \bar{S}_{T_1,T_2}(t_1^j,t_2^j)=\frac{1}{n}\sum_{i=1}^{n}S_{T_1,T_2}(t_1^j,t_2^j\mid Z_i)$.}
	\label{fig:MAEs_correct_model_bivar_cens}
\end{figure}

\begin{figure}[p]
	\centering
	\includegraphics[width=1\textwidth]{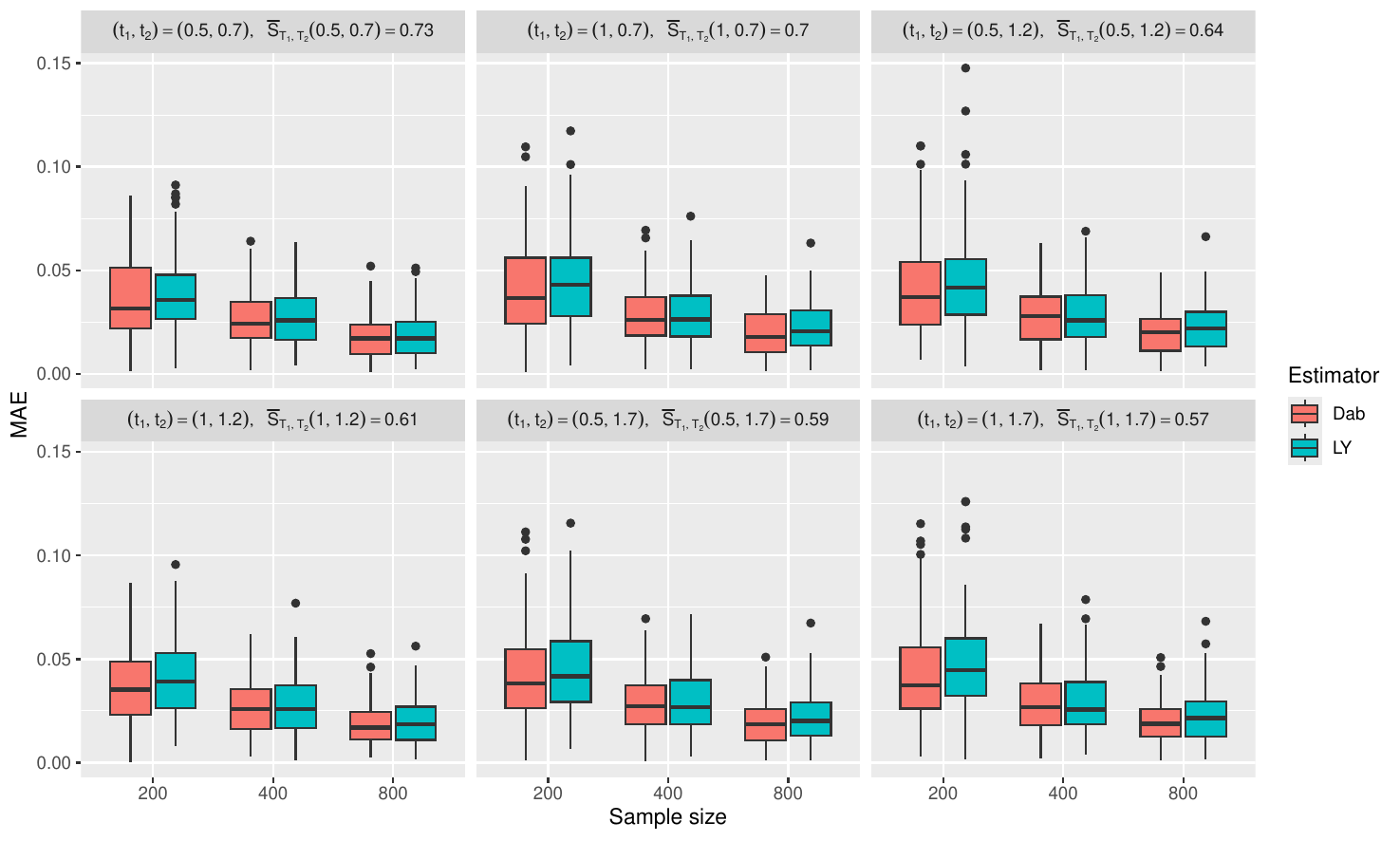}
	\caption{Bivariate logistic times, univariate censoring, $n=200,400,800$, six fixed time points, a single regression model. MAEs between the true joint survival of the bivariate logistic failure times, and the estimated survival that is based on a single regression model for all six time points, for both types of estimators and for 100 simulations for each sample size. The top row of each panel specifies the time point $(t_1^j,t_2^j)$, and the mean value of the true joint survival probability $\bar{S}\equiv \bar{S}_{T_1,T_2}(t_1^j,t_2^j)=\frac{1}{n}\sum_{i=1}^{n}S_{T_1,T_2}(t_1^j,t_2^j\mid Z_i)$.}
	\label{fig:MAEs_sample_size}
\end{figure}

\begin{figure}[p]
	\centering
	\begin{subfigure}[b]{1\textwidth}
		\centering
		\includegraphics[width=1\textwidth]{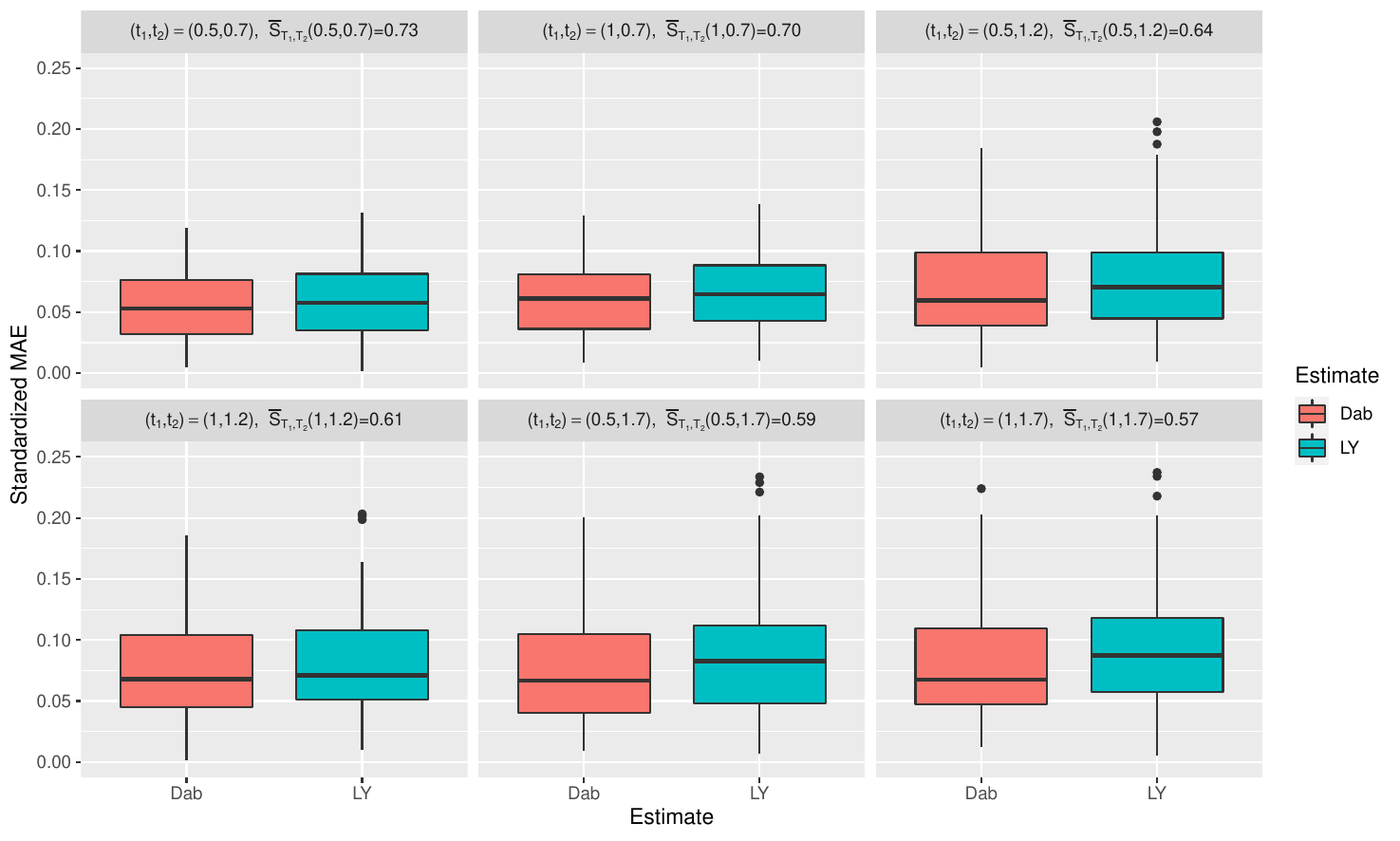}
		\caption{MAEs} \label{fig:MAES_6_reg_logis_data_m=500_n=200}
	\end{subfigure}
	\hfill
	\begin{subfigure}[b]{1\textwidth}
		\centering
		\includegraphics[width=1\textwidth]{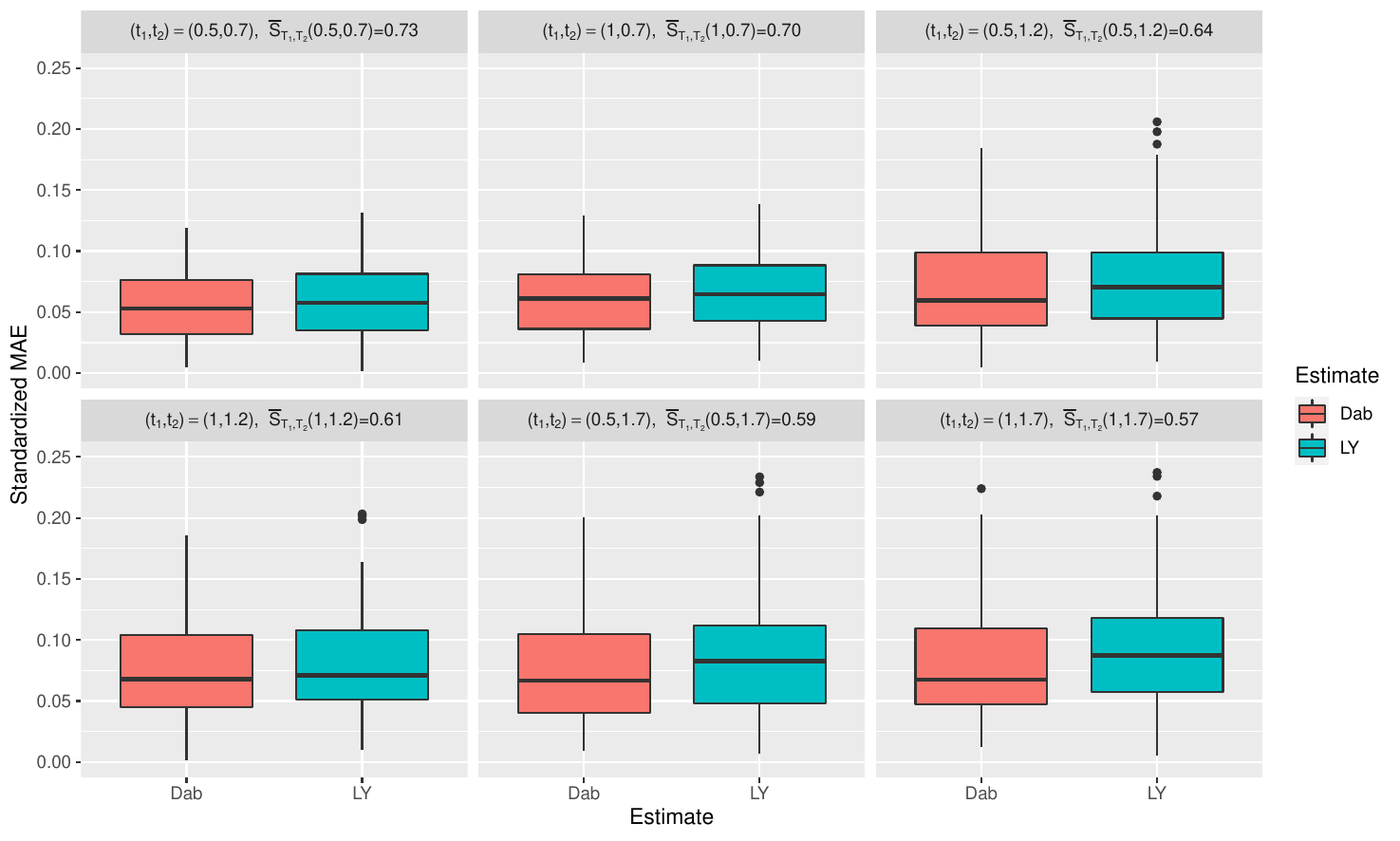}
		\caption{Standardized MAEs} \label{fig:SMAES_6_reg_logis_data_m=500_n=200}
	\end{subfigure}
	\caption{Bivariate logistic times, univariate censoring, $n=200$, six fixed time points, six separate regression models. MAEs and standardized MAEs between the true joint survival of the bivariate logistic failure times, and the estimated survival that is based on six separate regression models (one for each time point), for both types of estimators and for 500 simulations from the univariate censoring scenario. The top row of each panel specifies the time point $(t_1^j,t_2^j)$, and the mean value of the true joint survival probability $\bar{S}\equiv \bar{S}_{T_1,T_2}(t_1^j,t_2^j)=\frac{1}{n}\sum_{i=1}^{n}S_{T_1,T_2}(t_1^j,t_2^j\mid Z_i)$.}
	\label{fig:MAEs_correct_model_6_reg}
\end{figure}

\begin{figure}[p]
	\centering
	\begin{subfigure}[b]{1\textwidth}
		\centering
		\includegraphics[width=1\textwidth]{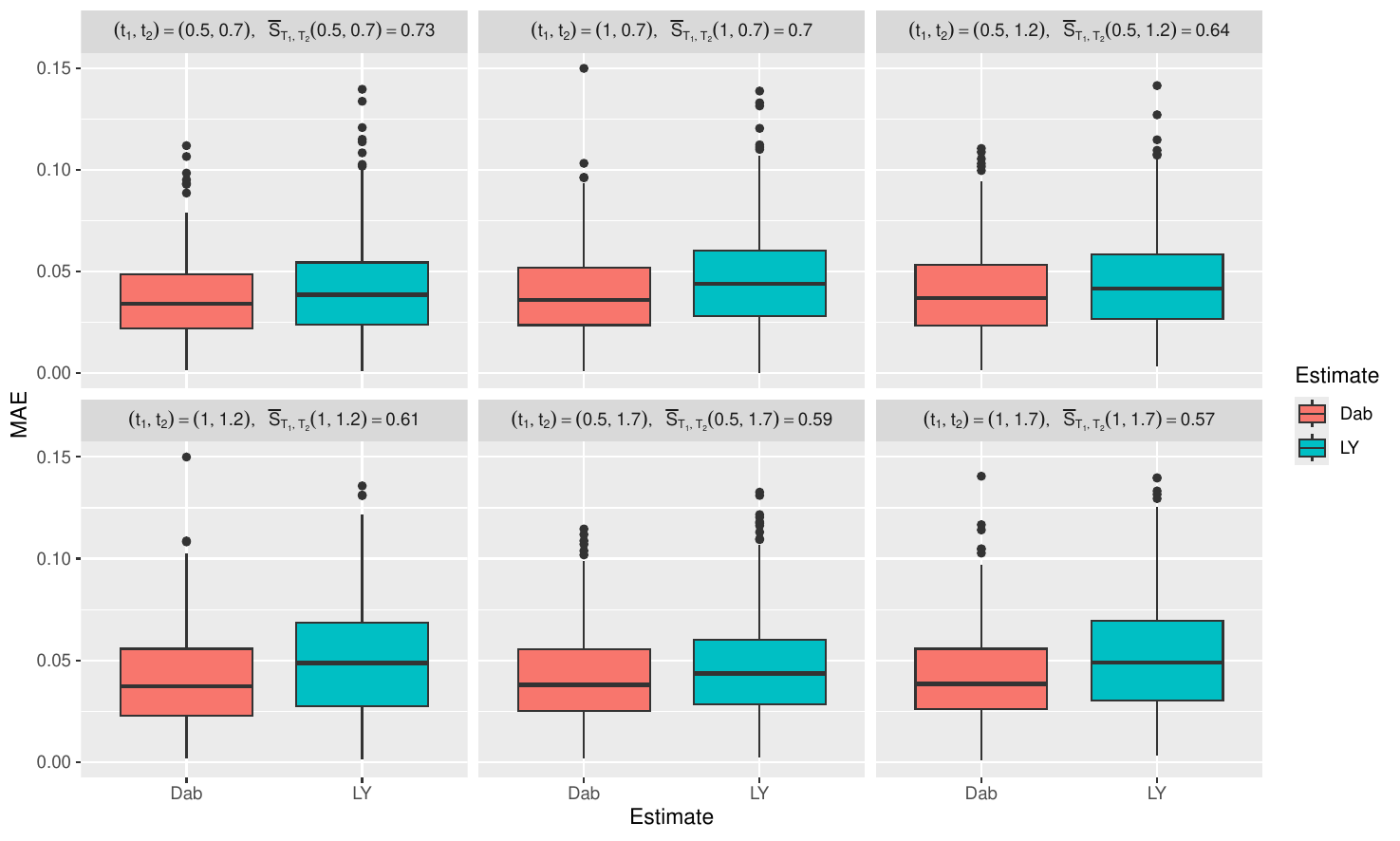}
		\caption{MAEs} \label{fig:MAES_6_reg_logis_data_m=500_n=200_bivar_cens}
	\end{subfigure}
	\hfill
	\begin{subfigure}[b]{1\textwidth}
		\centering
		\includegraphics[width=1\textwidth]{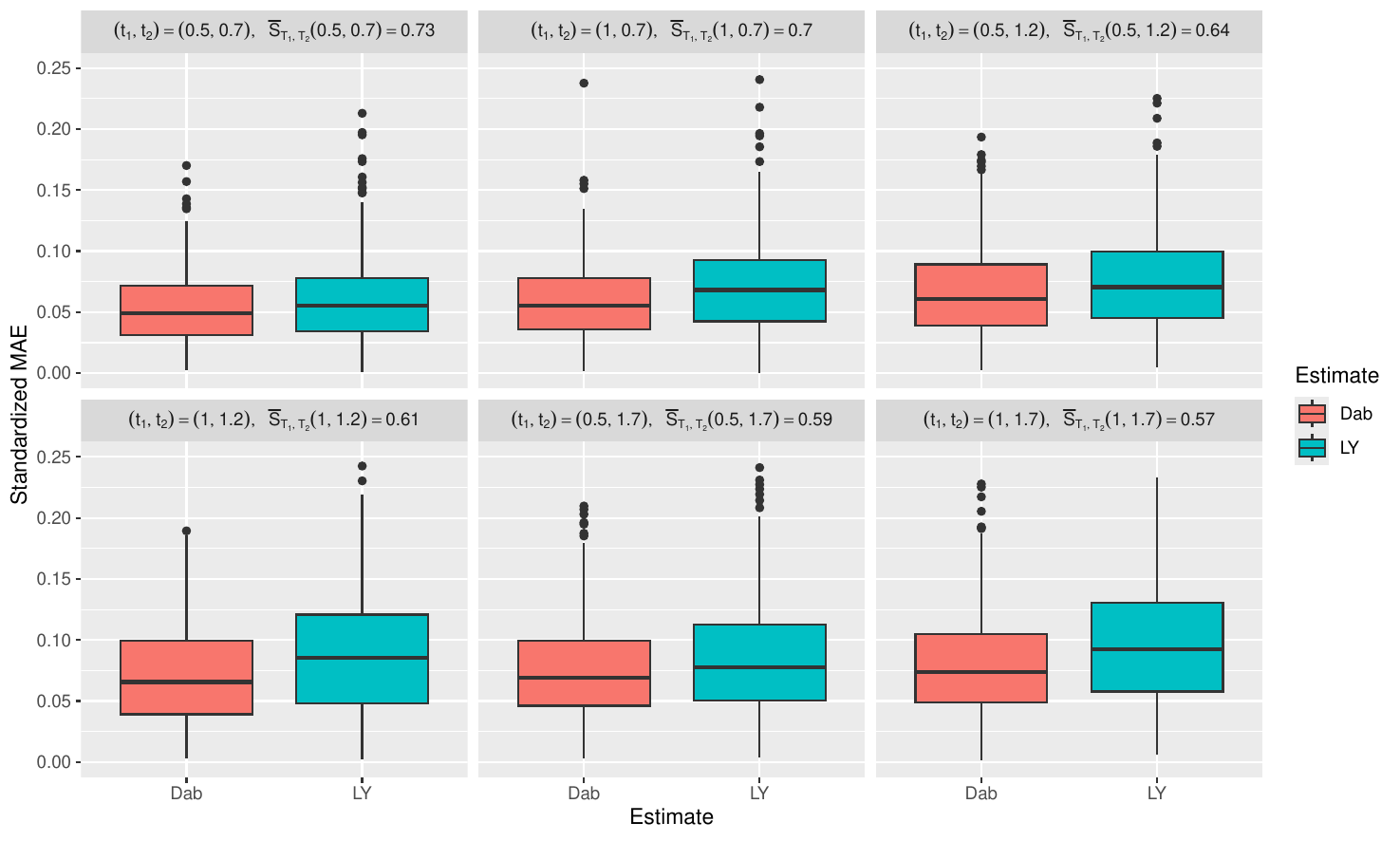}
		\caption{Standardized MAEs} \label{fig:SMAES_6_reg_logis_data_m=500_n=200_bivar_cens}
	\end{subfigure}
	\caption{Bivariate logistic times, bivariate censoring, $n=200$, six fixed time points, six separate regression models. MAEs and standardized MAEs between the true joint survival of the bivariate logistic failure times, and the estimated survival that is based on six separate regression models (one for each time point), for both types of estimators and for 500 simulations from the bivariate censoring scenario. The top row of each panel specifies the time point $(t_1^j,t_2^j)$, and the mean value of the true joint survival probability $\bar{S}\equiv \bar{S}_{T_1,T_2}(t_1^j,t_2^j)=\frac{1}{n}\sum_{i=1}^{n}S_{T_1,T_2}(t_1^j,t_2^j\mid Z_i)$.}
	\label{fig:MAEs_correct_model_6_reg_bivar_cens}
\end{figure}

\begin{figure}[p]
	\centering
	\begin{subfigure}[b]{1\textwidth}
		\centering
		\includegraphics[width=1\textwidth]{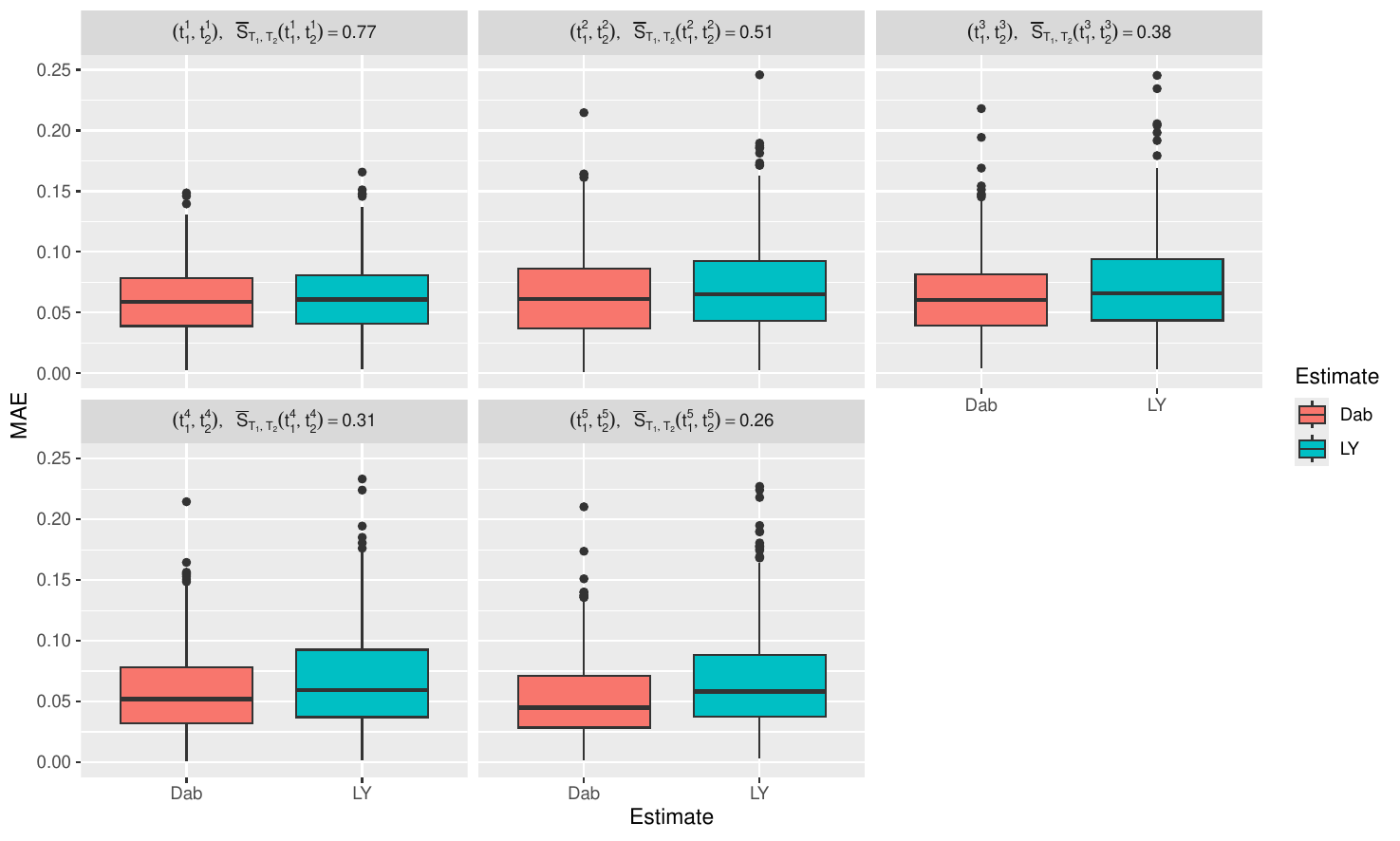}
		\caption{MAEs} \label{fig:MAEs_logitistic_data_5k_univ}
	\end{subfigure}
	\hfill
	\begin{subfigure}[b]{1\textwidth}
		\centering
		\includegraphics[width=1\textwidth]{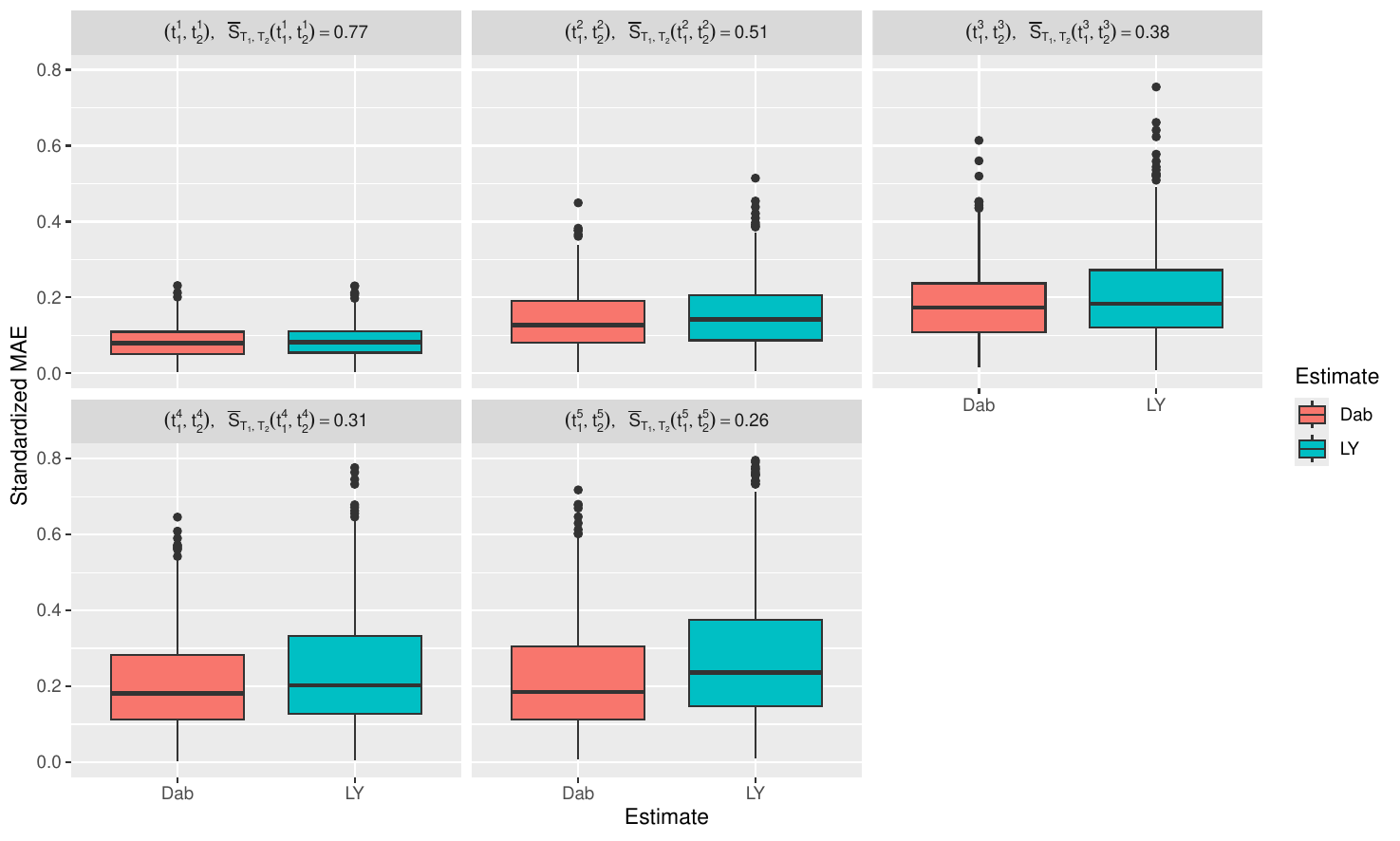}
		\caption{Standardized MAEs} \label{fig:SMAEs_logistic_data_5k_univ}
	\end{subfigure}
	\caption{Bivariate logistic times, univariate censoring, $n=200$, $k=5$ selected time points, a single regression model .MAEs and standardized MAEs between the true joint survival of the bivariate logistic failure times, and the estimated survival that is based on a single regression model for $k=5$ time points, for both types of estimators and for $m=500$ simulations from the univariate censoring scenario. The top row of each panel specifies the time point $(t_1^j,t_2^j)$ (which is different for each iteration), where $j=1,\ldots, k$ and the mean value of the true joint survival probability is $\bar{S}\equiv \bar{S}_{T_1,T_2}(t_1^j,t_2^j)=\frac{1}{nm}\sum_{i=1}^{n}\sum_{l=1}^{m}S_{T_1,T_2}(t_{1l}^j,t_{2l}^j\mid Z_i)$.}
	\label{fig:MAEs_correct_model_5k_univ}
\end{figure}

\begin{figure}[p]
	\centering
	\begin{subfigure}[b]{1\textwidth}
		\centering
		\includegraphics[width=1\textwidth]{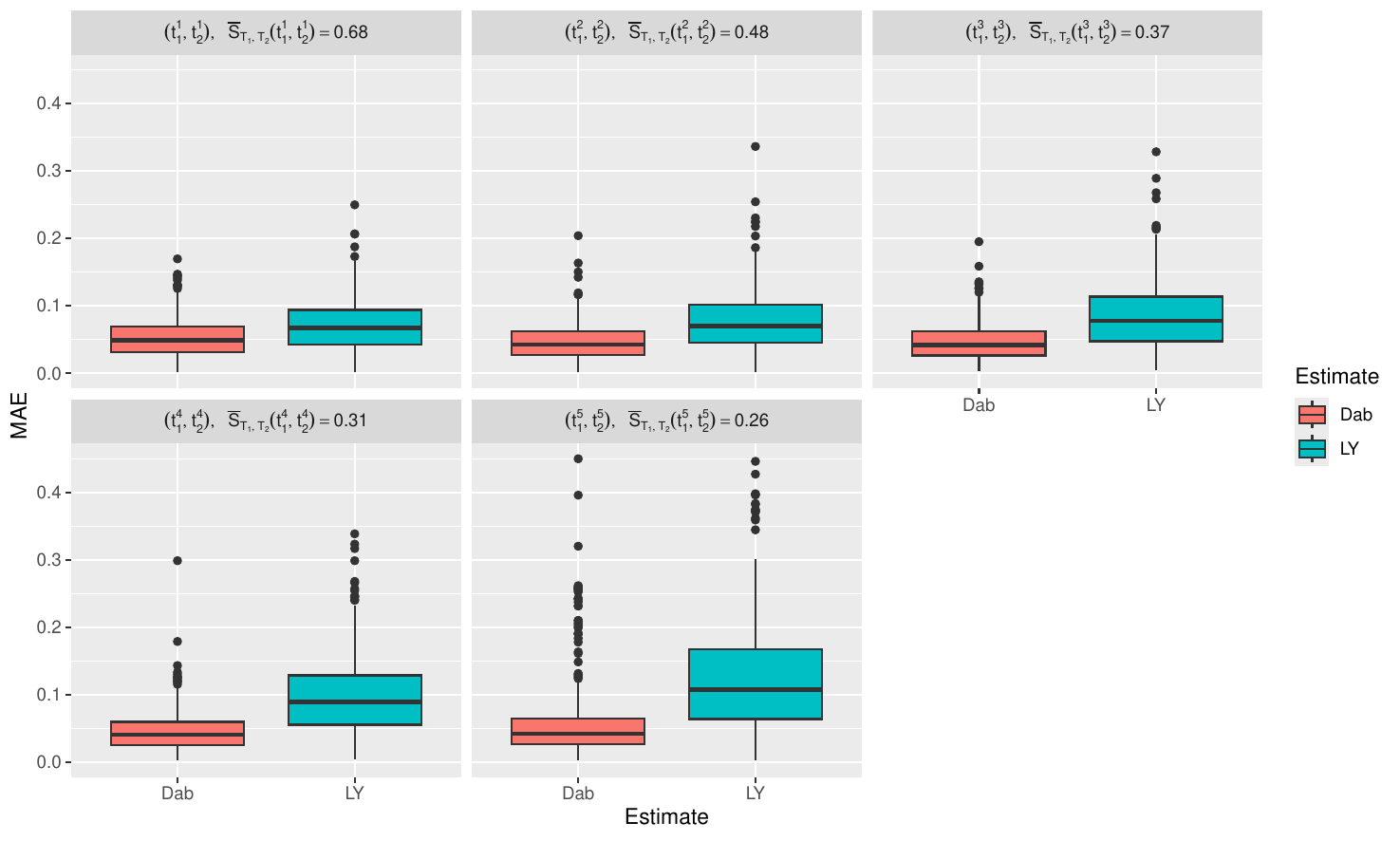}
		\caption{MAEs} \label{fig:MAEs_logitistic_data_5k_bivar_cens}
	\end{subfigure}
	\hfill
	\begin{subfigure}[b]{1\textwidth}
		\centering
		\includegraphics[width=1\textwidth]{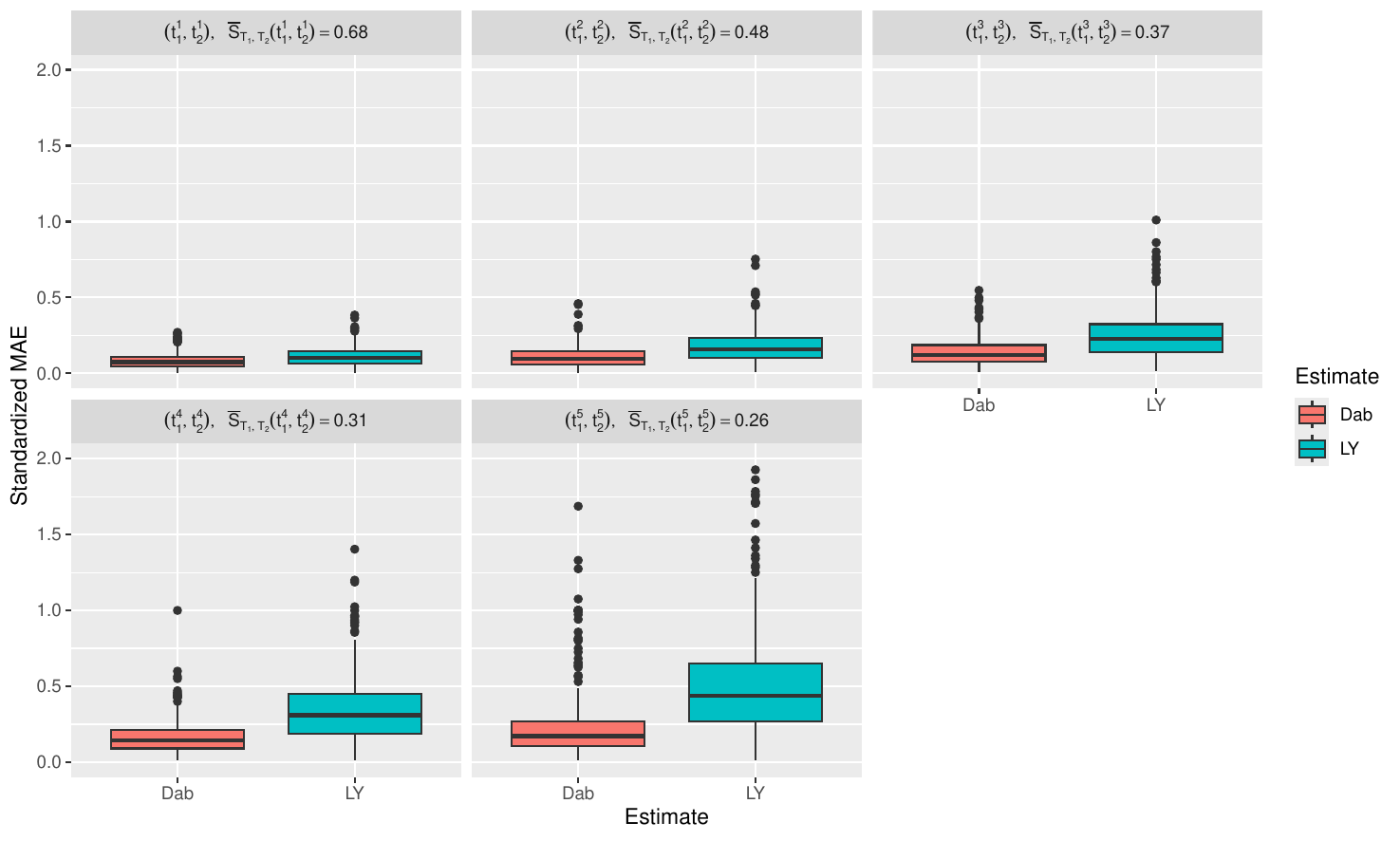}
		\caption{Standardized MAEs} \label{fig:SMAEs_logistic_data_5k_bivar_cens}
	\end{subfigure}
	\caption{Bivariate logistic times, bivariate censoring, $n=200$, $k=5$ selected time points, a single regression model .MAEs and standardized MAEs between the true joint survival of the bivariate logistic failure times, and the estimated survival that is based on a single regression model for $k=5$ time points, for both types of estimators and for $m=500$ simulations from the bivariate censoring scenario. The top row of each panel specifies the time point $(t_1^j,t_2^j)$ (which is different for each iteration), where $j=1,\ldots, k$ and the mean value of the true joint survival probability is $\bar{S}\equiv \bar{S}_{T_1,T_2}(t_1^j,t_2^j)=\frac{1}{nm}\sum_{i=1}^{n}\sum_{l=1}^{m}S_{T_1,T_2}(t_{1l}^j,t_{2l}^j\mid Z_i)$.}
	\label{fig:MAEs_correct_model_5k_bivar_cens}
\end{figure}

\begin{figure}[p]
	\centering
	\begin{subfigure}[b]{1\textwidth}
		\centering
		\includegraphics[width=1\textwidth]{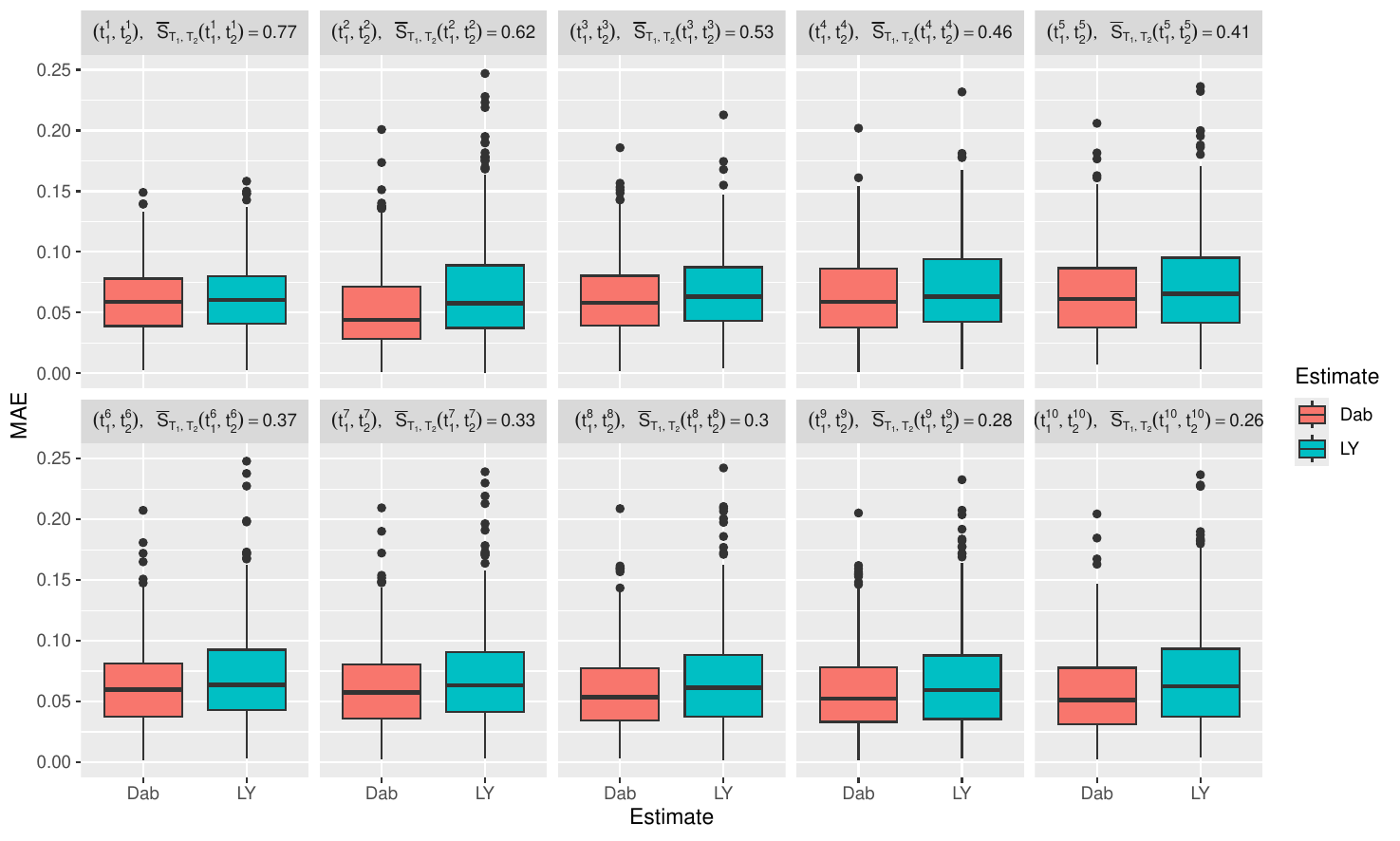}
		\caption{MAEs} \label{fig:MAEs_logitistic_data_10k_univ}
	\end{subfigure}
	\hfill
	\begin{subfigure}[b]{1\textwidth}
		\centering
		\includegraphics[width=1\textwidth]{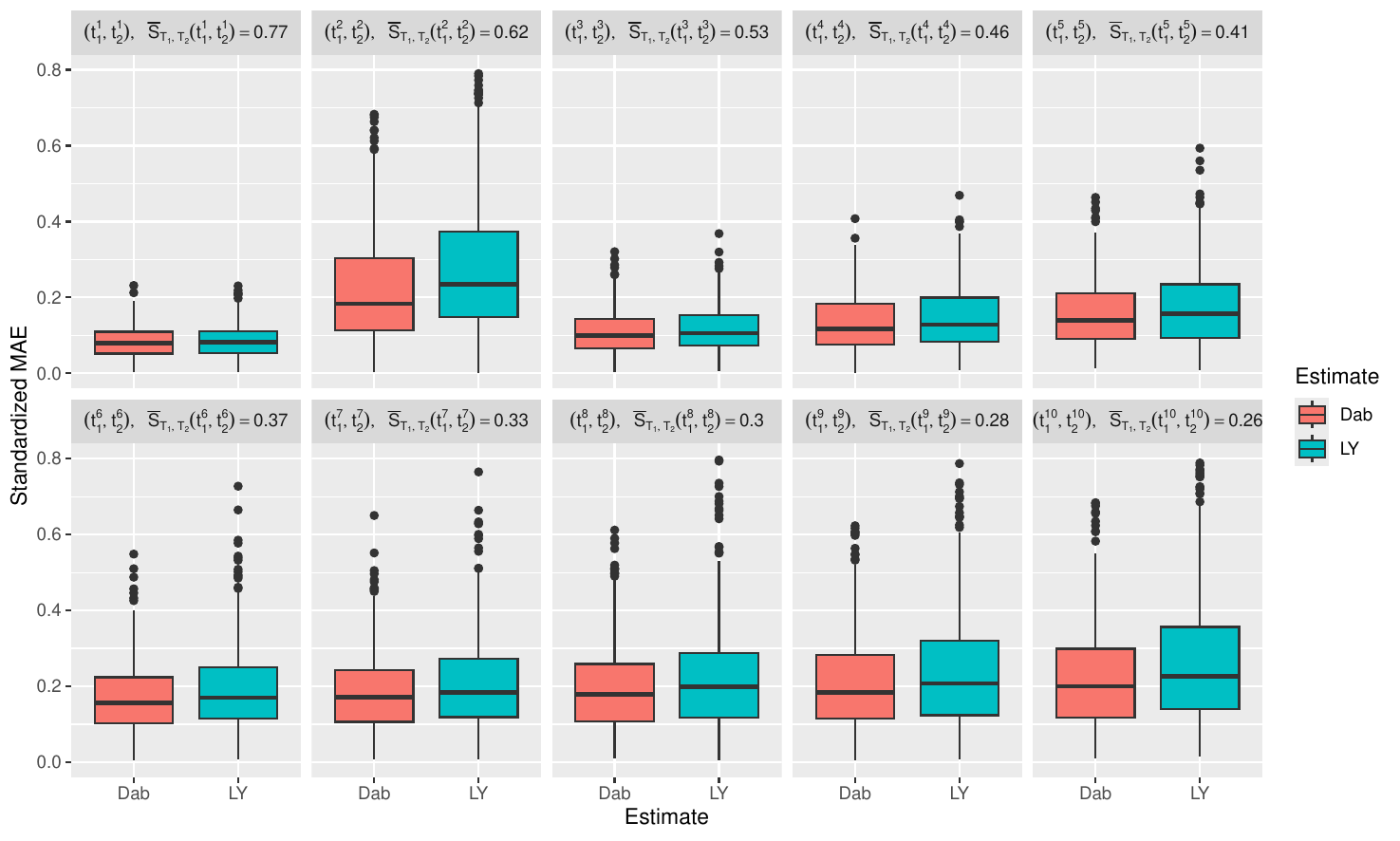}
		\caption{Standardized MAEs} \label{fig:SMAEs_logistic_data_10k_univ}
	\end{subfigure}
	\caption{Bivariate logistic times, univariate censoring, $n=200$, $k=10$ selected time points, a single regression model .MAEs and standardized MAEs between the true joint survival of the bivariate logistic failure times, and the estimated survival that is based on a single regression model for $k=105$ time points, for both types of estimators and for $m=500$ simulations from the univariate censoring scenario. The top row of each panel specifies the time point $(t_1^j,t_2^j)$ (which is different for each iteration), where $j=1,\ldots, k$ and the mean value of the true joint survival probability is $\bar{S}\equiv \bar{S}_{T_1,T_2}(t_1^j,t_2^j)=\frac{1}{nm}\sum_{i=1}^{n}\sum_{l=1}^{m}S_{T_1,T_2}(t_{1l}^j,t_{2l}^j\mid Z_i)$.}
	\label{fig:MAEs_correct_model_10k_univ}
\end{figure}

\begin{figure}[p]
	\centering
	\begin{subfigure}[b]{1\textwidth}
		\centering
		\includegraphics[width=1\textwidth]{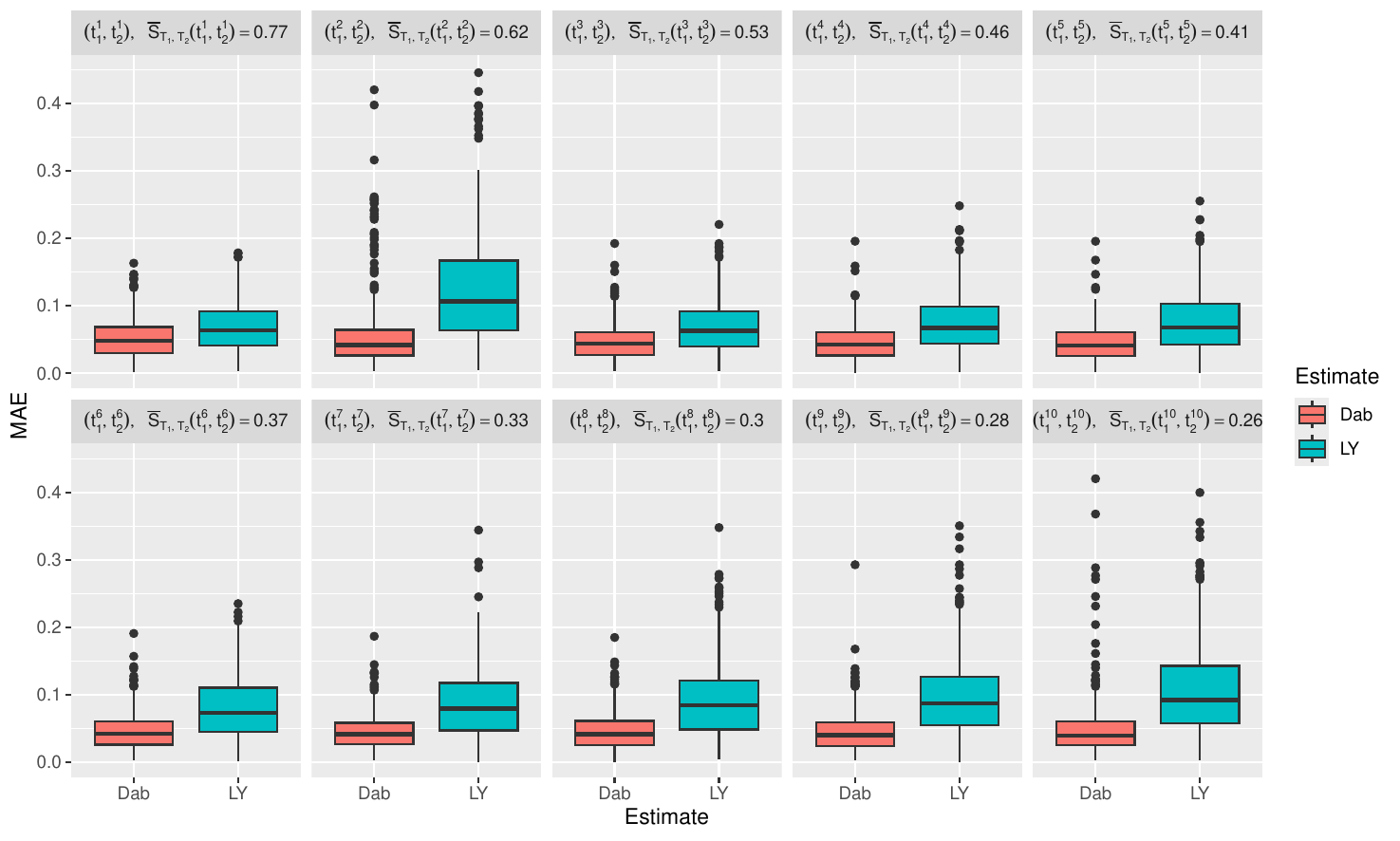}
		\caption{MAEs} \label{fig:MAEs_logitistic_data_10k_bivar_cens}
	\end{subfigure}
	\hfill
	\begin{subfigure}[b]{1\textwidth}
		\centering
		\includegraphics[width=1\textwidth]{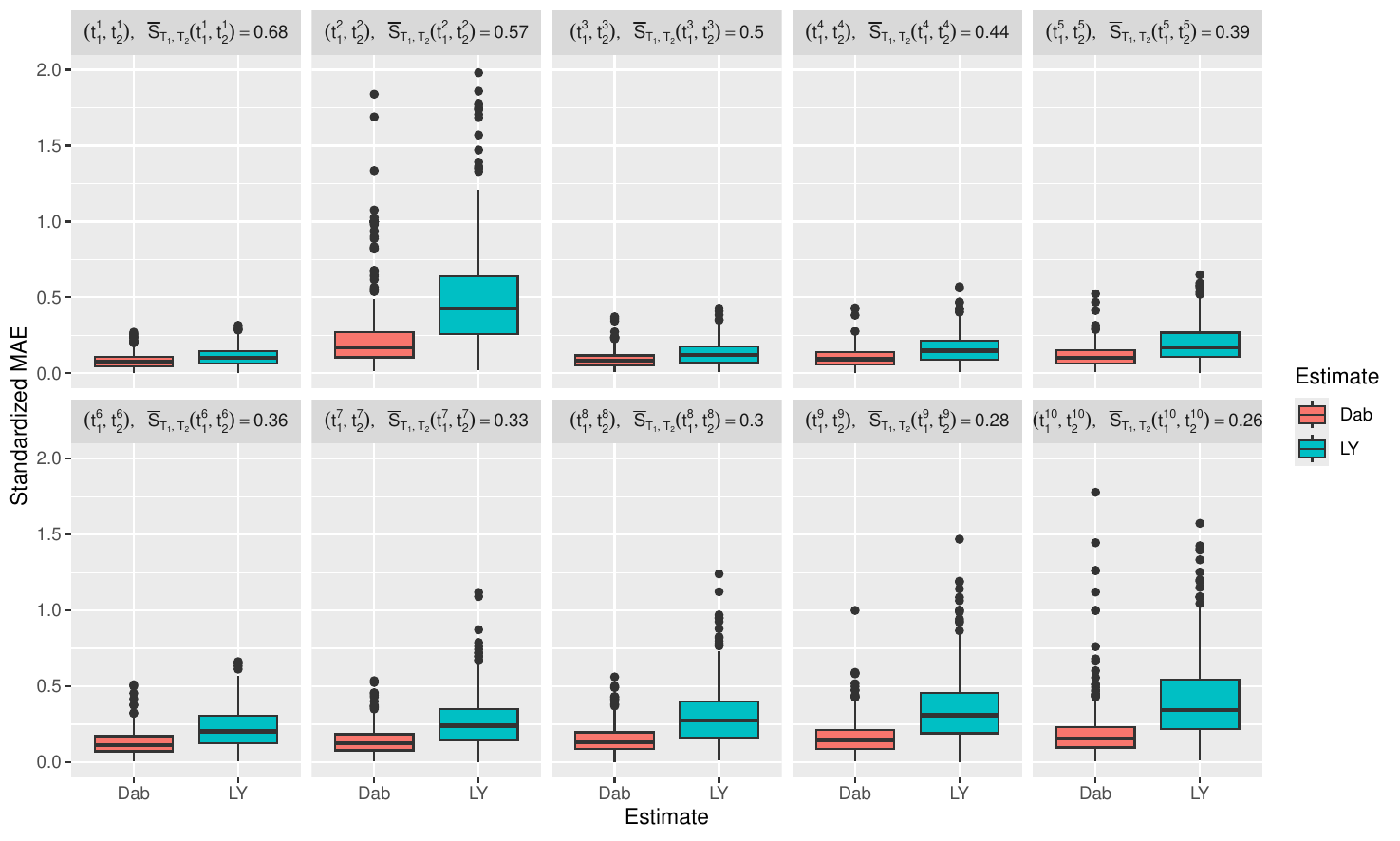}
		\caption{Standardized MAEs} \label{fig:SMAEs_logistic_data_10k_bivar_cens}
	\end{subfigure}
	\caption{Bivariate logistic times, bivariate censoring, $n=200$, $k=10$ selected time points, a single regression model .MAEs and standardized MAEs between the true joint survival of the bivariate logistic failure times, and the estimated survival that is based on a single regression model for $k=10$ time points, for both types of estimators and for $m=500$ simulations from the bivariate censoring scenario. The top row of each panel specifies the time point $(t_1^j,t_2^j)$ (which is different for each iteration), where $j=1,\ldots, k$ and the mean value of the true joint survival probability is $\bar{S}\equiv \bar{S}_{T_1,T_2}(t_1^j,t_2^j)=\frac{1}{nm}\sum_{i=1}^{n}\sum_{l=1}^{m}S_{T_1,T_2}(t_{1l}^j,t_{2l}^j\mid Z_i)$.}
	\label{fig:MAEs_correct_model_10k_bivar_cens}
\end{figure}

\begin{figure}[p]
	\centering
	\begin{subfigure}[b]{1\textwidth}
		\centering
		\includegraphics[width=1\textwidth]{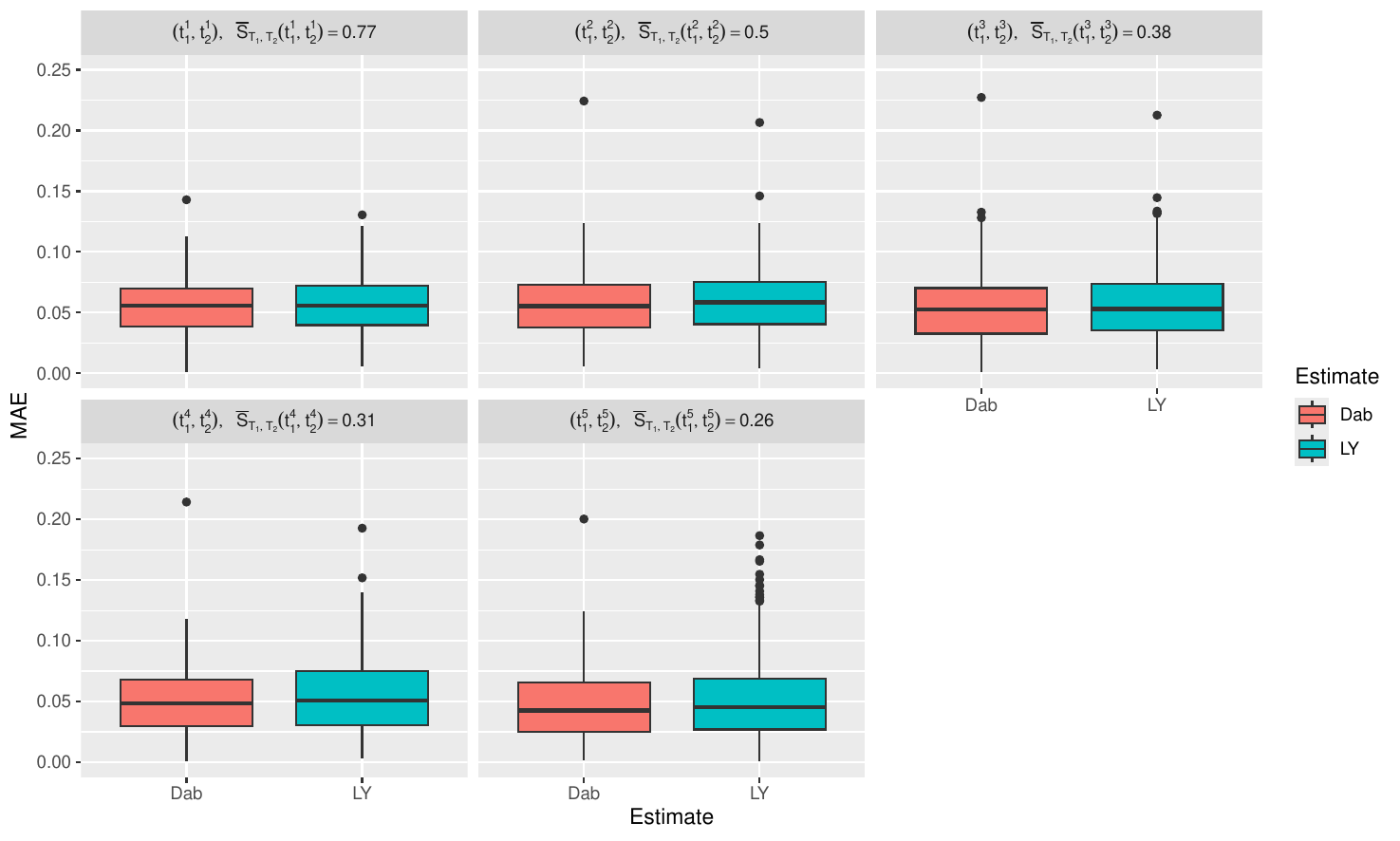}
		\caption{MAEs} \label{fig:MAEs_logitistic_data_5k_univ_n=400}
	\end{subfigure}
	\hfill
	\begin{subfigure}[b]{1\textwidth}
		\centering
		\includegraphics[width=1\textwidth]{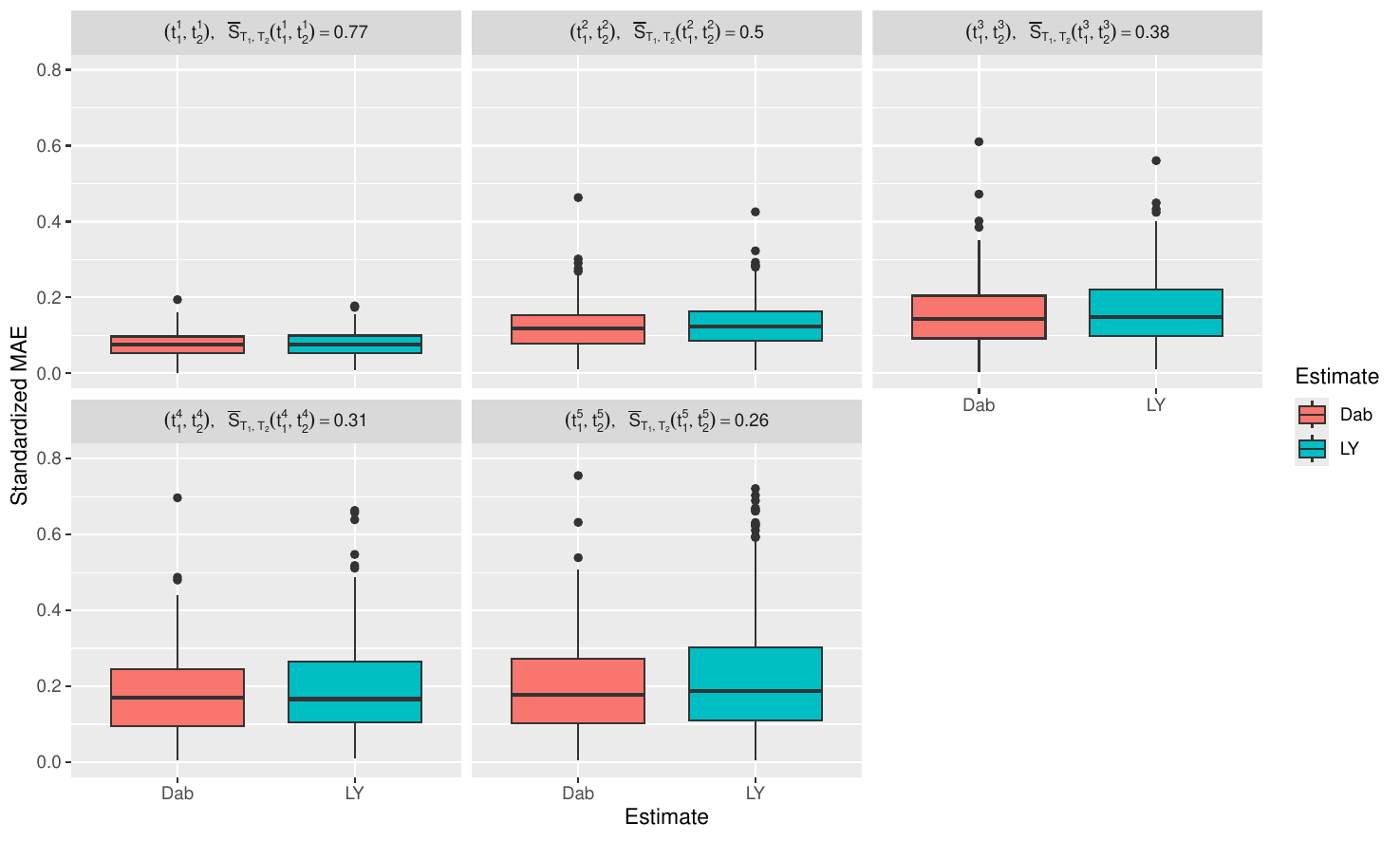}
		\caption{Standardized MAEs} \label{fig:SMAEs_logistic_data_5k_univ_n=400}
	\end{subfigure}
	\caption{Bivariate logistic times, univariate censoring, $n=400$, $k=5$ selected time points, a single regression model .MAEs and standardized MAEs between the true joint survival of the bivariate logistic failure times, and the estimated survival that is based on a single regression model for $k=5$ time points, for both types of estimators and for $m=500$ simulations from the univariate censoring scenario. The top row of each panel specifies the time point $(t_1^j,t_2^j)$ (which is different for each iteration), where $j=1,\ldots, k$ and the mean value of the true joint survival probability is $\bar{S}\equiv \bar{S}_{T_1,T_2}(t_1^j,t_2^j)=\frac{1}{nm}\sum_{i=1}^{n}\sum_{l=1}^{m}S_{T_1,T_2}(t_{1l}^j,t_{2l}^j\mid Z_i)$.}
	\label{fig:MAEs_correct_model_5k_univ_n=400}
\end{figure}

\begin{figure}[p]
	\centering
	\begin{subfigure}[b]{1\textwidth}
		\centering
		\includegraphics[width=1\textwidth]{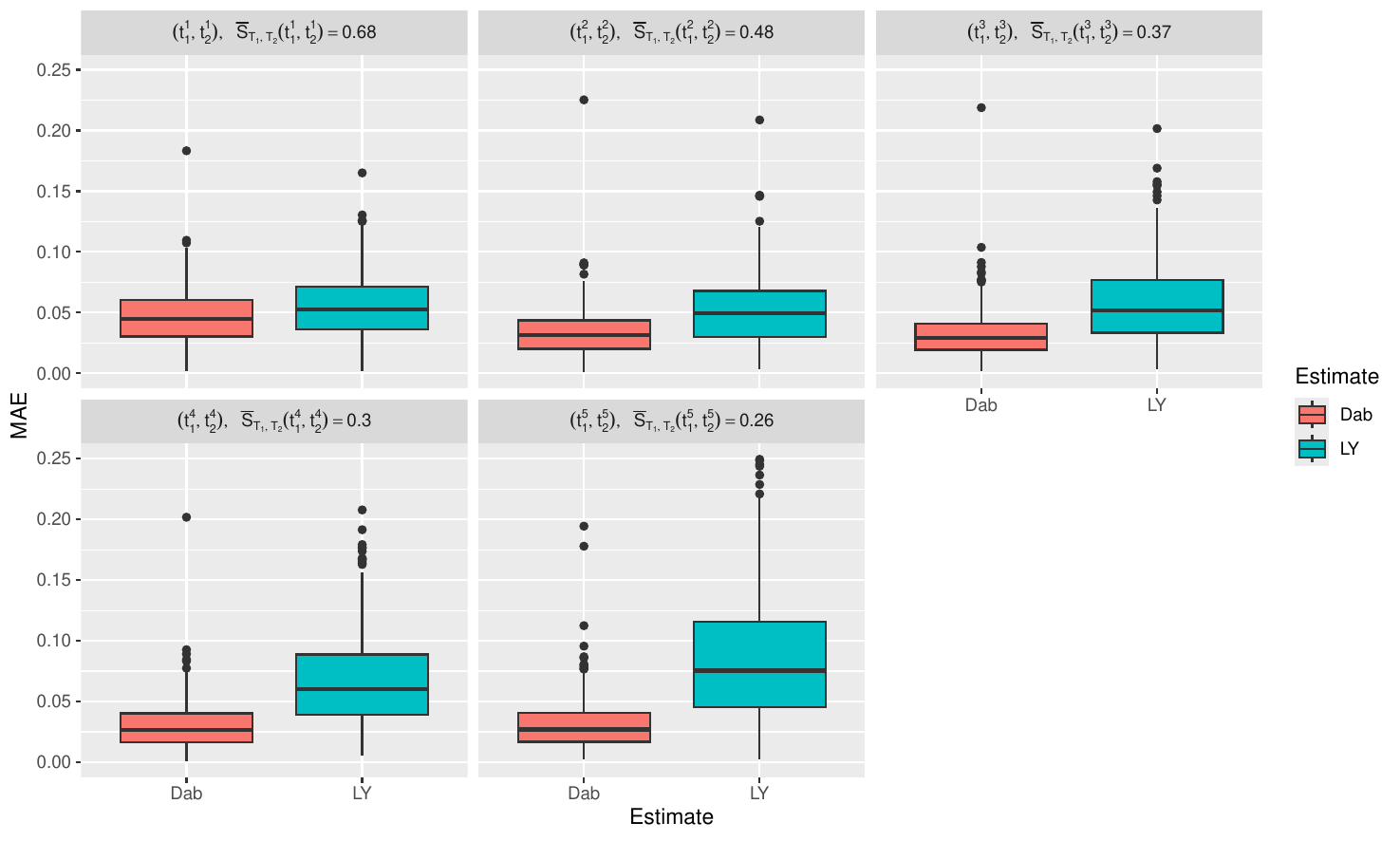}
		\caption{MAEs} \label{fig:MAEs_logitistic_data_5k_bivar_cens_n=400}
	\end{subfigure}
	\hfill
	\begin{subfigure}[b]{1\textwidth}
		\centering
		\includegraphics[width=1\textwidth]{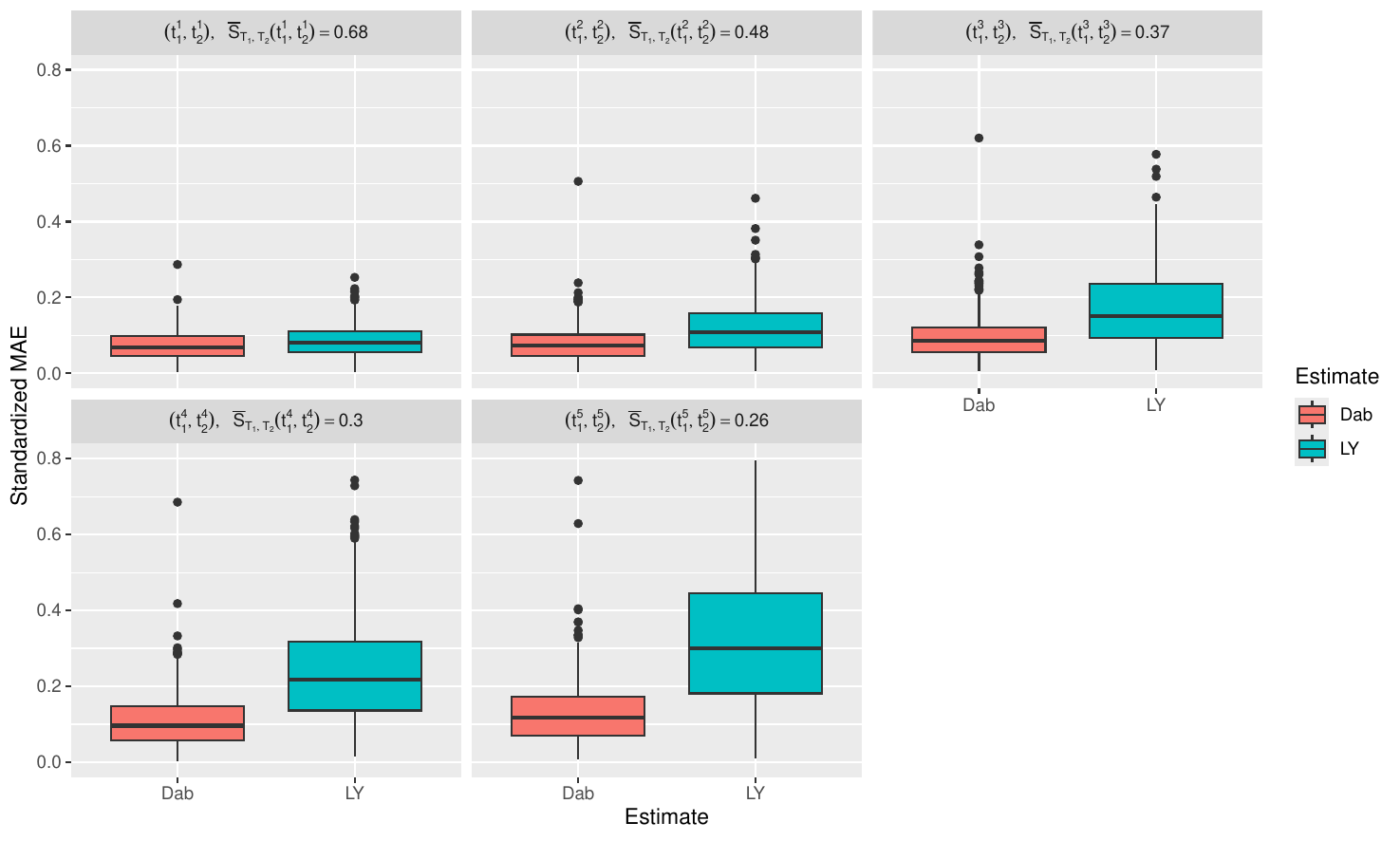}
		\caption{Standardized MAEs} \label{fig:SMAEs_logistic_data_5k_bivar_cens_n=400}
	\end{subfigure}
	\caption{Bivariate logistic times, bivariate censoring, $n=400$, $k=5$ selected time points, a single regression model .MAEs and standardized MAEs between the true joint survival of the bivariate logistic failure times, and the estimated survival that is based on a single regression model for $k=5$ time points, for both types of estimators and for $m=500$ simulations from the bivariate censoring scenario. The top row of each panel specifies the time point $(t_1^j,t_2^j)$ (which is different for each iteration), where $j=1,\ldots, k$ and the mean value of the true joint survival probability is $\bar{S}\equiv \bar{S}_{T_1,T_2}(t_1^j,t_2^j)=\frac{1}{nm}\sum_{i=1}^{n}\sum_{l=1}^{m}S_{T_1,T_2}(t_{1l}^j,t_{2l}^j\mid Z_i)$.}
	\label{fig:MAEs_correct_model_5k_bivar_cens_n=400}
\end{figure}

\begin{figure}[p]
	\centering
	\begin{subfigure}[b]{1\textwidth}
		\centering
		\includegraphics[width=1\textwidth]{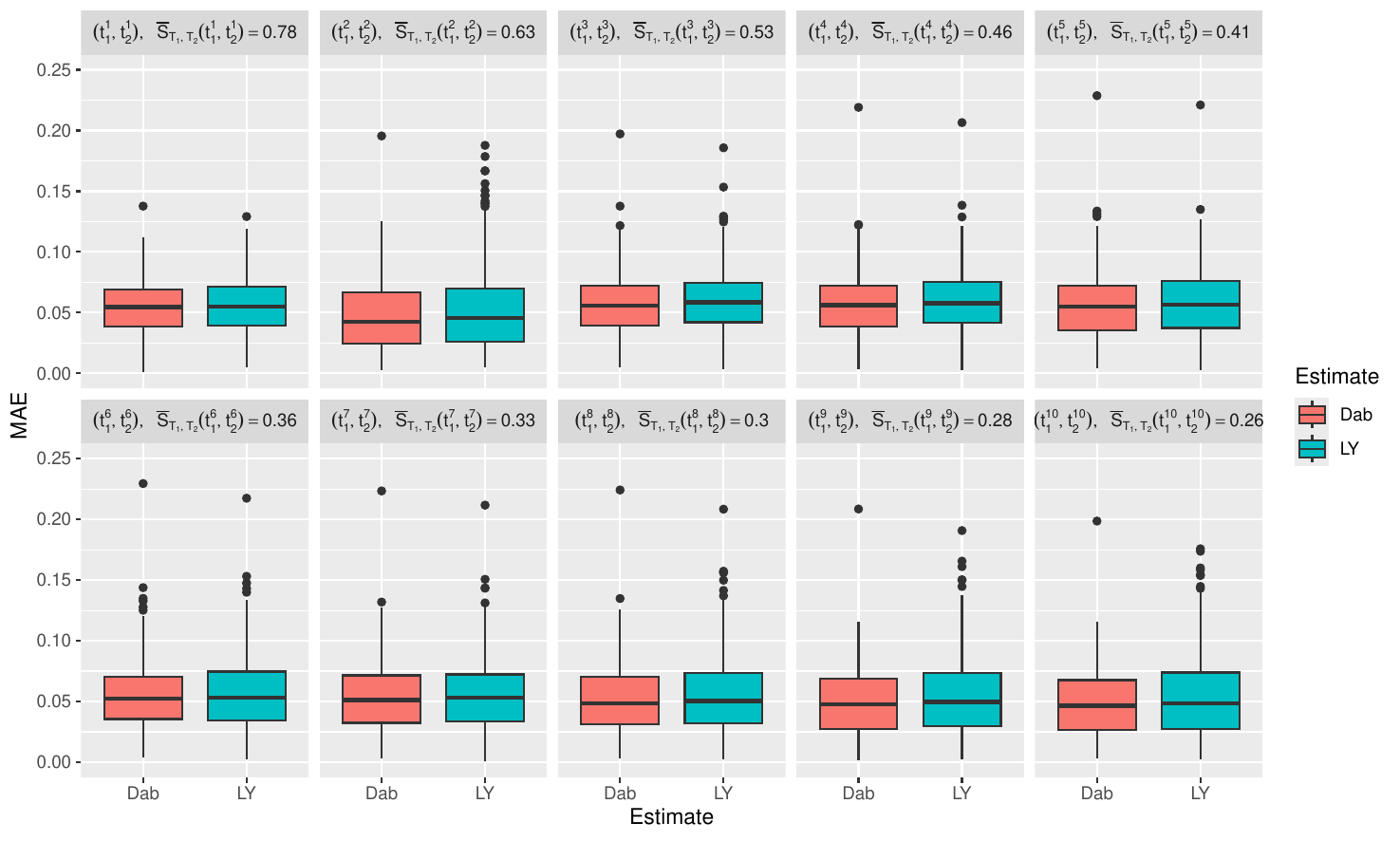}
		\caption{MAEs} \label{fig:MAEs_logitistic_data_10k_univ_n=400}
	\end{subfigure}
	\hfill
	\begin{subfigure}[b]{1\textwidth}
		\centering
		\includegraphics[width=1\textwidth]{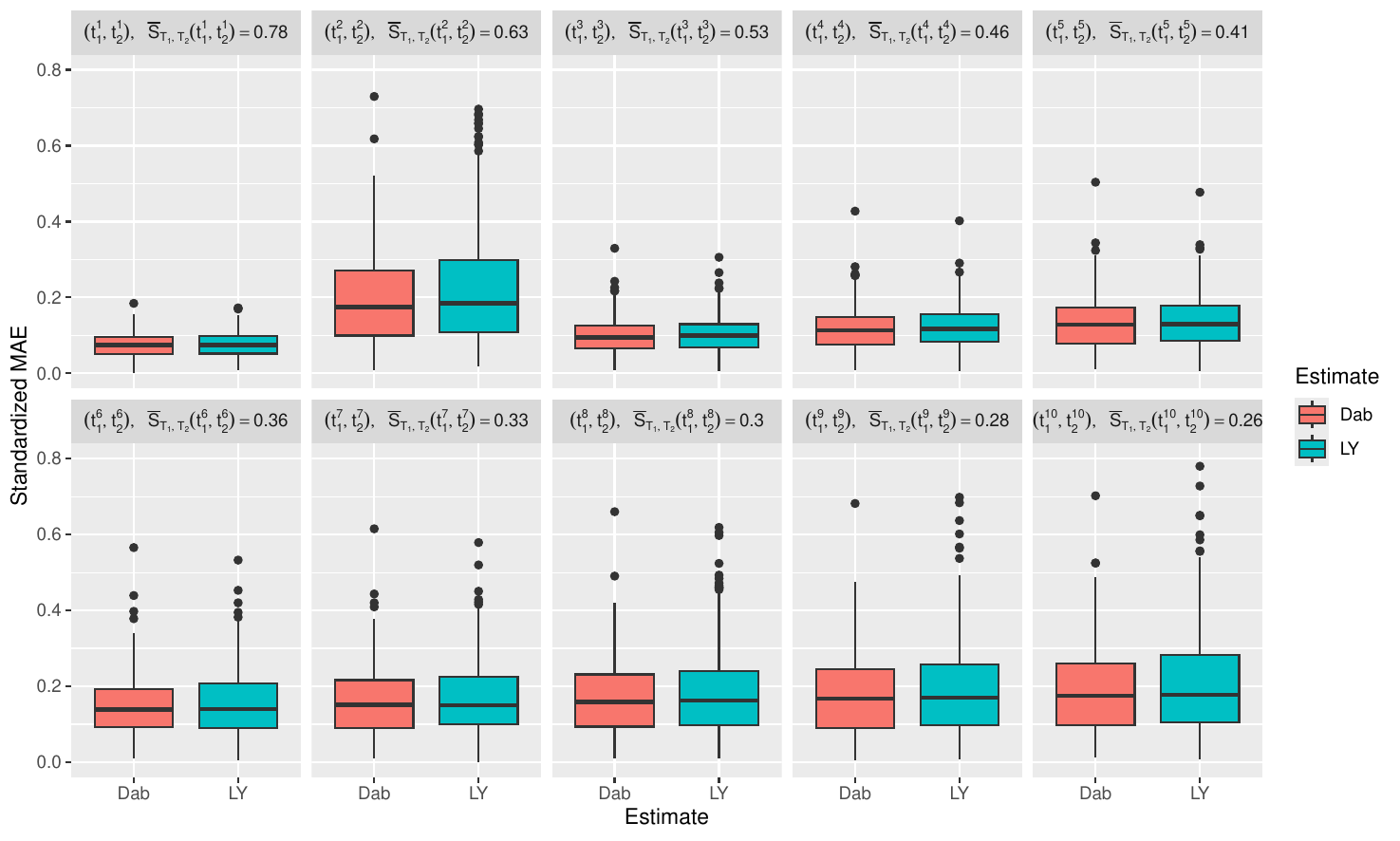}
		\caption{Standardized MAEs} \label{fig:SMAEs_logistic_data_10k_univ_n=400}
	\end{subfigure}
	\caption{Bivariate logistic times, univariate censoring, $n=400$, $k=10$ selected time points, a single regression model .MAEs and standardized MAEs between the true joint survival of the bivariate logistic failure times, and the estimated survival that is based on a single regression model for $k=105$ time points, for both types of estimators and for $m=500$ simulations from the univariate censoring scenario. The top row of each panel specifies the time point $(t_1^j,t_2^j)$ (which is different for each iteration), where $j=1,\ldots, k$ and the mean value of the true joint survival probability is $\bar{S}\equiv \bar{S}_{T_1,T_2}(t_1^j,t_2^j)=\frac{1}{nm}\sum_{i=1}^{n}\sum_{l=1}^{m}S_{T_1,T_2}(t_{1l}^j,t_{2l}^j\mid Z_i)$.}
	\label{fig:MAEs_correct_model_10k_univ_n=400}
\end{figure}

\begin{figure}[p]
	\centering
	\begin{subfigure}[b]{1\textwidth}
		\centering
		\includegraphics[width=1\textwidth]{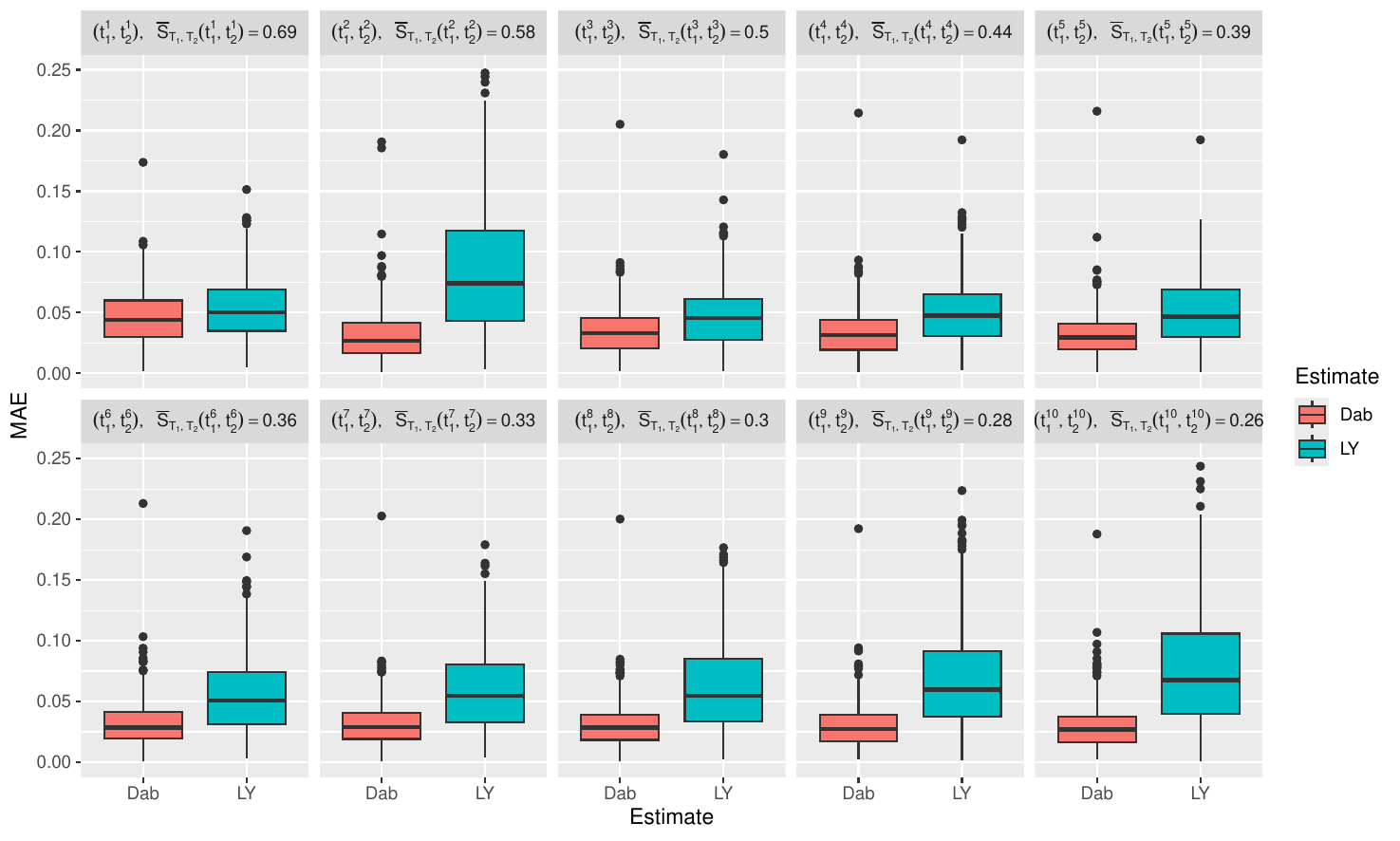}
		\caption{MAEs} \label{fig:MAEs_logitistic_data_10k_bivar_cens_n=400}
	\end{subfigure}
	\hfill
	\begin{subfigure}[b]{1\textwidth}
		\centering
		\includegraphics[width=1\textwidth]{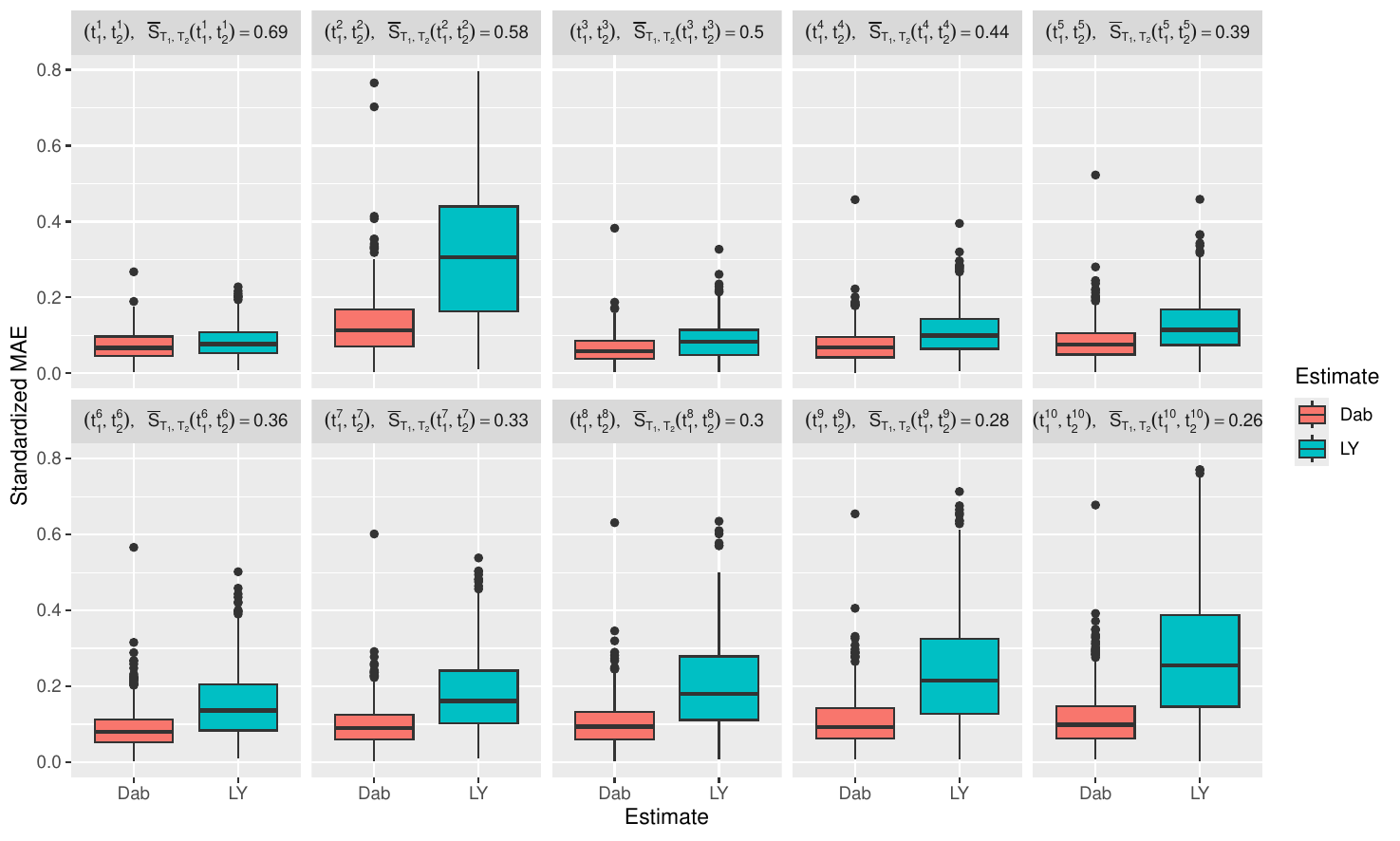}
		\caption{Standardized MAEs} \label{fig:SMAEs_logistic_data_10k_bivar_cens_n=400}
	\end{subfigure}
	\caption{Bivariate logistic times, bivariate censoring, $n=400$, $k=10$ selected time points, a single regression model .MAEs and standardized MAEs between the true joint survival of the bivariate logistic failure times, and the estimated survival that is based on a single regression model for $k=10$ time points, for both types of estimators and for $m=500$ simulations from the bivariate censoring scenario. The top row of each panel specifies the time point $(t_1^j,t_2^j)$ (which is different for each iteration), where $j=1,\ldots, k$ and the mean value of the true joint survival probability is $\bar{S}\equiv \bar{S}_{T_1,T_2}(t_1^j,t_2^j)=\frac{1}{nm}\sum_{i=1}^{n}\sum_{l=1}^{m}S_{T_1,T_2}(t_{1l}^j,t_{2l}^j\mid Z_i)$.}
	\label{fig:MAEs_correct_model_10k_bivar_cens_n=400}
\end{figure}
\begin{figure}[p]
	\centering
	\begin{subfigure}[b]{1\textwidth}
		\centering
		\includegraphics[width=1\textwidth]{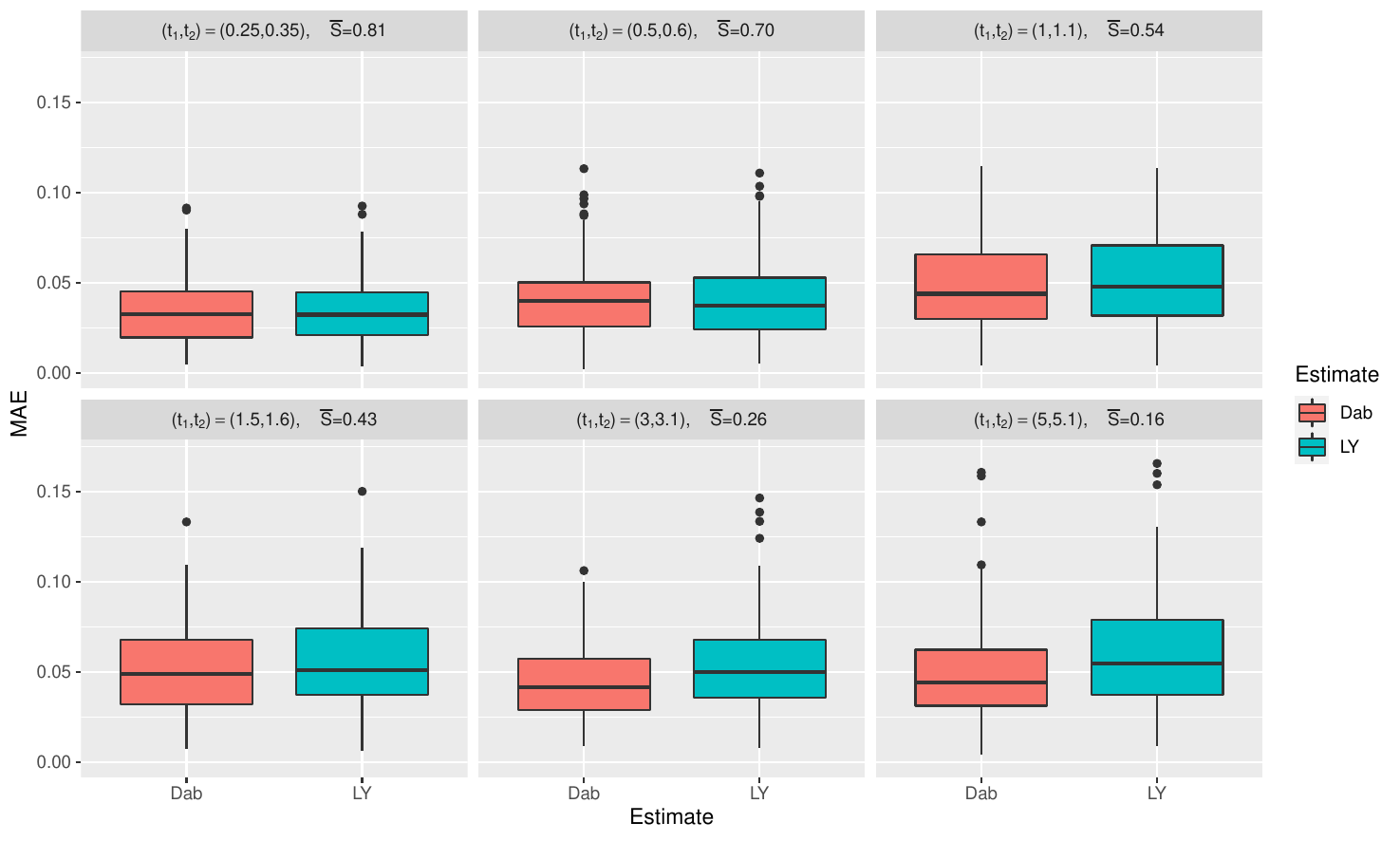}
		\caption{MAEs} \label{fig:MAEs_logit_LN}
	\end{subfigure}
	\hfill
	\begin{subfigure}[b]{1\textwidth}
		\centering
		\includegraphics[width=1\textwidth]{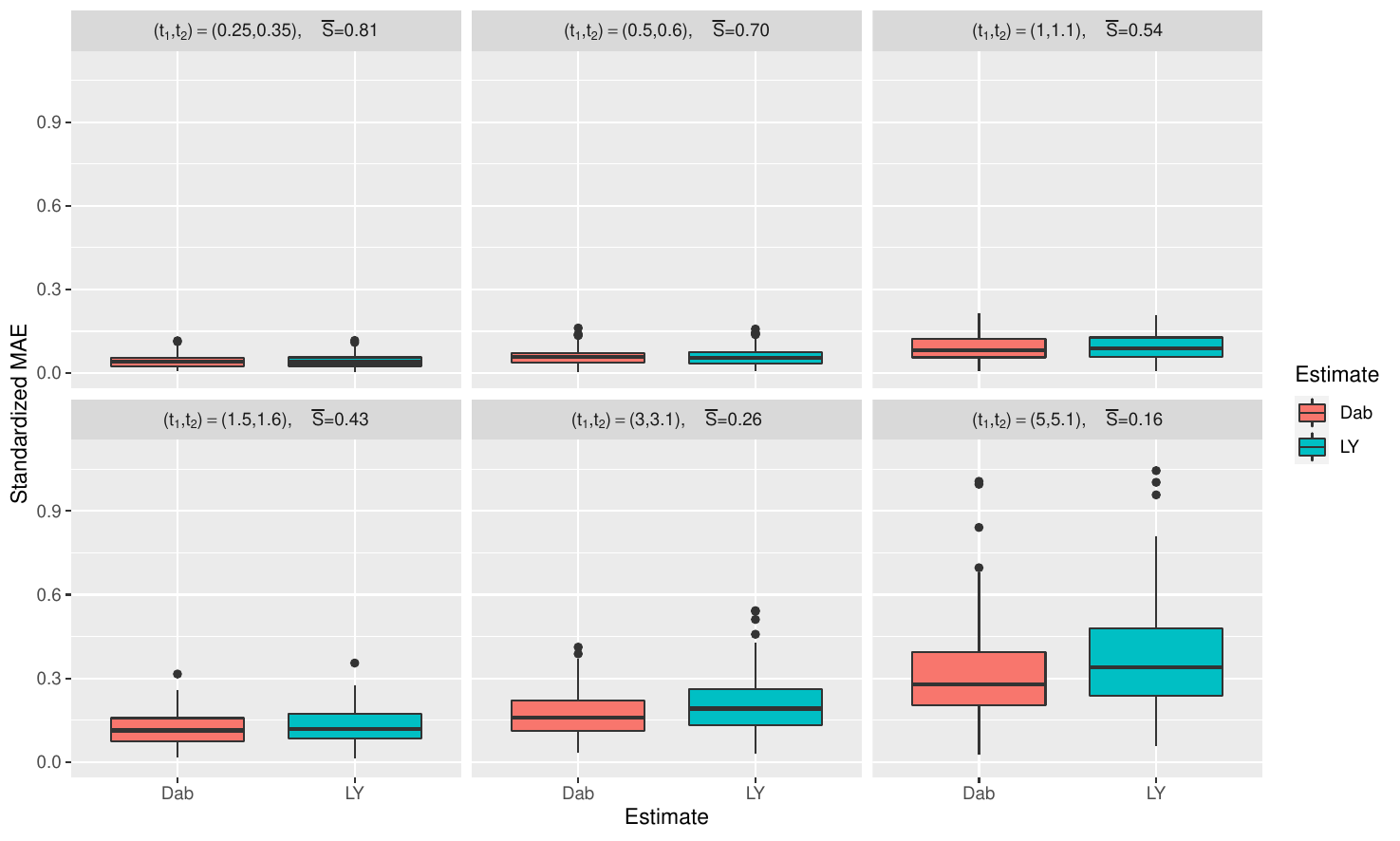}
		\caption{Standardized MAEs} \label{fig:SMAEs_logit_LN}
	\end{subfigure}
	\caption{Bivariate log-normal times, univariate censoring, $n=200$, six fixed time points, six separate regression models. MAEs and standardized MAEs between the true joint survival of the bivariate log-normal failure times, and the estimated survival that is based on six separate regression models (one for each time point), for both types of estimators and for 100 simulations from the univariate censoring scenario. The top row of each panel specifies the time point $(t_1^j,t_2^j)$, and the mean value of the true joint survival probability $\bar{S}\equiv \bar{S}_{T_1,T_2}(t_1^j,t_2^j)=\frac{1}{n}\sum_{i=1}^{n}S_{T_1,T_2}(t_1^j,t_2^j\mid Z_i)$.}
	\label{fig:MAEs}
\end{figure}

\begin{figure}[p]
	\centering
	\begin{subfigure}[b]{1\textwidth}
		\centering
		\includegraphics[width=1\textwidth]{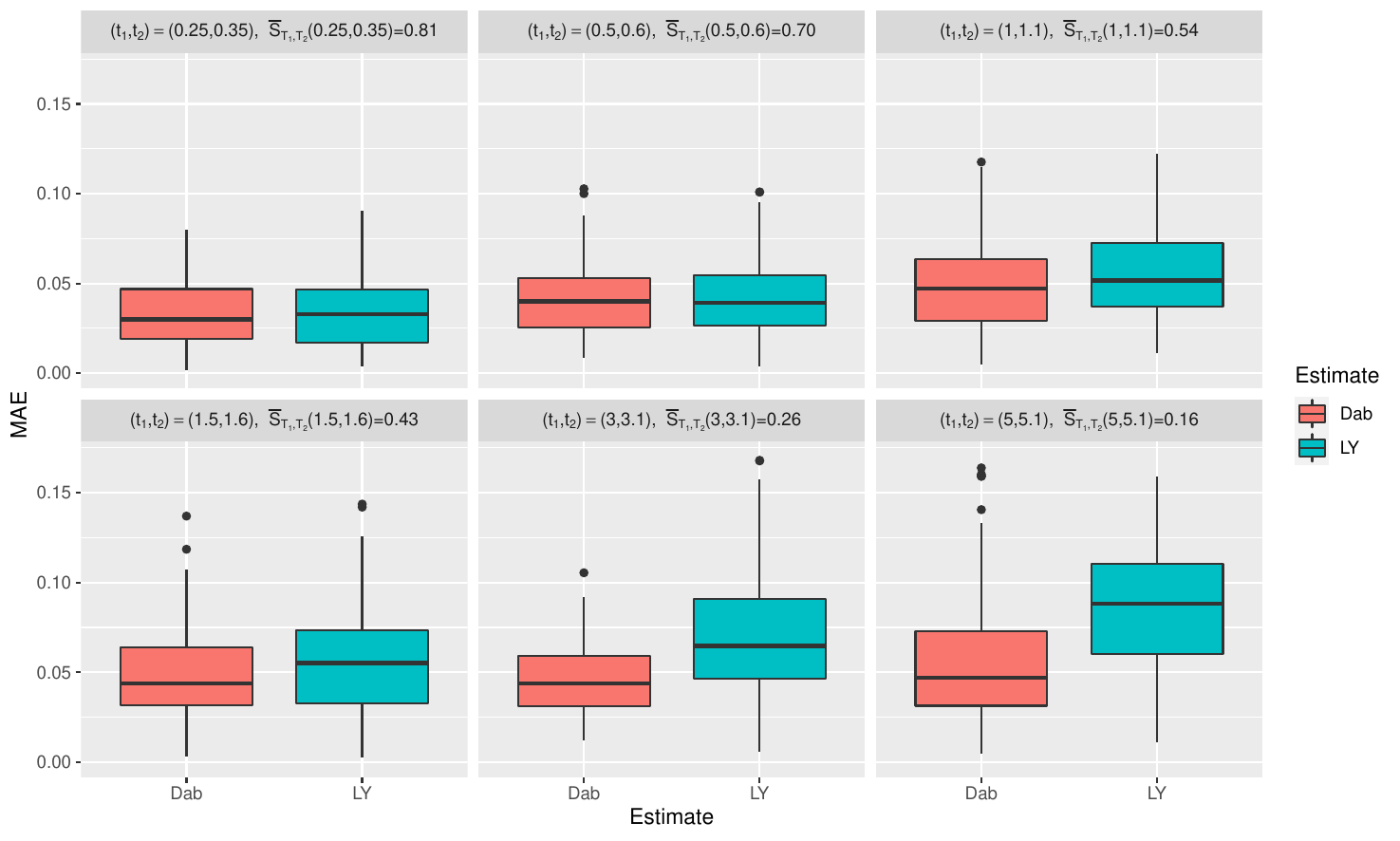}
		\caption{MAEs} \label{fig:MAEs_logit_LN_bivar_cens}
	\end{subfigure}
	\hfill
	\begin{subfigure}[b]{1\textwidth}
		\centering
		\includegraphics[width=1\textwidth]{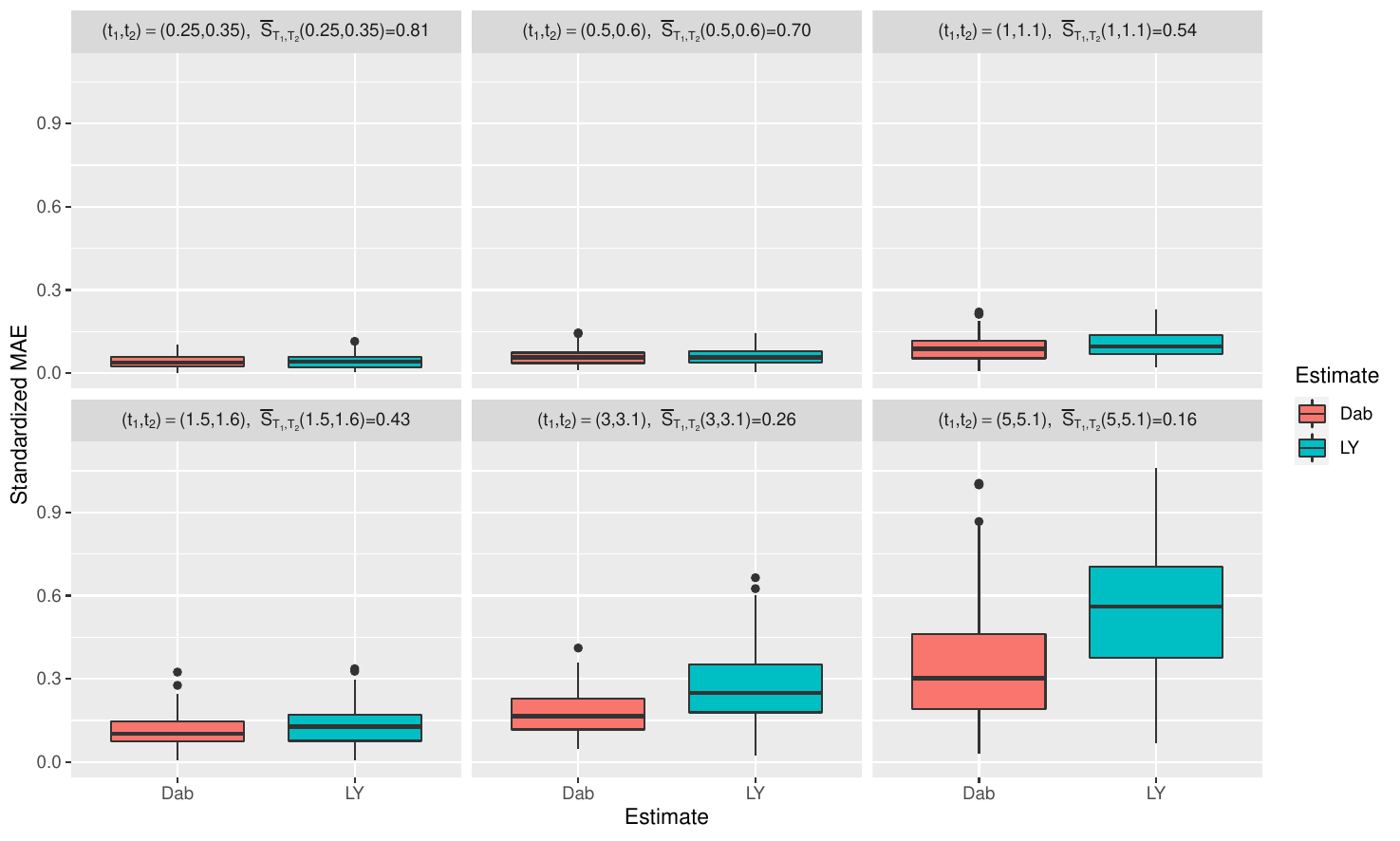}
		\caption{Standardized MAEs} \label{fig:SMAEs_logit_LN_bivar_cens}
	\end{subfigure}
	\caption{Bivariate log-normal times, bivariate censoring, $n=200$, six fixed time points, six separate regression models. MAEs and standardized MAEs between the true joint survival of the bivariate log-normal failure times, and the estimated survival that is based on six separate regression models (one for each time point), for both types of estimators and for 100 simulations from the bivariate censoring scenario. The top row of each panel specifies the time point $(t_1^j,t_2^j)$, and the mean value of the true joint survival probability $\bar{S}\equiv \bar{S}_{T_1,T_2}(t_1^j,t_2^j)=\frac{1}{n}\sum_{i=1}^{n}S_{T_1,T_2}(t_1^j,t_2^j\mid Z_i)$.}
	\label{fig:MAEs_bivar_cens}
\end{figure}

\begin{figure}[p]
	\centering
	\begin{subfigure}[b]{1\textwidth}
		\centering
		\includegraphics[width=1\textwidth]{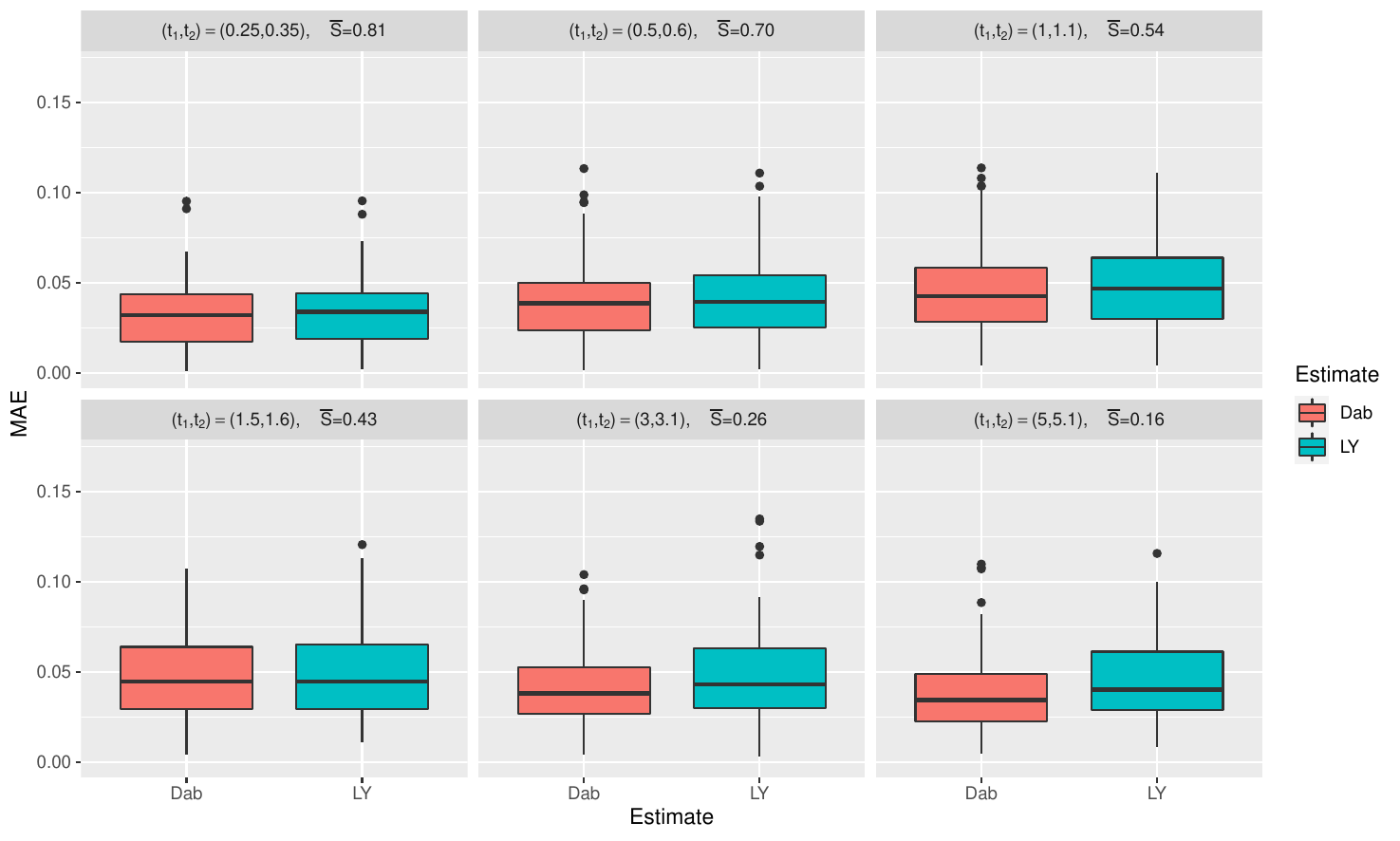}
		\caption{MAE} \label{fig:MAEs_logit_LN_sing_reg}
	\end{subfigure}
	\hfill
	\begin{subfigure}[b]{1\textwidth}
		\centering
		\includegraphics[width=1\textwidth]{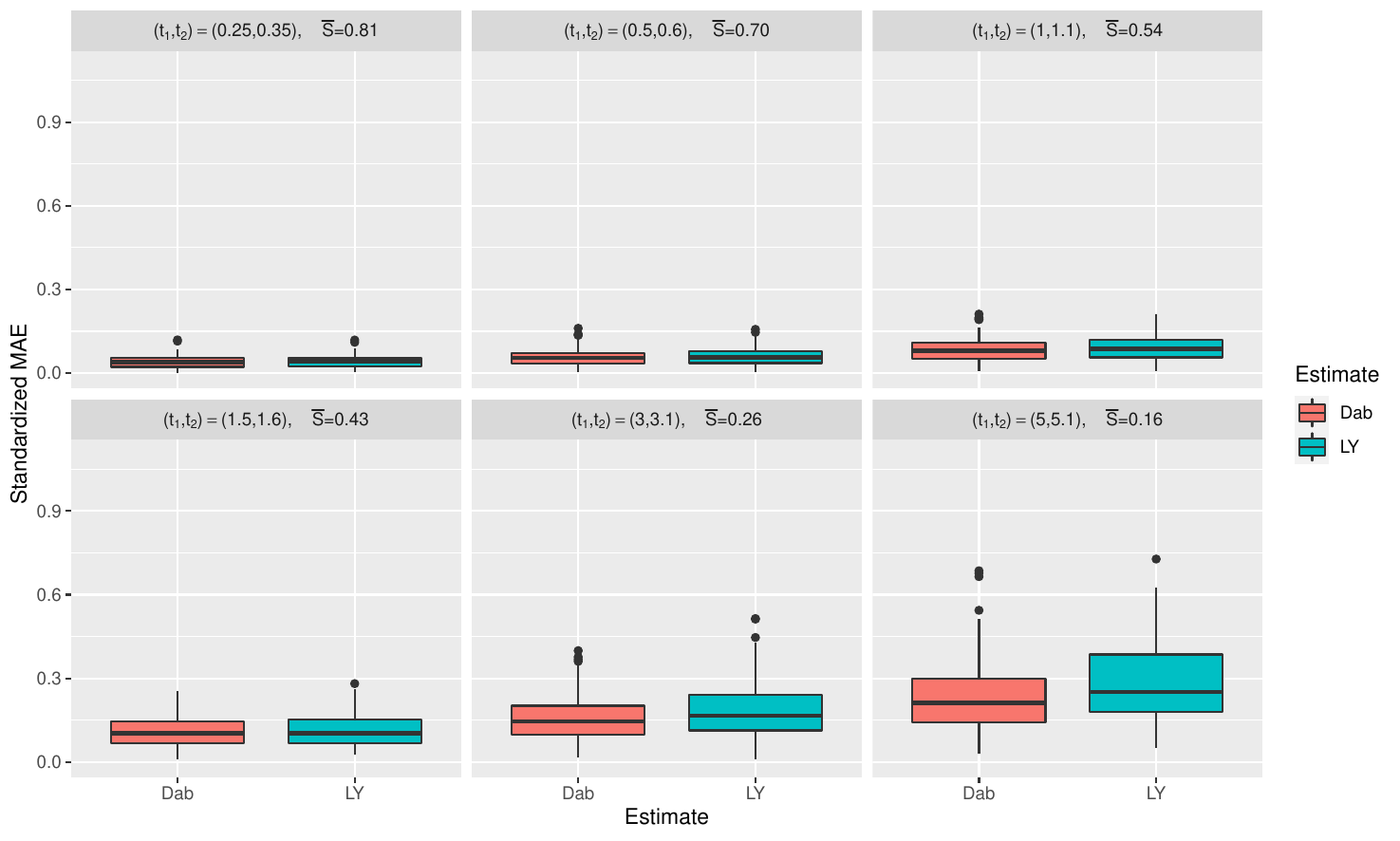}
		\caption{Standardized MAEs} \label{fig:SMAEs_logit_LN_sing_reg}
	\end{subfigure}
	\caption{Bivariate log-normal times, univariate censoring, $n=200$, six fixed time points, a single regression model. MAEs and standardized MAEs between the true joint survival of the bivariate log-normal failure times, and the estimated survival that is based on a single regression model for all six time points, for both types of estimators and for 100 simulations from the univariate censoring scenario. The top row of each panel specifies the time point $(t_1^j,t_2^j)$, and the mean value of the true joint survival probability $\bar{S}\equiv \bar{S}_{T_1,T_2}(t_1^j,t_2^j)=\frac{1}{n}\sum_{i=1}^{n}S_{T_1,T_2}(t_1^j,t_2^j\mid Z_i)$.}
	\label{fig:MAEs_sing_reg}
\end{figure}

\begin{figure}[p]
	\centering
	\begin{subfigure}[b]{1\textwidth}
		\centering
		\includegraphics[width=1\textwidth]{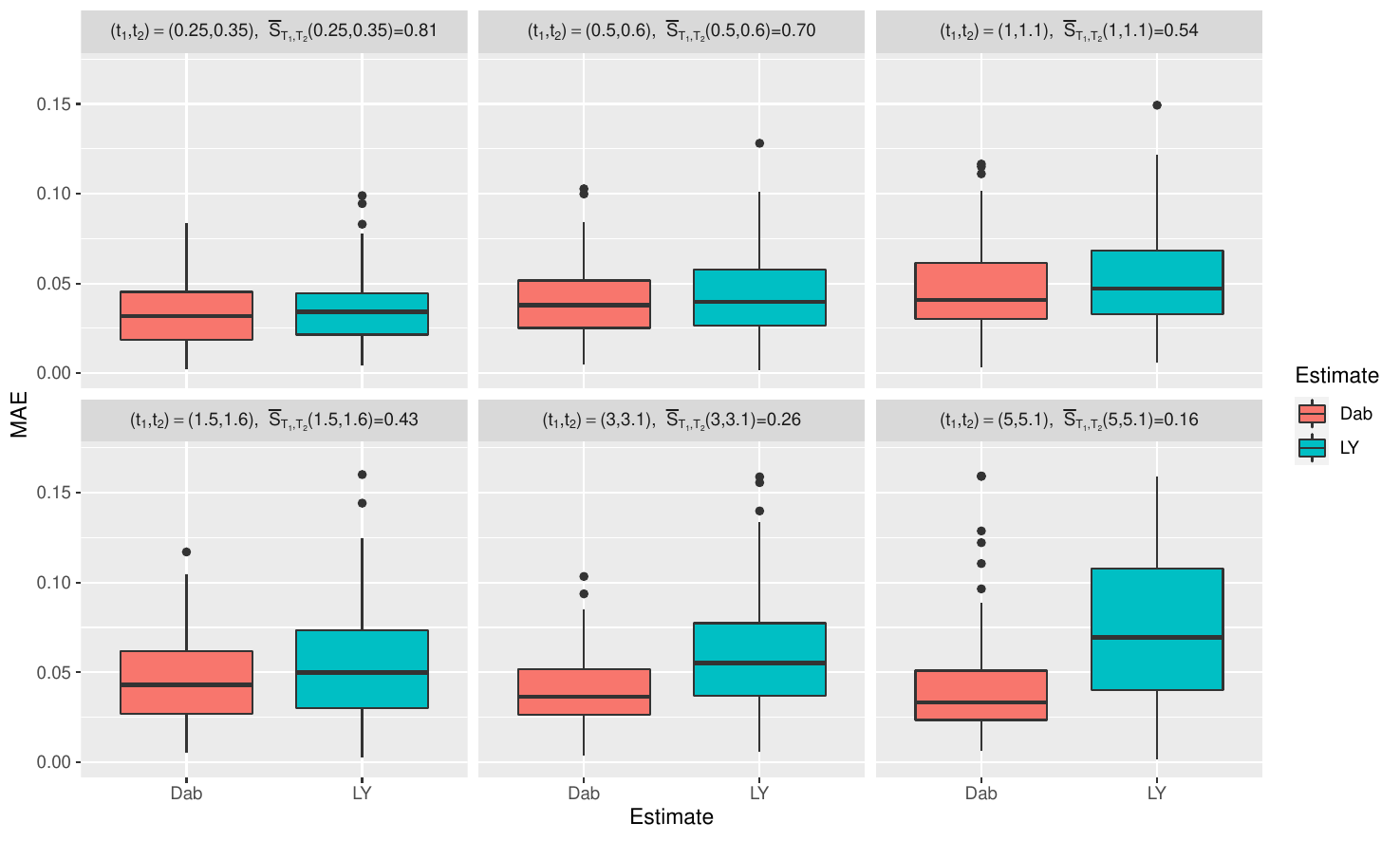}
		\caption{MAE} \label{fig:MAEs_logit_LN_sing_reg_bivar_cens}
	\end{subfigure}
	\hfill
	\begin{subfigure}[b]{1\textwidth}
		\centering
		\includegraphics[width=1\textwidth]{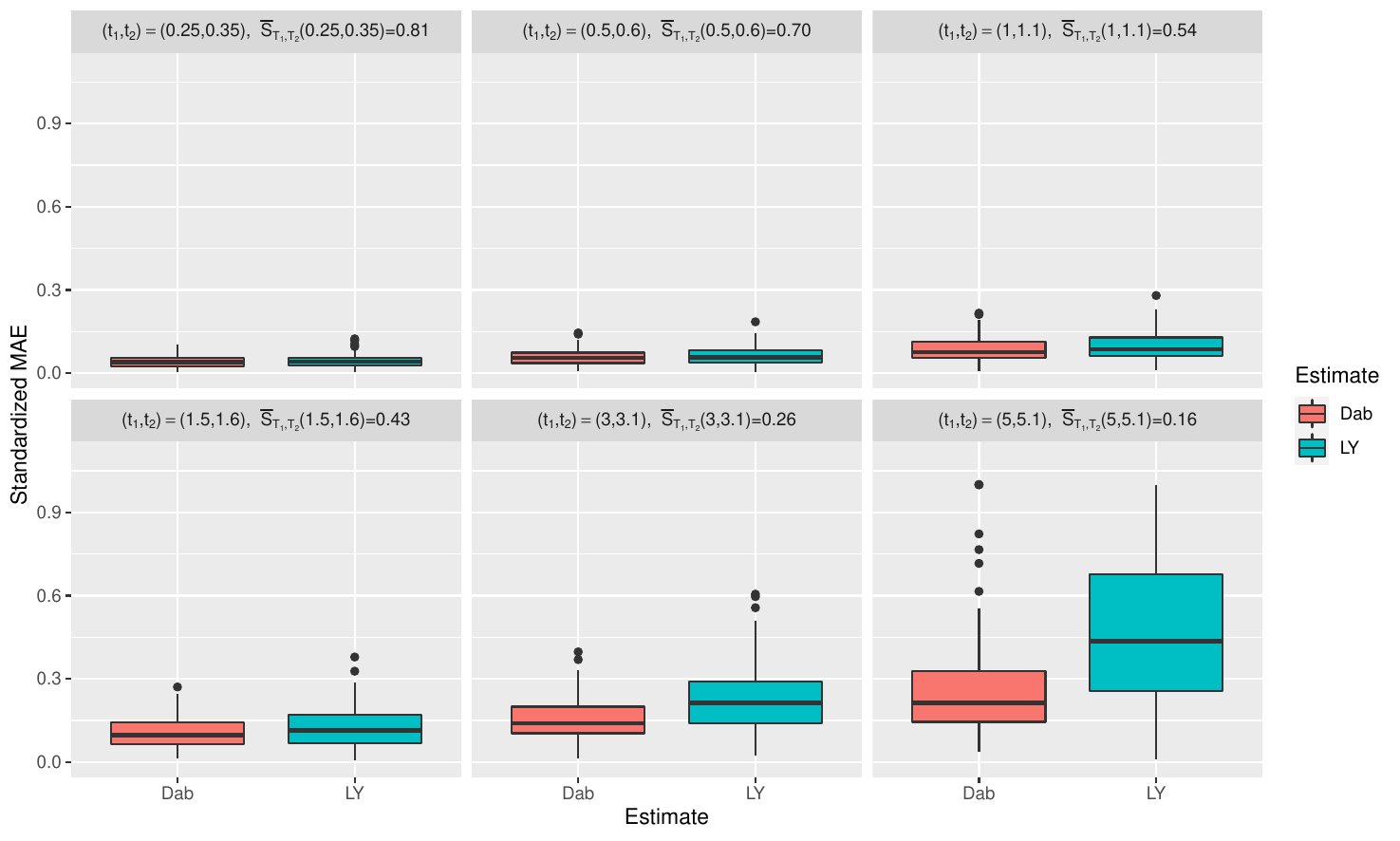}
		\caption{Standardized MAEs} \label{fig:SMAEs_logit_LN_sing_reg_bivar_cens}
	\end{subfigure}
	\caption{Bivariate log-normal times, bivariate censoring, $n=200$, six fixed time points, a single regression model. MAEs and standardized MAEs between the true joint survival of the bivariate log-normal failure times, and the estimated survival that is based on a single regression model for all six time points, for both types of estimators and for 100 simulations from the bivariate censoring scenario. The top row of each panel specifies the time point $(t_1^j,t_2^j)$, and the mean value of the true joint survival probability $\bar{S}\equiv \bar{S}_{T_1,T_2}(t_1^j,t_2^j)=\frac{1}{n}\sum_{i=1}^{n}S_{T_1,T_2}(t_1^j,t_2^j\mid Z_i)$.}
	\label{fig:MAEs_sing_reg_bivar_cens}
\end{figure}

\section{}
Here we test the effect of using a different time point selection procedure.
Previous research found that usually five time points are adequate, and there is hardly any gain in precision when selecting more than three time points \citep{furberg_bivariate_2023}. To test this, we used the bivariate time point selection procedure described in Section 3.2 of the main manuscript with both $k=5$ and $k=10$ time points. We repeated the data-generating procedure described in Section~3, with both univariate censoring and bivariate censoring (in which case we modified the Lin and Ying estimator accordingly). We used the univariate 20th and 30th empirical percentiles as the lower limits for our time point selection method, for the univariate and bivariate censoring scenarios, respectively. For each censoring mechanism, the pair of lower limits corresponded to the 10th empirical percentile of the observed bivariate times. We used the univariate 90th empirical percentile as the upper limit, which corresponded to the 83rd empirical percentile of the observed bivariate times, for both the univariate and bivariate censoring scenarios.

For the regression model, we used Equation~(7) of the main manuscript with either $k=5$ or $k=10$ time points. The estimated regression coefficient $\beta_1$, together with sd, se and coverage, is summarized in Web Table~\ref{table:equidistant}, for sample sizes $n=200$ (top) and $n=400$ (bottom). The estimated intercepts are not presented, as each realization in the simulation resulted in a slightly different selection of time points. The MAEs and the standardized MAEs are presented in Web Figures~\ref{fig:MAEs_correct_model_5k_univ}-\ref{fig:MAEs_correct_model_10k_bivar_cens_n=400}.

Web Table~\ref{table:equidistant} shows that for a sample size of $n=200$, our bivariate pseudo-observations approach based on the Dabrowska estimator is consistently better than that based on the Lin and Ying estimator, both in terms of lower bias and lower variance. Additionally, for the Dabrowska estimator there is hardly any difference between the univariate and bivariate censoring scenarios, or between $k=5$ and $k=10$ time points. However, the coverage is slightly lower than 95\%, especially when the analysis is based on $k=5$ time points. For the Lin and Ying estimator, the bias and variance of the estimated regression parameter is much higher in the bivariate censoring scenario than in the univariate censoring scenario, yet the coverage is generally correct. 
Furthermore, for a sample size of $n=200$ and for both estimators, increasing the number of time points from $k=5$ to $k=10$ seems to offer some improvement in variance values and in achieving correct coverage.
When increasing the sample size from $n=200$ to $n=400$, we see that the discussion above is still valid, with some distinctions: (1) the bias and variance of the estimated regression parameter decreases when the sample size increases, (2) for the Lin and Ying estimator, there is no longer much difference in bias between the univariate and bivariate censoring scenarios, (3) for both estimators, the coverage is generally correct or conservative, and (4) increasing the number of time points offers only a slight improvement in variance values and in achieving correct coverage.
Finally, for both sample sizes and for each combination of estimator, censoring scenario, and number of time points $k$, the standard deviations are greater than or equal to the standard errors. Moreover, when comparing the results based on this bivariate time point selection method to the results based on the fixed time point selection approach presented in Table~1 of the main manuscript, we see that there is hardly any difference in the regression estimates that are based on the Dabrowska estimator. However, for a sample of size $n=200$, the Lin and Ying estimator is more sensitive to the time point selection approach.

To better understand why the modified estimator of Lin and Ying isn't working so well in the bivariate censoring case with a low sample size, we also checked the effect of either lowering the censoring rate, or lowering the upper limit from the 90th percentile to the 80th percentile.
We use the same bivariate logistic data generating mechanism with bivariate independent censoring, a sample of size $n \in \{200,400\}$, and our time point selection procedure with either $k=5$ or $k=10$ time points. 
To test the effect of lowering the censoring rates on the results, we use two independent exponential censoring variables, with rates $\lambda_1=0.1$, and $\lambda_2=0.05$, which correspond to about $50\%$ censoring of $T_1$, and about $20\%$ censoring of $T_2$ (compared to the original censoring rates of 72\% and 48\%, respectively). When compared to Web Table~\ref{table:equidistant} , Web Table~\ref{table:equidistant_all} shows that lowering the censoring rate in the bivariate censoring case increases the estimation accuracy of the modified Lin and Ying estimator. We also consider the effect of lowering the univariate upper limit in our bivariate time point selection method, from the 90th percentile to the 80th percentile. Note that the empirical bivariate CDF of the observed times $\hat{F}(t_1,t_2)=\frac{1}{n}\sum_{i=1}^{n}I\left(Y_{1i}\leq t_1, Y_{2i}\leq t_2\right)$ satisfies $\hat{F}(t_1^{80},t_2^{80})\approxeq 0.66$, whereas before we had $\hat{F}(t_1^{90},t_2^{90})\approxeq 0.83$. Web Table~\ref{table:equidistant_all} shows that for this specific data generating mechanism, lowering the upper limit improves estimation for the Lin and Ying estimator, both in terms of lower bias and lower variance. However, it has hardly any effect on the results for the Dabrowska estimator, although the coverage is slightly lower then expected when $n=200$ (when compared to Web Table~\ref{table:equidistant}).

\begin{table}[ht]
	\begin{subtable}{1\textwidth}
		\centering
		\subcaption{Sample size $n=200$}\label{tab:log_equidistant_200}
		\centering
		\begin{tabular}{llrrrrrr}
			\hline
			Estimator & Censoring & $k$ & True value & est & sd & se & cov (\%)\\ 
			\hline
			Dabrowska & Univariate & 5.00 & 2.00 & 2.06 & 0.73 & 0.62 & 93.60 \\ 
			Dabrowska & Univariate & 10.00 & 2.00 & 2.06 & 0.71 & 0.61 & 94.80 \\ 
			Dabrowska & Bivariate & 5.00 & 2.00 & 2.07 & 0.73 & 0.67 & 94.20 \\ 
			Dabrowska & Bivariate & 10.00 & 2.00 & 2.05 & 0.68 & 0.62 & 95.00 \\ 
			Lin \& Ying & Univariate & 5.00 & 2.00 & 2.10 & 0.93 & 0.81 & 94.40 \\ 
			Lin \& Ying & Univariate & 10.00 & 2.00 & 2.10 & 0.89 & 0.79 & 95.00 \\ 
			Lin \& Ying & Bivariate & 5.00 & 2.00 & 2.24 & 1.62 & 1.50 & 95.20 \\ 
			Lin \& Ying & Bivariate & 10.00 & 2.00 & 2.18 & 1.34 & 1.29 & 95.20 \\ 
			\hline
		\end{tabular}
	\end{subtable}
	\hfill\\
	\hfill\\
	\begin{subtable}{1\textwidth}
		\centering
		\subcaption{Sample size $n=400$}\label{tab:log_equidistant_400}
		\centering
		\begin{tabular}{llrrrrrr}
			\hline
			Estimator & Censoring & $k$ & True value & est & sd & se & cov (\%)\\ 
			\hline
			Dabrowska & Univariate & 5.00 & 2.00 & 2.06 & 0.52 & 0.42 & 95.40 \\ 
			Dabrowska & Univariate & 10.00 & 2.00 & 2.05 & 0.51 & 0.41 & 95.40 \\ 
			Dabrowska & Bivariate & 5.00 & 2.00 & 2.03 & 0.52 & 0.42 & 95.00 \\ 
			Dabrowska & Bivariate & 10.00 & 2.00 & 2.03 & 0.50 & 0.41 & 95.40 \\ 
			Lin \& Ying & Univariate & 5.00 & 2.00 & 2.07 & 0.59 & 0.55 & 96.60 \\ 
			Lin \& Ying & Univariate & 10.00 & 2.00 & 2.07 & 0.57 & 0.53 & 97.00 \\ 
			Lin \& Ying & Bivariate & 5.00 & 2.00 & 2.11 & 0.86 & 0.85 & 95.20 \\ 
			Lin \& Ying & Bivariate & 10.00 & 2.00 & 2.10 & 0.81 & 0.81 & 94.80 \\ 
			\hline
		\end{tabular}
	\end{subtable}
	\hfill\\
	\caption{Top: sample size $n=200$. Bottom: sample size $n=400$. Mean of the estimated slope parameter $\beta_1$, together with sd, se and coverage, based on a single regression model with either $k=5$ time points or $k=10$ time points, univariate or bivariate censoring, for both the Dabrowska estimator and the Lin and Ying estimator. The estimates are based on $m=500$ simulations.}\label{table:equidistant}
\end{table}

\begin{table}[ht]
	\begin{subtable}{1\textwidth}
		\centering
		\subcaption{Sample size $n=200$}\label{tab:log_equidistant_200}
		\centering
		\begin{tabular}{llrrrrrr}
			\hline
			Estimator & Type & $k$ & True value & est & sd & se & cov (\%)\\ 
			\hline
			Dabrowska & Lower censoring rates & 5.00 & 2.00 & 2.05 & 0.65 & 0.57 & 94.00 \\ 
			Dabrowska & Lower censoring rates & 10.00 & 2.00 & 2.06 & 0.65 & 0.57 & 94.20 \\ 
			Dabrowska & Lower upper limits & 5.00 & 2.00 & 2.04 & 0.63 & 0.57 & 93.60 \\ 
			Dabrowska & Lower upper limits & 10.00 & 2.00 & 2.05 & 0.63 & 0.57 & 93.80 \\ 
			Lin \& Ying & Lower censoring rates & 5.00 & 2.00 & 2.08 & 0.86 & 0.81 & 94.40 \\ 
			Lin \& Ying & Lower censoring rates & 10.00 & 2.00 & 2.09 & 0.87 & 0.79 & 95.00 \\ 
			Lin \& Ying & Lower upper limits & 5.00 & 2.00 & 2.11 & 1.02 & 1.50 & 95.20 \\ 
			Lin \& Ying & Lower upper limits & 10.00 & 2.00 & 2.12 & 1.04 & 1.29 & 95.20 \\ 
			\hline
		\end{tabular}
	\end{subtable}
	\hfill\\
	\hfill\\
	\begin{subtable}{1\textwidth}
		\centering
		\subcaption{Sample size $n=400$}\label{tab:log_equidistant_400}
		\centering
		\begin{tabular}{llrrrrrr}
			\hline
			Estimator & Type & K & True value & est & sd & se & cov (\%)\\ 
			\hline
			Dabrowska & Lower censoring rates & 5.00 & 2.00 & 2.04 & 0.52 & 0.38 & 94.20 \\ 
			Dabrowska & Lower censoring rates & 10.00 & 2.00 & 2.05 & 0.51 & 0.39 & 93.00 \\ 
			Dabrowska & Lower upper limits & 5.00 & 2.00 & 2.02 & 0.48 & 0.38 & 94.60 \\ 
			Dabrowska & Lower upper limits & 10.00 & 2.00 & 2.02 & 0.47 & 0.39 & 95.00 \\ 
			Lin \& Ying & Lower censoring rates & 5.00 & 2.00 & 2.09 & 0.64 & 0.53 & 94.60 \\ 
			Lin \& Ying & Lower censoring rates & 10.00 & 2.00 & 2.10 & 0.62 & 0.53 & 95.20 \\ 
			Lin \& Ying & Lower upper limits & 5.00 & 2.00 & 2.06 & 0.64 & 0.63 & 95.60 \\ 
			Lin \& Ying & Lower upper limits & 10.00 & 2.00 & 2.06 & 0.62 & 0.62 & 96.00 \\ 
			\hline
		\end{tabular}
	\end{subtable}
	\hfill\\
	\caption{Top: sample size $n=200$. Bottom: sample size $n=400$. Mean, sd, se, and coverage of the estimated slope parameter $\beta_1$, for both the Dabrowska estimator and the Lin and Ying estimator. The type column corresponds to two types of modifications: a lower censoring rate and a lower upper percentile for the bivariate time point selection method. These estimates are based on a single regression model for either $k=5$ or $k=10$ time points, and on 500 simulations from the bivariate logistic model with bivariate censoring.}\label{table:equidistant_all}
\end{table}

\section{}\label{sec:LN}
To test for the effect of model misspecification, we consider also data that are generated from a log-normal distribution in which case the logistic regression model is misspecified. We generate a random vector $(Z_1,\ldots,Z_n)^T$ from a uniform $U(0.5,1.5)$ distribution. We define the following two functions: $\mu_1(Z)=\gamma_1 Z$ and  $\mu_2(Z)=\gamma_2 Z$ where we use $\gamma_1=0.2$ and $\gamma_2=1$. For each $i=1,\ldots,n$ we generate a bivariate random variable from the following Gaussian distribution
\[
\begin{pmatrix}
	X_{1i}\\
	X_{2i}
\end{pmatrix}
\sim N\left(\begin{pmatrix}
	\mu_1(Z_i) \\
	\mu_2(Z_i)
\end{pmatrix},
\begin{pmatrix}
	\sigma_1^2 & \rho\sigma_1\sigma_2\\
	\rho\sigma_1\sigma_2 & \sigma_2^2
\end{pmatrix}\right),
\]
where $\sigma_1^2=2$, $\sigma_2^2=5$, and $\rho=0.9487$, which correspond to the covariance matrix \[
\begin{pmatrix}
	\sigma_1^2 & \rho\sigma_1\sigma_2\\
	\rho\sigma_1\sigma_2 & \sigma_2^2
\end{pmatrix}=
\begin{pmatrix}
	2 & 3\\
	3 & 5
\end{pmatrix}.
\]
We then define our bivariate failure times as $(T_{1},T_{2})^T=(e^{X_{1}},e^{X_{2}})^T$. 

Finally, we consider two censoring mechanisms: univariate censoring and independent bivariate censoring. For the univariate censoring scenario, we generate $n$ independent univariate censoring variables from an exponential distribution with rate $\lambda=0.3$, which correspond to about $40\%$ censoring of $T_1$, and about $55\%$ censoring of $T_2$.  For the bivariate independent censoring scenario, we generate $n$ independent univariate censoring variables from an exponential distribution with rate $\lambda=0.3$, and $n$ independent univariate censoring variables from an exponential distribution with rate $\lambda=0.2$, which correspond to about $39\%$ censoring of $T_1$, and about $48\%$ censoring of $T_2$. In both cases, the observed data consists of $(Y_{i1},\Delta_{1i},Y_{21},\Delta_{2i},Z_i)$.

We consider the six time points $\{(0.25,0.35),(0.5,0.6),(1,1.1),(1.5,1.6),(3,3.1),(5,5.1)\}$ and the six functions $f_j(T_1,T_2)=I(T_1>t_1^j, T_2>t_2^j)$, $j=1,\ldots, 6$. We use the logistic regression model with either six separate regression models for the six time points, or a single regression model for all time points.

For each time point, iteration, and type of estimator, we calculate the MAE and the standardized MAE. These are more relevant summary measures of performance, as the regression model studied here is misspecified.
Web Figure~\ref{fig:MAEs} and Web Figure~\ref{fig:MAEs_bivar_cens} presents the boxplot of the MAEs (top) and the standardized MAEs (bottom) for the pseudo-observations approach using six separate regression models and bivariate log-normal failure times, with either univariate censoring or bivariate censoring, respectively. As can be seen by the relatively low MAEs, the bivariate pseudo-observation approach estimates quite well the conditional joint survival probability, for both censoring mechanisms, even under model misspecification. Note also that, on average, the MAEs of the Dabrowska estimator are lower than those of Lin and Ying. Additionally, note that the errors for the modified Lin and Ying estimator in the bivariate censoring scenario are higher than those based on the standard Lin and Ying estimator in the univariate censoring scenario. Finally, note that the standardized MAEs are by definition larger than the MAEs, as we are dividing the MAEs with the bivariate survival probability (a value between 0 and 1). This phenomenon is especially evident for time points which are further away from the origin and where the survival probability is low.

Web Figure~\ref{fig:MAEs_sing_reg} and Web Figure~\ref{fig:MAEs_sing_reg_bivar_cens} presents the MAEs and standardized MAEs based on the single regression model, for the univariate censoring mechansim and for the bivariate censoring mechanism, respectively. The performance of the single regression model is quite good, and produces even lower errors than the six different regression models, for both types of censoring mechanisms. However, the modified Lin and Ying estimator for the bivariate censoring scenario produced larger errors than the standard Lin and Ying estimator for the univariate censoring scenario, and than the Dabrowska estimator in both censoring scenarios.

\section{}
Here we check the proportional odds assumption for the diabetic retinopathy dataset studied in Section 4 of the main manuscript, by comparing the empirical un-adjusted (log) odds for juvenile versus adult onset diabetes, for all bivariate time points on the diagonal $(t,t)$. For each of these two diabetes groups, the empirical log odds are calculated by $\log(\text{odds})=\log\left(\frac{\hat{S}(t,t)}{1-\hat{S}(t,t)}\right)$, where $\hat{S}$ is the Dabrowska estimator of the (un-adjusted) bivariate survival function $S$. Web Figure~\ref{fig:odds_retinopathy} shows that the vertical displacement between the two curves are reasonably close to being constant for all $t>30$ (where $t$ denotes the number of months since the laser treatment), thus providing some confirmation for the proportional odds model for time points $t$ larger than 30.

Web Table~\ref{tab:retinopathy} presents the estimated regression coefficients for the regression models considered in Section 4 of the main manuscript (one based on a single time point and another based on three time points).
\begin{figure}[p]
	\begin{center}
		\includegraphics[width=0.95\linewidth]{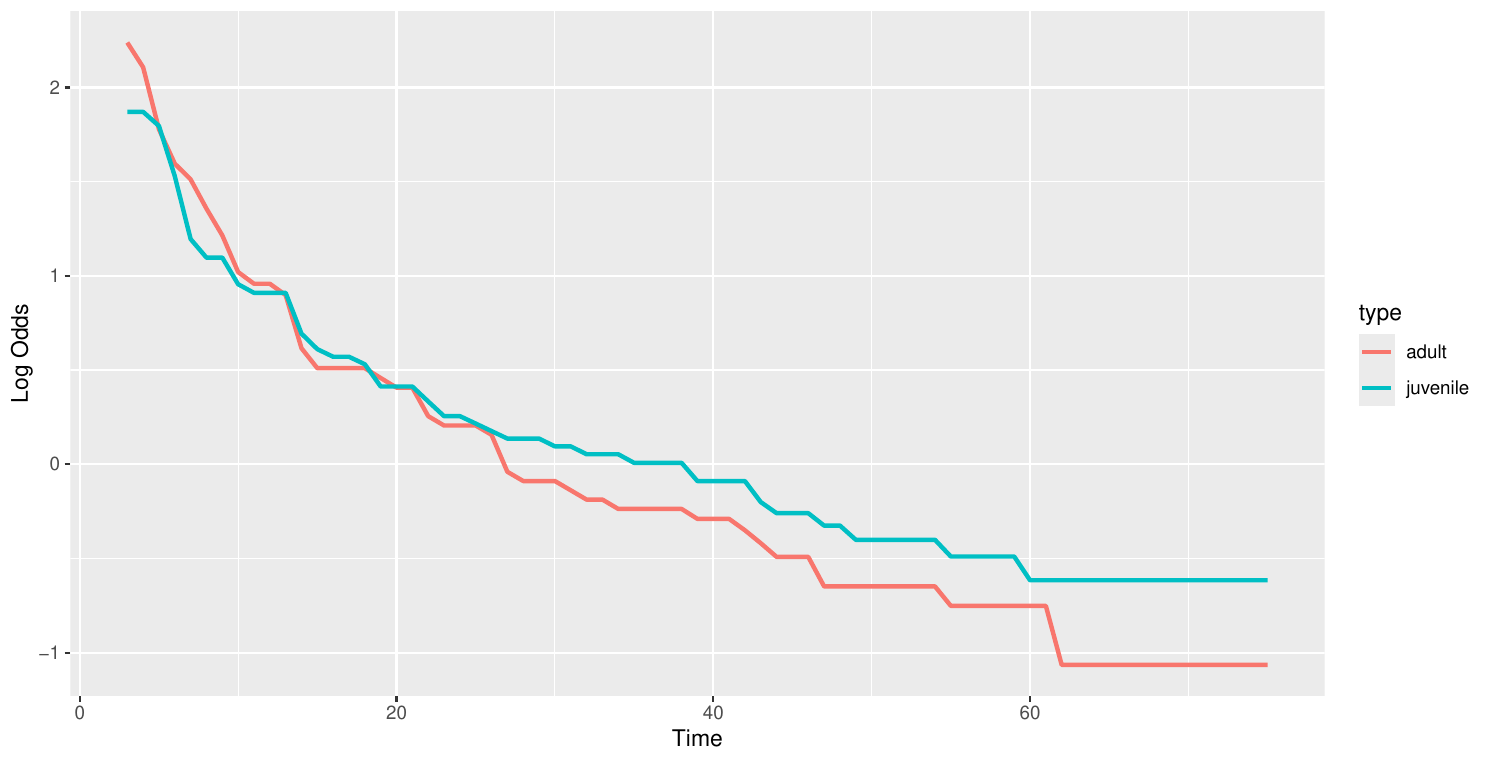}
		\caption{Diabetic retinopathy study. Empirical log odds for juvenile versus adult onset diabetes, for all bivariate time points on the diagonal $(t,t)$.} \label{fig:odds_retinopathy}
	\end{center}
\end{figure}

\begin{table}[ht]
	\centering
	\begin{tabular}{lrrrr}
		\hline
		&  \multicolumn{2}{r}{Single time point} &
		\multicolumn{2}{r}{Three time points} 
		\\
		\hline
		Covariate  &estimate & p-value & estimate & p-value \\ 
		\hline
		Age &  -0.01 & 0.65 &  -0.01 & 0.69\\ 
		Risk score &  -0.19 & 0.21& -0.18 & 0.11\\ 
		Juvenile diabetes  & -0.13 & 0.87 & -0.13 & 0.80 \\ 
		\hline
	\end{tabular}
	\caption{Diabetic retinopathy study. Estimates of the regression parameters for the bivariate pseudo observations approach with the Dabrowska estimator and the logit link function based on either a single time point or three time points.}\label{tab:retinopathy}
\end{table}

\section{}
\cite{ichida_evaluation_1993} conducted a study to evaluate protocol change in burn care management. Specifically, the authors were interested in comparing the effects of two different bathing agents on the risk of wound infection, based on 154 burn victims. The authors also considered time to wound excision (as a time varying covariate), where times are measured in days from hospital admission.  Interestingly, out of the 154 patients, 99 patients underwent wound excision, with the excision occurring between 1 to 49 days after hospital admission. Additionally, 48 patients suffered from a bacterial infection during their hospital stay, with the infection occurring between 1 to 97 days after hospital admission. 
This dataset is available from the R package KMsurv \citep{KMsurv}. \cite{van_der_laan_locally_2002} analyzed this same dataset, and estimated the joint survival of time to wound excision ($T_1$) and time to wound infection ($T_2$), assuming a univariate censoring variable $C$. \cite{van_der_laan_locally_2002} used an inverse probability of censoring weighting (IPCW) estimator and a one-step estimator. In their IPCW estimator, covariates were incorporated into the estimator through the probability of censoring weights, while their one-step estimator depended also on a model for the conditional distribution of the pair of failure times given the covariates. 
In their preliminary analyses, \cite{van_der_laan_locally_2002} found that only treatment type and gender were significantly associated with either $T_1$ or $T_2$. The authors estimated the joint survival probability of $(T_1,T_2)$ at eight different time points.

We use our bivariate pseudo-observations approach with the logit link function to estimate these eight survival probabilities, as a function of treatment and gender, the two binary covariates considered also by \cite{van_der_laan_locally_2002}, and an additional continuous covariate measuring the percentage of total surface area burned. 
We use a single regression model for all time points, with bivariate pseudo-observations based on both the Dabrowska estimator and the Lin and Ying estimator.

Web Figure~\ref{fig:burn_8} presents graphically a boxplot of the covariate adjusted survival probabilities for each of the eight time points and for the pseudo-observations approach based on the Dabrowska estimator. In addition, the figure presents the mean survival probability (a dark red rhombus) and the value of the one-step estimator of \cite{van_der_laan_locally_2002} (a green circle). The mean value is relatively close to the corresponding value of \cite{van_der_laan_locally_2002} for all time points. However, the range of the estimated survival probabilities at certain time points is quite wide, indicating that the joint survival probability is sensitive to specific covariate values.
Web Figure~\ref{fig:burn_8_LY_1_models} presents graphically a boxplot of the covariate adjusted survival probabilities for each of the eight time points and for the pseudo-observations approach based on the Lin and Ying estimator. As can be seen, the pseudo-observations approach that is based on the Lin and Ying estimator seems to produce conditional survival probabilities that are biased downwards.
Finally, note that while \cite{van_der_laan_locally_2002} offer an estimate of the unconditional bivariate survival function $S(t_1,t_2)$, our bivariate pseudo-observations approach allows greater flexibility by considering also covariate effects.

\begin{figure}[p]
	\begin{center}
		\includegraphics[width=0.95\linewidth]{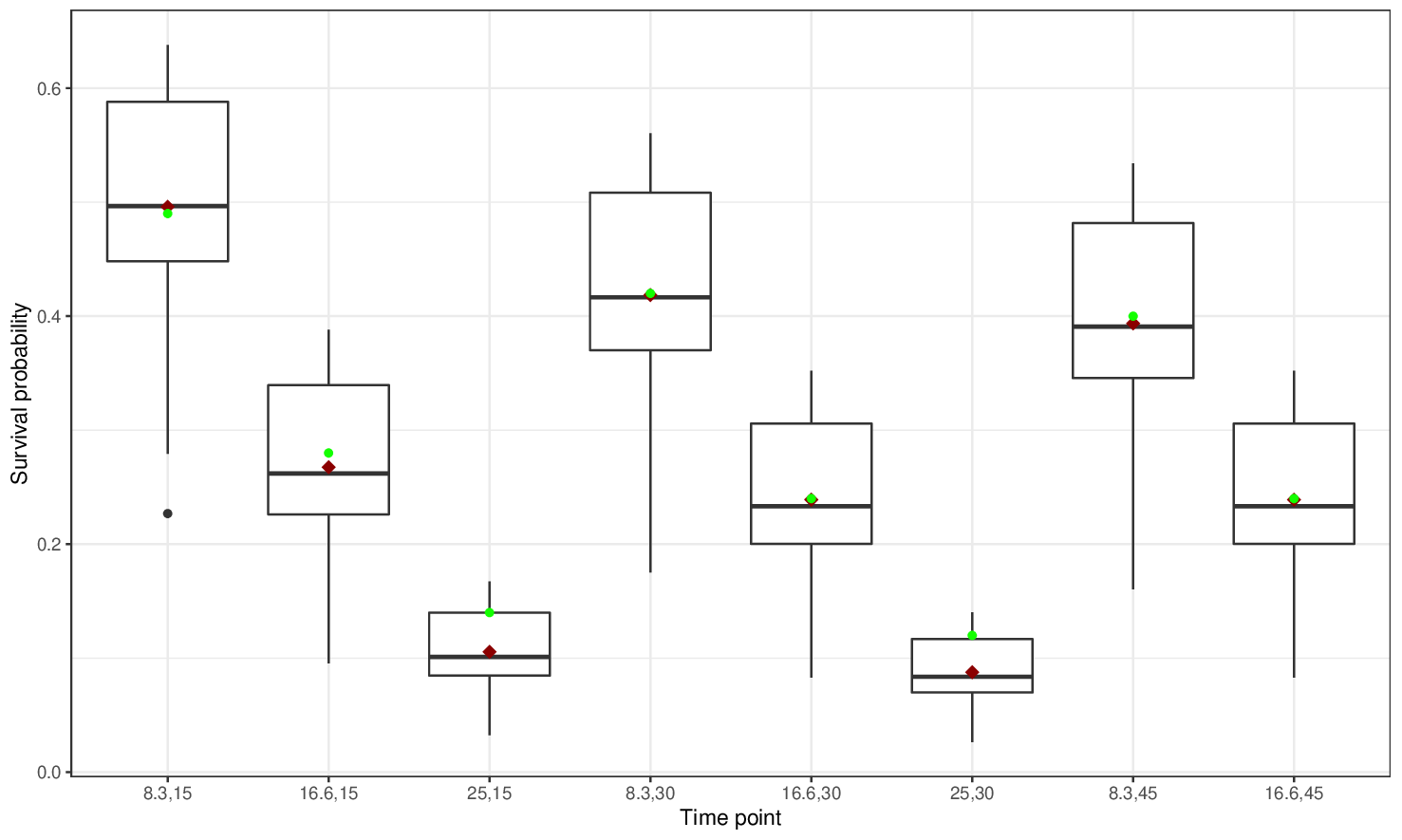}
		\caption{Estimates of the joint survival probabilities of time to wound excision and time to wound infection, for eight time points. The dark red rhombus marks the mean value of the estimated survival probabilities at each time point, and the green circle marks the value of the one-step estimator of \citet{van_der_laan_locally_2002}.} \label{fig:burn_8}
	\end{center}
\end{figure}

Web Table~\ref{table:burn_coef} presents the estimated regression coefficients for the pseudo-observations approach based on both the Dabrowska estimator and the Lin and Ying estimator. Web Table~\ref{table:burn_coef} shows that the estimated values of the estimated regression coefficients based on the two types of estimators can be quite different.
To investigate further, we looked at the distribution of the pseudo-values themselves, based on both estimators, for each of these eight time points. Web Figure~\ref{fig:PO_hist} presents the histograms of the pseudo-values, for each of the eight time points and for both estimators. Web Figure~\ref{fig:PO_hist_t4} presents a zoomed in version at the single time point $(t_1,t_2)=(25,15)$. 

Additionally, we also consider eight separate regression models, one for each time point, based on the pseudo-observations approach with either the Dabrowska estimator, or the Lin and Ying estimator. Web Figure~\ref{fig:burn_8_8_models} presents the boxplots of these survival probabilities for each of the eight time points;  results based on the Dabrowska estimator can be found in Web Figure~\ref{fig:burn_8_dab_8_models}, and results based on the Lin and Ying can be found in Web Figure~\ref{fig:burn_8_LY_8_models}. Compared to Web Figure~\ref{fig:burn_8}, the results based on the Dabrowska estimator with eight separate regression models are generally consistent with the results based on a single regression model for all time points, but exhibit a larger variability. Note that for the single time point $(t_1,t_2)=(25,15)$, the separate regression model did not provide good estimates of the joint survival probability. Similarly, the results based on the Lin and Ying estimator seem to underestimate the true survival probability at all time points, with three time points producing nearly zero probabilities. In conclusion, for the burn victim data, the single regression model for all time points, with the pseudo-observations that are based on the Dabrowska estimator, is preferable.

Finally, we use our suggested method for time points selection, with lower limits corresponding to the univariate 30th percentiles, and upper limits corresponding to the 90th percentiles. Web Table~\ref{table:burn_coef_k} present the slope estimates and their corresponding p-values for different choices of $k$. The resulting regression estimates are relatively close for the different number of time points. However, the regression estimates of the two binary variables (treatment and gender) are different than their corresponding estimates reported in Web Table~\ref{table:burn_coef}. To further investigate, we looked at the odds ratio (OR) of both gender and treatment, for each of the eight time points considered by \cite{van_der_laan_locally_2002}, where the OR of a binary covariate $Z$ is defined by $\text{OR}=\frac{\text{odds}(Z=1)}{\text{odds}(Z=0)}$. Specifically, if the proportional odds assumption is satisfied, we would expect the OR to be relatively constant over the eight time point. However, Web Figure~\ref{fig:burn_PO_assm} shows that the OR of each group is not constant over the eight time points.

In summary, it seems like the regression estimates for the burn victim data are quite sensitive to the location of the time points, which probably means that the assumed logistic regression model is misspecified for this data example. This conclusion is further validated by Web Figure~\ref{fig:burn_PO_assm}, that shows that the proportional odds assumption is not satisfied for this dataset. However, based on our simulations from the log-normal bivariate data (\ref{sec:LN}), our bivariate pseudo-observations approach should still provide good predictions of the conditional joint survival function, even under model misspecification.

\begin{figure}[p]
	\begin{center}
		\includegraphics[width=1\textwidth]{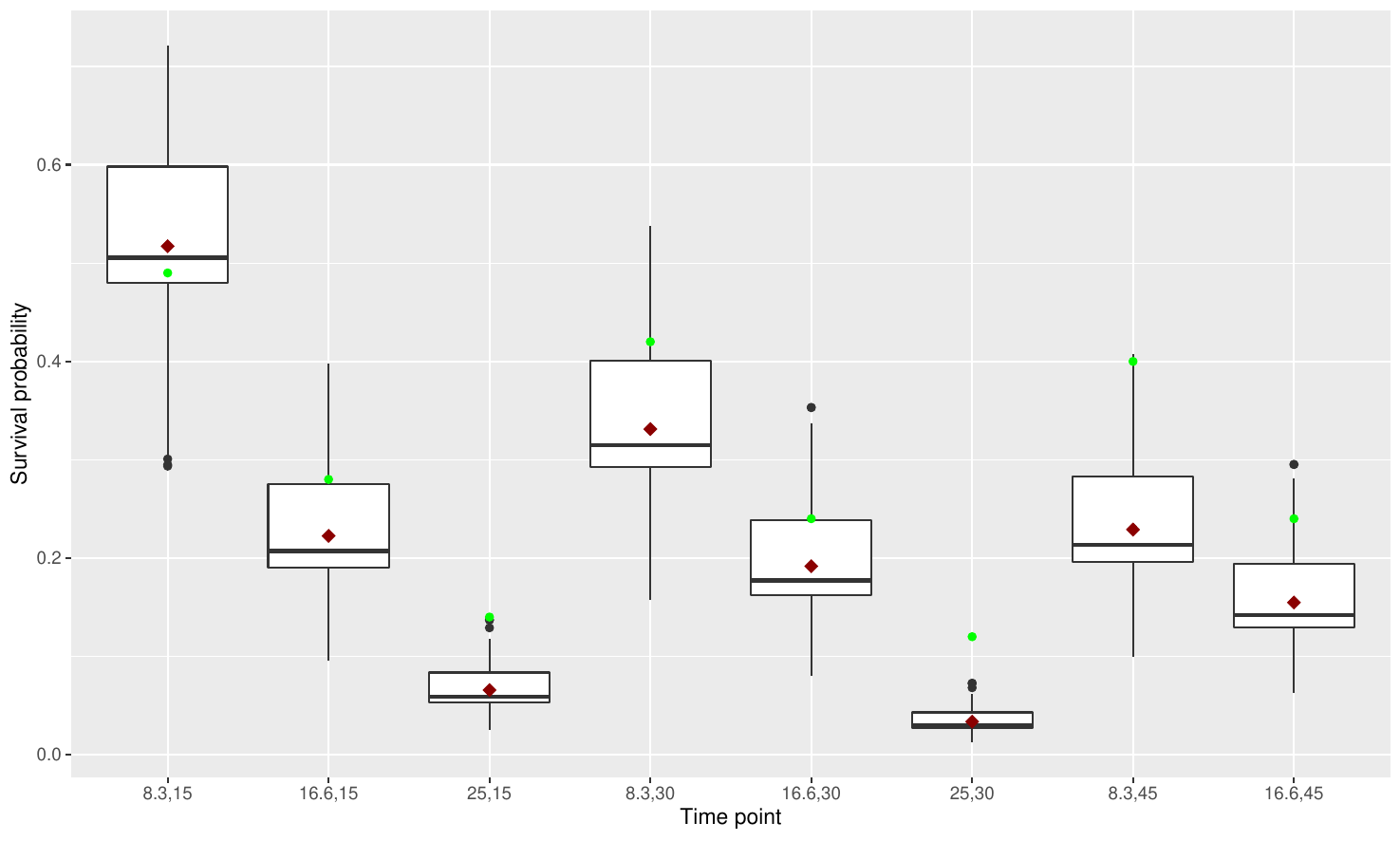}
		\caption{ Estimates of the joint survival probabilities, for eight time points. The estimates are based on the pseudo-observations approach with the Lin and Ying estimator, with a single regression model for all time points. The dark red rhombus marks the mean value of the estimated survival probabilities at each time point, and the green circle marks the value of the one-step estimator of \cite{van_der_laan_locally_2002}. } \label{fig:burn_8_LY_1_models}
	\end{center}
\end{figure}

\begin{table}[h]
	\centering
	\begin{tabular}{llrrrr}
		\hline
		Covariate & Value	& $\beta^{Dab}$ & p-value & $\beta^{LY}$ & p-value \\ 
		\hline
		Intercept &  &  0.58 & 0.107 &  0.29 & 0.515 \\ 
		Time point & 16.6,15 & -1.02 & $<$0.001 & -1.37 & $<$0.001 \\ 
		Time point & 25,15 & -2.17 & $<$0.001 & -2.79 & $<$0.001 \\ 
		Time point & 8.3,30 & -0.32 & 0.007 & -0.80 & $<$0.001 \\ 
		Time point & 16.6,30 & -1.18 & $<$0.001 & -1.56 & $<$0.001 \\ 
		Time point & 25,30 & -2.38 & $<$0.001 & -3.50 & $<$0.001 \\ 
		Time point & 8.3,45  & -0.43 & 0.009 & -1.33 & 0.004 \\ 
		Time point & 16.6,45& -1.18 & $<$0.001 & -1.82 & 0.001 \\ 
		Treatment & Body cleansing & -0.48 & 0.148 & -0.41 & 0.410 \\ 
		Gender & Female & -0.66 & 0.152 & -0.77 & 0.190 \\ 
		Percentage of area burned & Continuous & -0.79 & 0.380 &  0.69 & 0.619 \\ 
		\hline
	\end{tabular}\caption{Estimated regression coefficients and p-values, for the burn data analysis based on both the Dabrowska estimator (Dab) and the Lin and Ying estimator (LY). The covariates include: Treatment: 0-routine bathing 1-Body cleansing, Gender (0=male 1=female), and Percentage of area burned. The p-values are based on the generalized estimating equations using an independence correlation structure.}\label{table:burn_coef}
\end{table}

\begin{table}[h]
	\centering
	\begin{tabular}{llrrrrrr}
		\hline
		Covariate & Value	& $\beta^{Dab}_{k=5}$ & p-value & $\beta^{Dab}_{k=8}$ & p-value &$\beta^{Dab}_{k=10}$ & p-value  \\ 
		\hline
		Treatment & Body cleansing & -0.77 & 0.02 & -0.72 & 0.03 & -0.75 & 0.02 \\ 
		Gender & Female & -0.22 & 0.56 & -0.20 & 0.61 & -0.23 & 0.56 \\ 
		Percentage of area burned & Continuous & -0.72 & 0.43 & -0.59 & 0.52 & -0.70 & 0.45 \\ 
		\hline
	\end{tabular}\caption{Estimated regression coefficients and p-values, for the burn data analysis based on the Dabrowska estimator (Dab), for time points $k \in \{5,8,10\}$ which are based on our time point selection procedure. The covariates include: Treatment: 0-routine bathing 1-Body cleansing, Gender (0=male 1=female), and Percentage of area burned. The p-values are based on the generalized estimating equations using an independence correlation structure.}\label{table:burn_coef_k}
\end{table}

\begin{figure}[p]
	\centering
	\begin{subfigure}[b]{1\textwidth}
		\centering
		\includegraphics[width=1\textwidth]{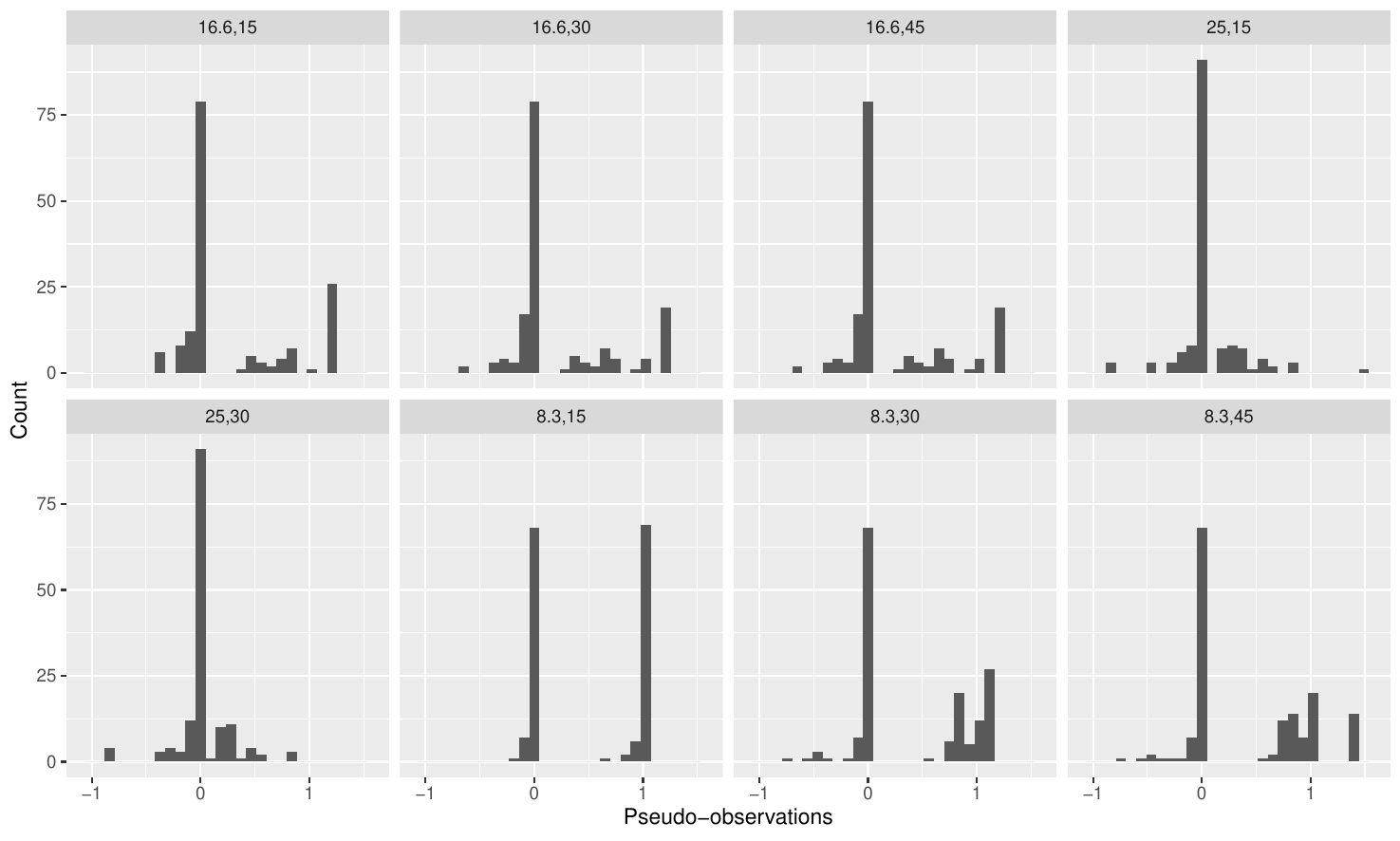}
		\caption{Frequency of pseudo-observation values that are based on the Dabrowska estimator, for each of the eight time point in the burn data.}
	\end{subfigure}
	\hfill
	\begin{subfigure}[b]{1\textwidth}
		\centering
		\includegraphics[width=1\textwidth]{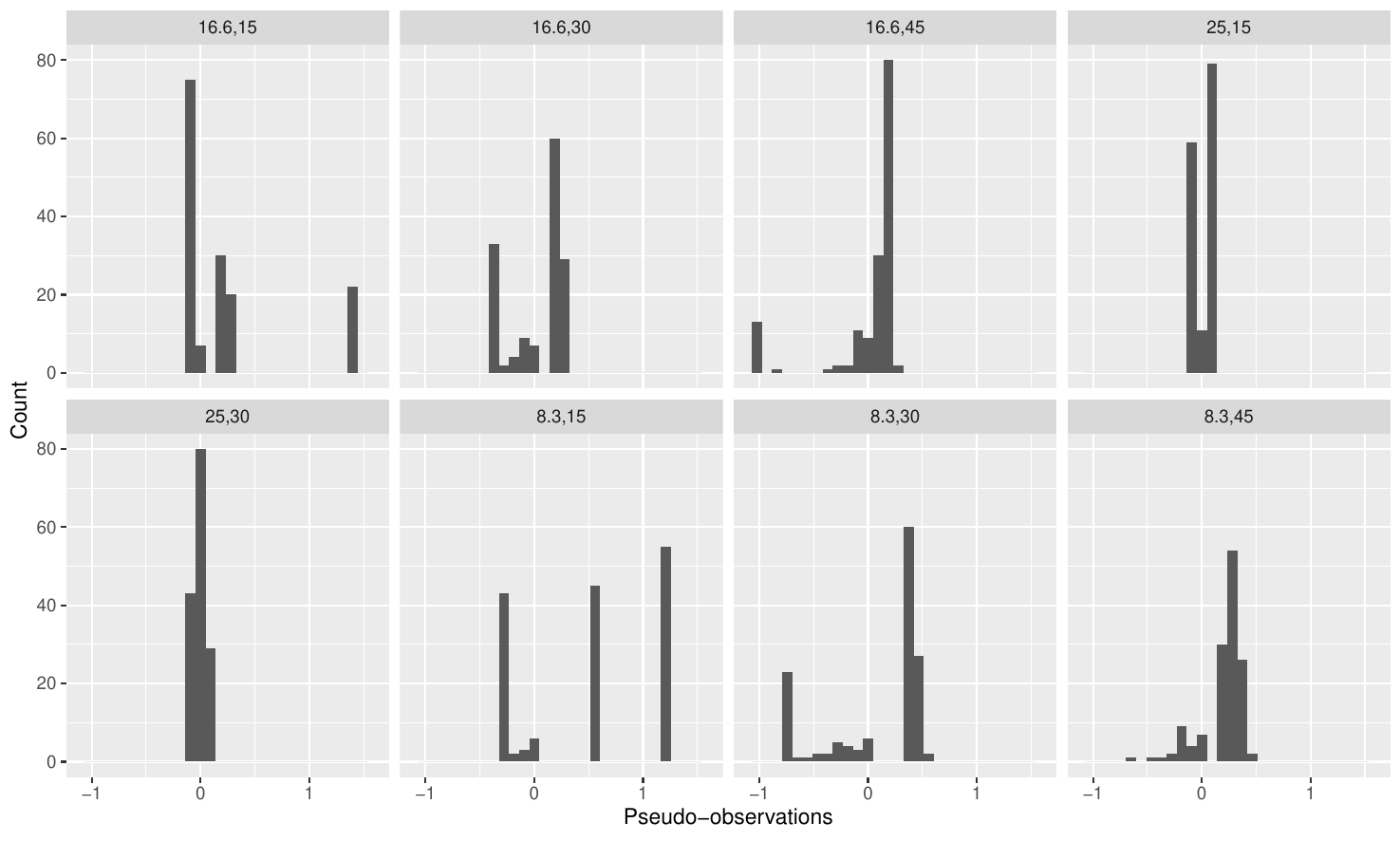}
		\caption{Frequency of pseudo-observation values that are based on the Lin and Ying estimator, for each of the eight time point in the burn data.} 
	\end{subfigure}
	\caption{Frequency of pseudo-observation values, for each of the eight time point in the burn data. The top panel corresponds to pseudo-observations that are based on the Dabrowska estimator, and the bottom panel corresponds to pseudo-observations that are based on the Lin and Ying estimator.}
	\label{fig:PO_hist}
\end{figure}

\begin{figure}[p]
	\centering
	\begin{subfigure}[b]{1\textwidth}
		\centering
		\includegraphics[width=1\textwidth]{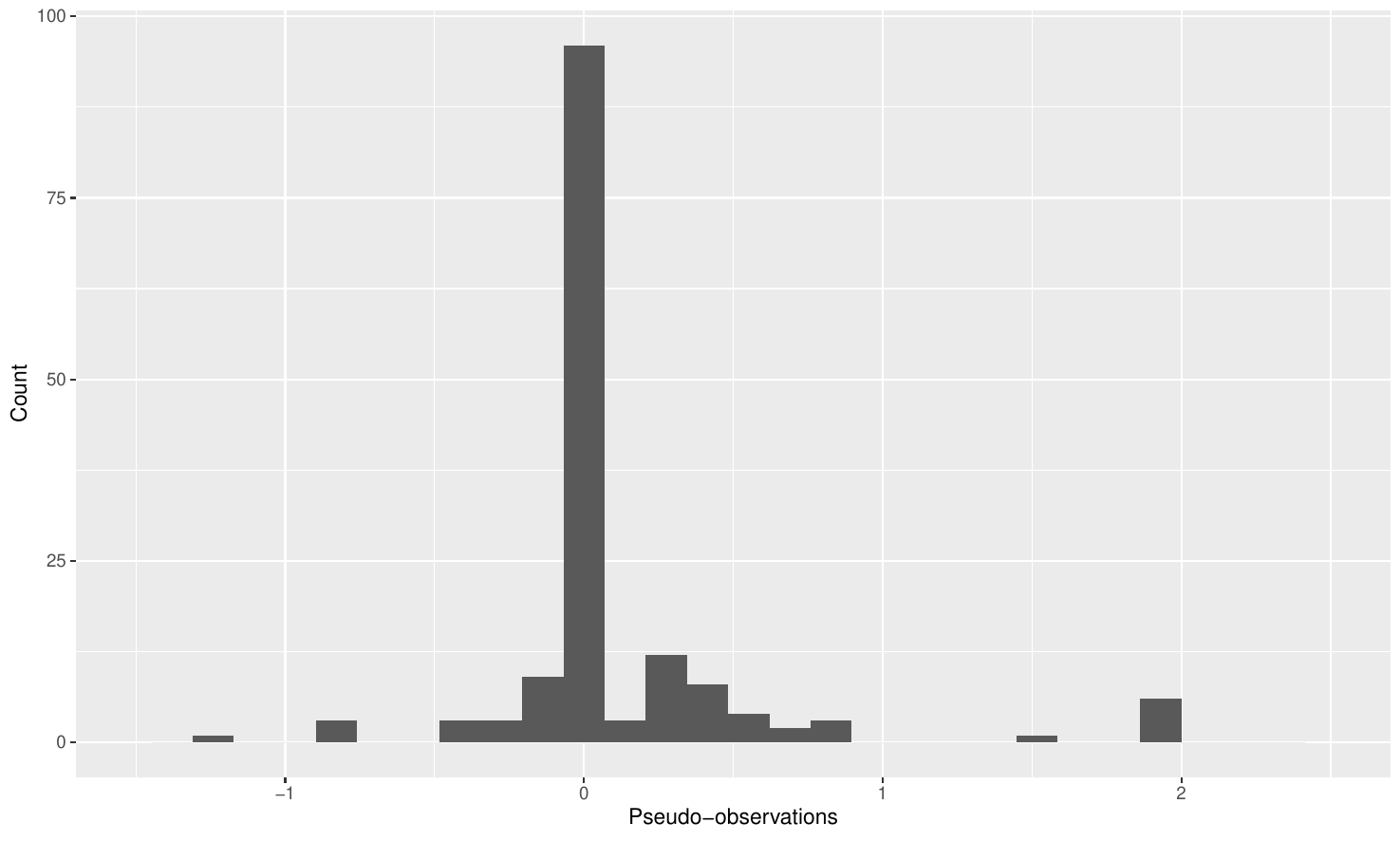}
		\caption{Frequency of pseudo-observation values that are based on the Dabrowska estimator, for the single time point $(t_1,t_2)=(25,15)$.}
	\end{subfigure}
	\hfill
	\begin{subfigure}[b]{1\textwidth}
		\centering
		\includegraphics[width=1\textwidth]{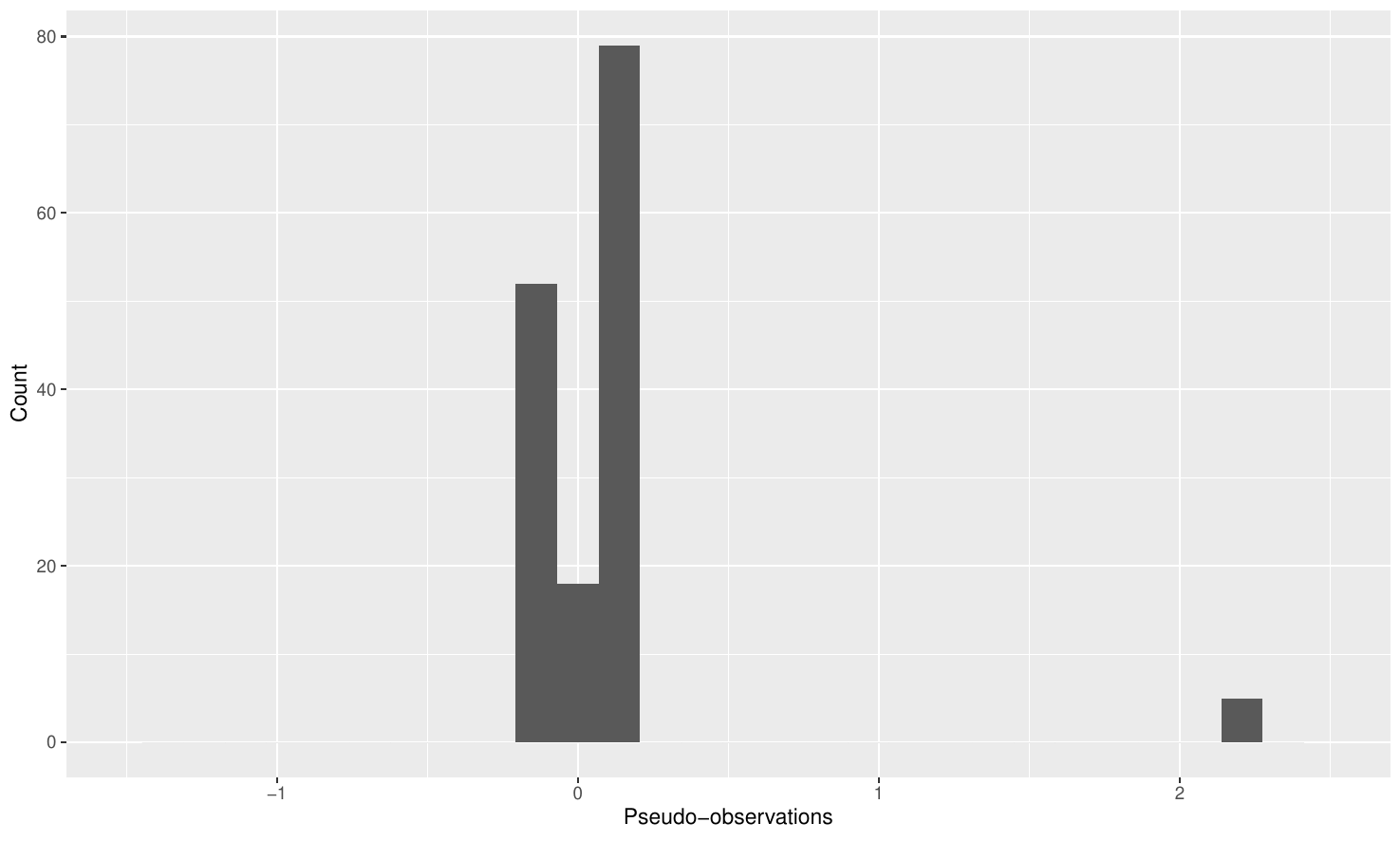}
		\caption{Frequency of pseudo-observation values that are based on the Lin and Ying estimator,  for the single time point $(t_1,t_2)=(25,15)$.} 
	\end{subfigure}
	\caption{Frequency of pseudo-observation values, for the single time point $(t_1,t_2)=(25,15)$. The top panel corresponds to pseudo-observations that are based on the Dabrowska estimator, and the bottom panel corresponds to pseudo-observations that are based on the Lin and Ying estimator.}
	\label{fig:PO_hist_t4}
\end{figure}

\begin{figure}[p]
	\centering
	\begin{subfigure}[b]{1\textwidth}
		\centering
		\includegraphics[width=1\textwidth]{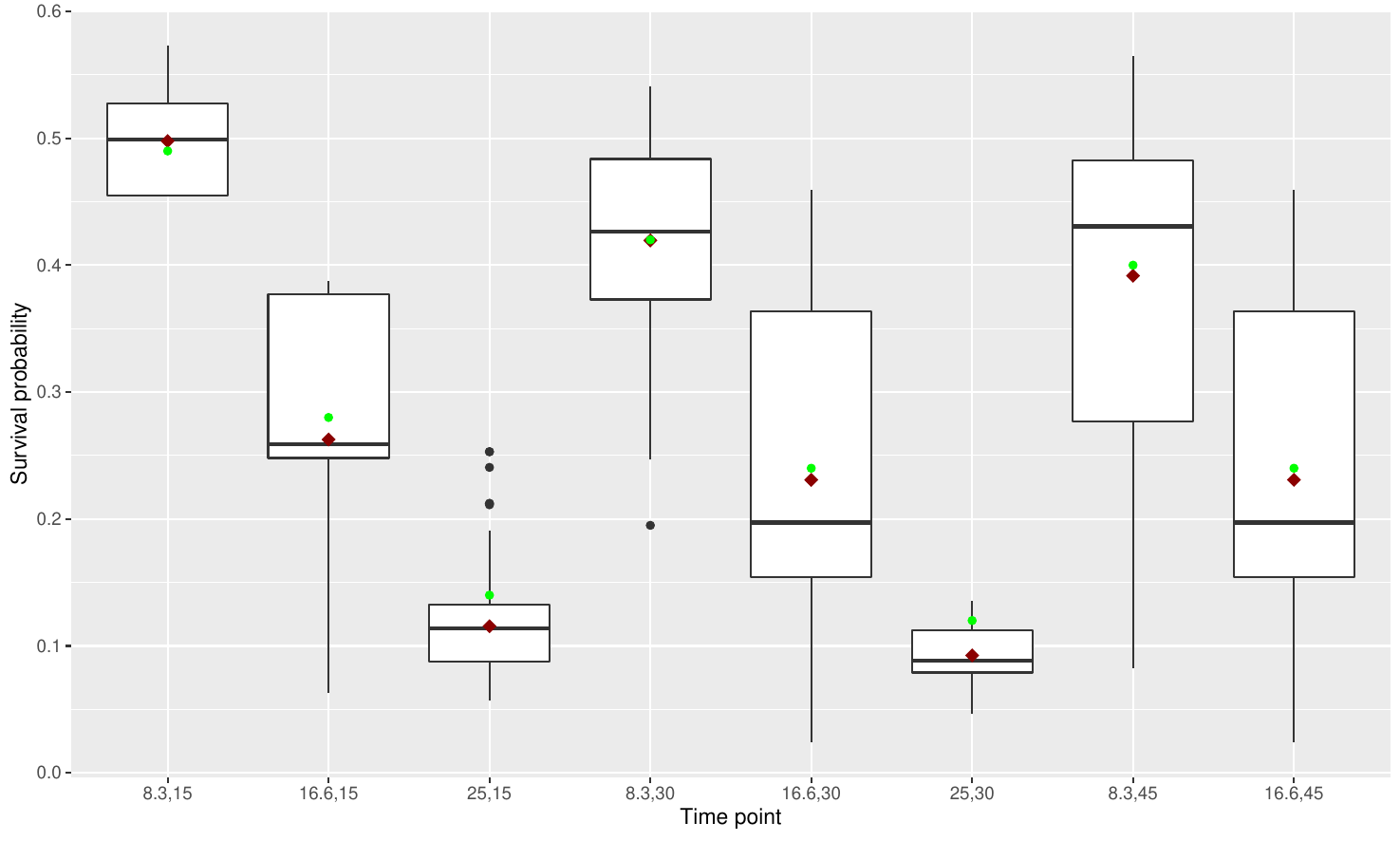}
		\caption{Joint survival probabilities based on the pseudo-observations approach with the Dabrowska estimator. A different regression model for each time point.} \label{fig:burn_8_dab_8_models}
	\end{subfigure}
	\hfill
	\begin{subfigure}[b]{1\textwidth}
		\centering
		\includegraphics[width=1\textwidth]{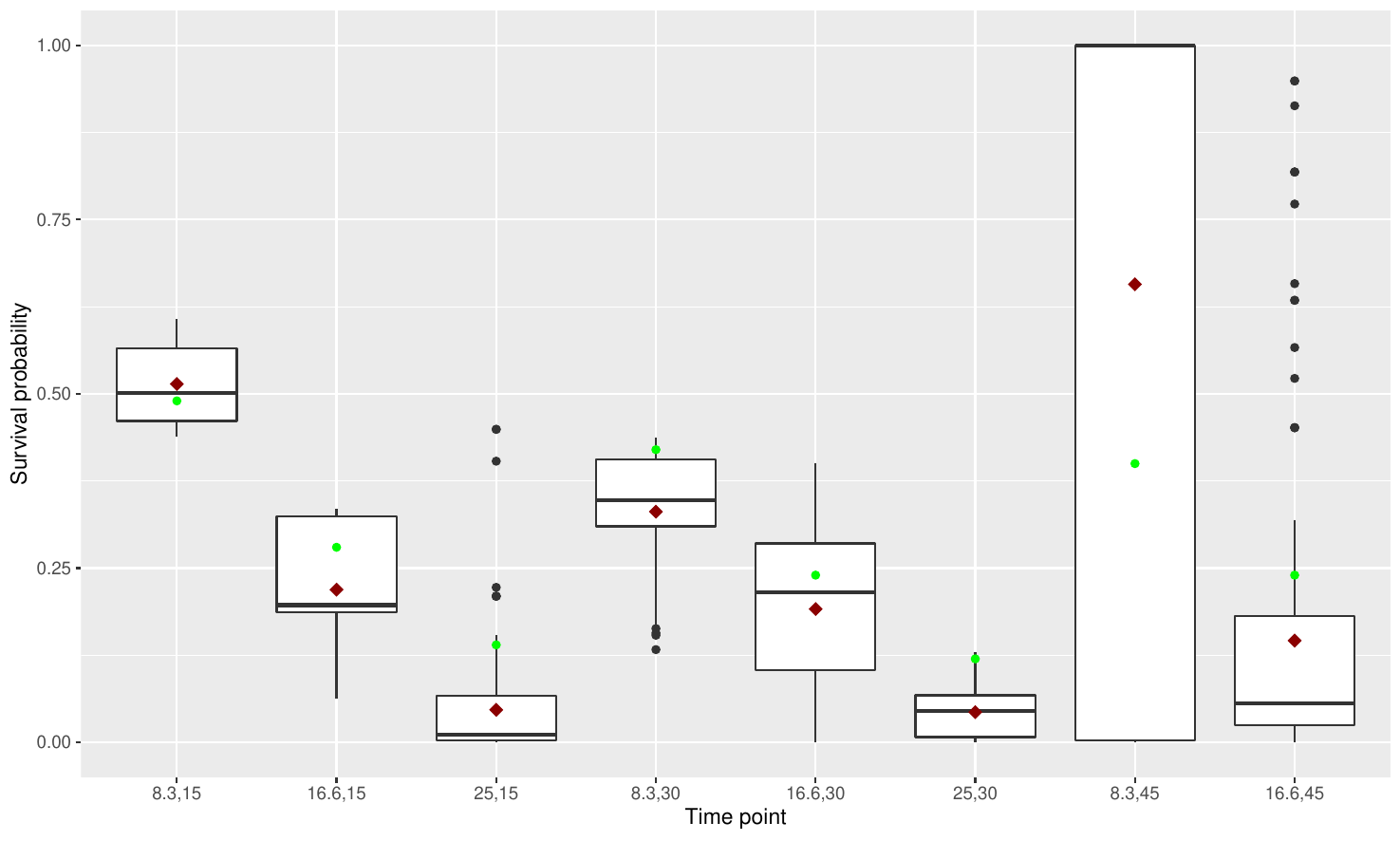}
		\caption{Joint survival probabilities based on the pseudo-observations approach with the Lin and Ying estimator. A different regression model for each time point.} \label{fig:burn_8_LY_8_models} 
	\end{subfigure}
	\caption{Top: Dabrowska, bottom: Lin and Ying. Estimates of the joint survival probabilities, for eight time points, where each time point has its own regression model. The dark red rhombus marks the mean value of the estimated survival probabilities at each time point, and the green circle marks the value of the one-step estimator of \cite{van_der_laan_locally_2002}. } \label{fig:burn_8_8_models} 
\end{figure}

\begin{figure}[p]
	\begin{center}
		\includegraphics[width=1\textwidth]{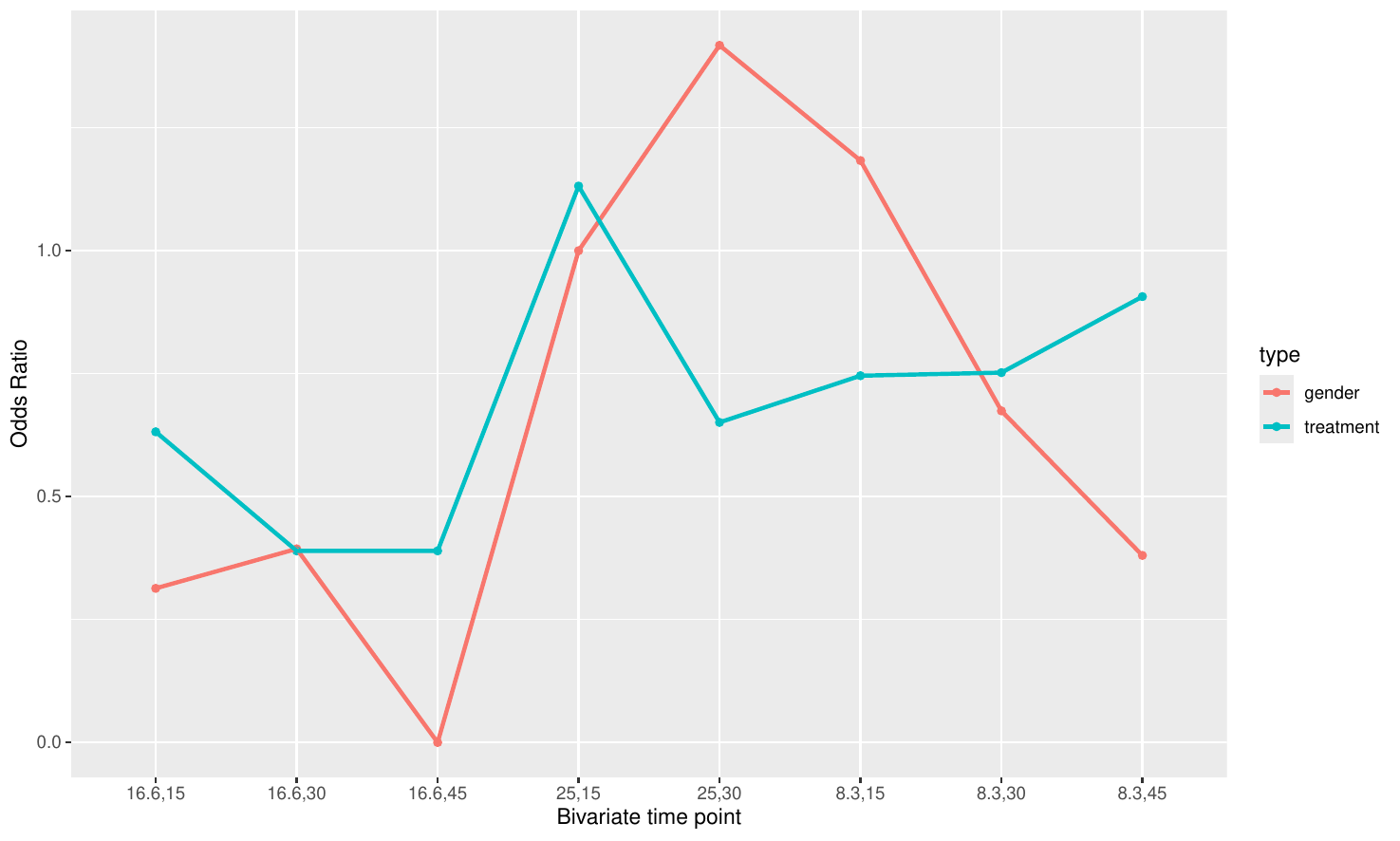}
		\caption{Burn victim data: empirical odds ratio (OR), for the two binary covariates treatment and gender, evaluated at the eight time points of \cite{van_der_laan_locally_2002}. } \label{fig:burn_PO_assm}
	\end{center}
\end{figure}

\section{}
Here we show how confidence intervals for conditional survival probabilities, such as $S_{T_1\mid T_2 \leq 36}(60\mid Z)$, can be calculated using the delta method, as described below.
Assume that we are interested in the covariate-adjusted bivariate survival probability at $k$ bivariate time points. Using the same notation as in our manuscript (Section 2.4), denote $\beta=(\beta_{01},\ldots,\beta_{0k},\beta_1,\ldots,\beta_p)^T$, and $Z_{ij}^*$ is $(0,...,1,...0,Z_i^T)^T$ where the $1$ is in the $j$'th place. By Theorem 1, for the true regression parameter $\beta^* \in \mathbb{R}^{p+k}$,
\[
\sqrt{n}(\hat{\beta}-\beta^*) \xrightarrow{d} N_{p+k}(0,M^{-1}\Sigma M^{-1})
\]
(note that we are using here $k$ time points instead of single time point and consequently the covariance matrix is now of dimension $p+k$).
Consequently, for $\mu=(\mu_1,\ldots,\mu_k)^T$, and $T=(T_1,\ldots,T_k)^T$, where $\mu_j={\beta^*}^TZ_{ij}^*$ and $T_j={\hat{\beta}}^TZ_{ij}^*$, we have that 
\[
\sqrt{n}\left(\begin{pmatrix}
	T_1\\ \vdots \\ T_k
\end{pmatrix}
-\begin{pmatrix}
	\mu_1\\ \vdots \\ \mu_k
\end{pmatrix}
\right) 
\xrightarrow{d}
N_{k}(0,(Z_i^*)^TM^{-1}\Sigma M^{-1}Z_i^*),
\]
where $Z_i^*=\begin{pmatrix}
	\mid & \mid & & \mid \\
	Z_{i1}^* & Z_{i2}^* & \cdots & Z_{ik}^*\\
	\mid & \mid & & \mid \\
\end{pmatrix}$ is a matrix of dimension $(p+k)\times k$. For simplicity, denote by
\begin{equation}\label{eq:vcov_i}
	\tilde{\Sigma}=(Z_i^*)^TM^{-1}\Sigma M^{-1}Z_i^*)
\end{equation}
the $k$-dimensional covariance matrix for the specific covariate of the $i$th observation.

Next, we define a function $\phi:\mathbb{R}^k \mapsto \mathbb{R}^k$ by $\phi(\mu)=(g^{-1}(\mu_1),\ldots,g^{-1}(\mu_p))^T$, where $g^{-1}(x)=[1+\exp(-x)]^{-1}$ is the inverse of the logit link. Using the delta method, it can be shown that 
\[
\sqrt{n}(\phi(T)-\phi(\mu))=\sqrt{n}\left(\begin{pmatrix}
	g^{-1}(T_1)\\ \vdots \\ g^{-1}(T_k)
\end{pmatrix}
-\begin{pmatrix}
	g^{-1}(\mu_1)\\ \vdots \\ g^{-1}(\mu_k)
\end{pmatrix}
\right) 
\xrightarrow{d}
N_{k}(0,\phi_{\mu}'\tilde{\Sigma}\phi_{\mu}'),
\]
where $\tilde{\Sigma}$ is defined in Equation~\eqref{eq:vcov_i}, and where $\phi_{\mu}'=\{\frac{\partial g^{-1}(\mu_i)}{\partial \mu_j}\}_{i,j=1}^k$ is a square matrix of size $k\times k$ consisting of the partial derivatives. That is,
\[
\phi_{\mu}'=diag\left(\frac{\exp(-\mu_1)}{[1+\exp(-\mu_1)]^2},\ldots, \frac{\exp(-\mu_k)}{[1+\exp(-\mu_k)]^2}\right).
\]
Next, note that $S(t_1^j,t_2^j\mid Z_i)=g^{-1}(\mu_j)=g^{-1}({\beta^*}^TZ_{ij}^*)$ and $\hat{S}(t_1^j,t_2^j\mid Z_i)=g^{-1}(T_j)=g^{-1}({\hat{\beta}}^TZ_{ij}^*)$. Consequently,
\[
\sqrt{n}\left(\begin{pmatrix}
	g^{-1}(T_1)\\ \vdots \\ g^{-1}(T_k)
\end{pmatrix}
-\begin{pmatrix}
	g^{-1}(\mu_1)\\ \vdots \\ g^{-1}(\mu_k)
\end{pmatrix}
\right)=
\sqrt{n}\left(\begin{pmatrix}
	\hat{S}(t_1^1,t_2^1\mid Z_i)\\ \vdots \\ \hat{S}(t_1^k,t_2^k\mid Z_i)
\end{pmatrix}
-\begin{pmatrix}
	S(t_1^1,t_2^1\mid Z_i)\\ \vdots \\ S(t_1^k,t_2^k\mid Z_i)
\end{pmatrix}
\right)  
\xrightarrow{d}
N_{k}(0,\phi_{\mu}'\tilde{\Sigma}\phi_{\mu}').
\]
In our data analysis example, we consider the three timepoints $\{(60,36), (60,0), (0,36)\}$. We define a function $h:\mathbb{R}^3\mapsto \mathbb{R}$ by $h(x,y,z)=\frac{y-x}{1-z}$ such that 
\[
h\left(S(60,36\mid Z_i),S(60,0\mid Z_i),S(0,36\mid Z_i)\right)=\frac{S(60,0 \mid Z)-S(60,36 \mid Z)}{1-S(0,36 \mid Z)}=S_{T_1| T_2\leq 36}(60\mid Z).
\]
Note that the gradient of $h$ at some point $a=(x,y,z)^T\in \mathbb{R}^3$ is $\nabla h(a)=\left(\frac{\partial h(a)}{\partial x},\frac{\partial h(a)}{\partial y},\frac{\partial h(a)}{\partial z}\right)^T=\left(-\frac{1}{1-z},\frac{1}{1-z},\frac{x-y}{(1-z)^2}\right)^T$.
Consequently, for $s=\left(S(60,36\mid Z_i),S(60,0\mid Z_i),S(0,36\mid Z_i)\right)^T$ we have that $$\nabla h(s)=\left(-\frac{1}{1-S(0,36\mid Z_i)},\frac{1}{1-S(0,36\mid Z_i)},\frac{S(60,36\mid Z_i)-S(60,0\mid Z_i)}{[1-S(0,36\mid Z_i)]^2}\right)^T,$$
and 
\[
\sqrt{n}\left(h(\hat{s})-h(s)\right)\xrightarrow{d}
N_{1}(0,\nabla h(s)^T\phi_{\mu}'\tilde{\Sigma}\phi_{\mu}'\nabla h(s)),
\]
where $\hat{s}=\left(\hat{S}(60,36\mid Z_i),\hat{S}(60,0\mid Z_i),\hat{S}(0,36\mid Z_i)\right)^T$. Finally, a 95\% confidence interval for $S_{T_1| T_2\leq 36}(60\mid Z)=h(s)$ is given by 
\[
\text{CI}=(h(\hat{s})-1.96 \cdot \sigma,h(\hat{s})+1.96 \cdot \sigma)=(\hat{S}_{T_1| T_2\leq 36}(60\mid Z)-1.96 \cdot \sigma,\hat{S}_{T_1| T_2\leq 36}(60\mid Z)+1.96 \cdot \sigma),
\]
where $\sigma=\frac{\sqrt{\nabla h(s)^T\phi_{\mu}'\tilde{\Sigma}\phi_{\mu}'\nabla h(s)}}{\sqrt{n}}$. 

\end{document}